\documentclass[noinfoline]{imsart}

\RequirePackage[OT1]{fontenc}
\RequirePackage{amsthm,amsmath}
\RequirePackage[numbers]{natbib}
\RequirePackage[colorlinks,citecolor=blue,urlcolor=blue]{hyperref}

\usepackage{verbatim}
\usepackage{xcolor}

\allowdisplaybreaks

\usepackage{fullpage}
\usepackage{enumitem}
\usepackage{algorithm}
\usepackage{mathtools}
\usepackage{booktabs}

\newcommand {\R}{\mathbb{R}}

\renewcommand {\O}{\mathcal{O}}
\newcommand {\Prob}{\mathbb{P}}

\newcommand{\E}{\mathbb{E}}

\newcommand{\verti}[1]{{\left\vert\kern-0.25ex\left\vert\kern-0.25ex\left\vert #1 
    \right\vert\kern-0.25ex\right\vert\kern-0.25ex\right\vert}}

\newcommand{\argmin}{\operatornamewithlimits{arg\,min}}

\DeclareMathOperator{\diag}{diag}

\newcommand{\ip}[1]{{\left\langle\kern-0.25ex #1 \kern-0.25ex\right\rangle}}

\newcommand {\bigop}[1]{ \O_{\Prob}\left( #1 \right) }
\newcommand {\smallop}[1]{ o_{\Prob}\left( #1 \right) }
\newcommand {\Ez}[1]{ \E \left[ #1 \right] }

\newcommand {\Ltwosq}{ \mathcal{L}^2([0,1]^2) }
\DeclareMathOperator{\cov}{Cov}
\DeclareMathOperator{\var}{Var}
\DeclareMathOperator{\corr}{Corr}

\definecolor{tomas}{rgb}{0.8, 0.5, 0.2}
\definecolor{soham}{RGB}{0,167,159}
\definecolor{victor}{rgb}{0.8,0, 0}
\definecolor{rubin}{rgb}{0.5,0.3, 0.8}
\definecolor{lavender}{rgb}{0.45, 0.31, 0.59}

\usepackage{subfigure}
\usepackage{layout}

\usepackage{dsfont}
\usepackage[mathscr]{eucal}
\usepackage[toc,page]{appendix}
\usepackage{mathrsfs}
\usepackage{color}
\usepackage{pifont}
\usepackage{bm}
\usepackage{latexsym}
\usepackage{amsfonts}
\usepackage{amssymb}
\usepackage{epsfig}
\usepackage{graphicx}
\usepackage{multirow}
\usepackage{caption}
\newtheorem{theorem}{Theorem}
\newtheorem{proposition}{Proposition}
\newtheorem{corollary}{Corollary}
\newtheorem{lemma}{Lemma}

\newtheorem{remark}{Remark}
\newtheorem{example}{Example}

\newcommand{\transpose}{^{\top}}
\startlocaldefs
\numberwithin{equation}{section}
\theoremstyle{plain}
\endlocaldefs

\begin{document}

\begin{frontmatter}

\vspace*{-0.15cm}
\title{\LARGE Inference and Computation for Sparsely Sampled Random Surfaces}

% :-) 

\runtitle{Sparsely Observed Random Surfaces}

\begin{aug}
\author{\fnms{Tomas} \snm{Masak}},
\author{\fnms{Tomas} \snm{Rubin}}
\and
\author{\fnms{Victor M.} \snm{Panaretos*}}
\thankstext{t1}{Research supported by a Swiss National Science Foundation grant.}

\runauthor{T. Masak, T. Rubin \& V.M. Panaretos}

\affiliation{Ecole Polytechnique F\'ed\'erale de Lausanne}

\address{Institut de Math\'ematiques\\
Ecole Polytechnique F\'ed\'erale de Lausanne\\
e-mail: 
\href{mailto:tomas.masak@epfl.ch}{tomas.masak@epfl.ch}, 
\href{mailto:tomas.rubin@epfl.ch}{tomas.rubin@epfl.ch}, 
\href{mailto:victor.panaretos@epfl.ch}{victor.panaretos@epfl.ch}}

\end{aug}

\begin{abstract}
Non-parametric inference for functional data over two-dimensional domains entails additional computational and statistical challenges, compared to the one-dimensional case. Separability of the covariance is commonly assumed to address these issues in the densely observed regime. Instead, we consider the sparse regime, where the latent surfaces are observed only at few irregular locations with additive measurement error, and propose an estimator of covariance based on local linear smoothers. Consequently, the assumption of separability reduces the intrinsically four-dimensional smoothing problem into several two-dimensional smoothers and allows the proposed estimator to retain the classical minimax-optimal convergence rate for two-dimensional smoothers. Even when separability fails to hold, imposing it can be still advantageous as a form of regularization. A simulation study reveals a favorable bias-variance trade-off and massive speed-ups achieved by our approach. Finally, the proposed methodology is used for qualitative analysis of implied volatility surfaces corresponding to call options, and for prediction of the latent surfaces based on information from the entire data set, allowing for uncertainty quantification. Our cross-validated out-of-sample quantitative results show that the proposed methodology outperforms the common approach of pre-smoothing every implied volatility surface separately.

\end{abstract}



\end{frontmatter}

\vspace{-0.5cm}
\tableofcontents

\newpage
\section{Introduction}

The term \emph{random surfaces} refers to continuous data on a two-dimensional domain.
Such data sets consist of multiple independent replications of some underlying process $X = (X(t,s), t \in \mathcal{T}, s \in \mathcal{S})$, forming a random sample $X_1, \ldots, X_N$, and can be tackled in the context of functional data analysis (FDA), see \cite{ramsay2005,ramsay2007applied}.
In practice, surface-valued data cannot be observed as infinite-dimensional objects, one instead observes only a finite number of measurements per surface. The classical split in literature (see \cite{zhang2016} for a comprehensive overview) consists of two measurement regimes: the \textit{dense sampling} and the \textit{sparse sampling}. The former considers the setting where the surface data are recorder densely enough (usually on a grid) so they can be worked with as if they were truly infinite dimensional, possibly after a pre-smoothing step, while retaining the same asymptotic properties as truly infinite dimensional data \cite{hall2006properties}.
The sparse regime, however, requires a different approach. Here, the surfaces are observed only on a small (varying) number of irregular locations, which also vary across the surfaces and are corrupted with measurement errors. If such observations are gridded, which is common for computational reasons, the resulting matrix-valued samples contain many missing entries. Densely observed random surfaces arise frequently in medical imaging, linguistics \cite{pigoli2018}, or climate studies \cite{gneiting2006}.
Examples of sparsely observed random surfaces include longitudinal studies (where only a part of a functional profile is measured at each visit) \cite{lopez2020}, geolocalized data \cite{yarger2020,zhang2020, wang2020}, or financial data \cite{fengler2009arbitrage, kearney2018}.

Regardless of the sampling regime, any functional data analysis will begin by estimating two fundamental objects: the mean function $\mu(t,s)=\E X(t,s)$ and the covariance kernel $c(t,s,t',s')=\cov(X(t,s),X(t',s'))$. These ought to be estimated non-parametrically, since availability of replicated observations should allow so. While the mean has the same dimension as the underlying process itself, the covariance is an object of a higher dimension, and its estimation thus suffers from the curse of dimensionality. 
The size of a general covariance often poses both computational and statistical issues, e.g. restricting the resolution of the grid one can handle \cite{aston2017,masak2020}.

Separability of the covariance is a popular non-parametric assumption to reduce the statistical and computational burden, which arises when working on a multi-dimensional domain. Well-known and often coupled together with parametric models in the field of spatial statistics (see e.g. \cite{deIaco2002}, and references therein), separability has lately attracted a lot of attention also from the non-parametric viewpoint of FDA \cite{bagchi2017,constantinou2017,masak2019,dette2020}. Intuitively, it allows one to decouple the dimensions, regarding the random surfaces as an enlarged ensemble of random curves instead. Specifically, separability assumes that the full covariance can be written as a product of covariance kernels corresponding to individual dimensions, i.e.:
\begin{equation*}
    c(t,s,t',s') = a(t,t') b(s,s')
\end{equation*}
for some bivariate kernels $a=a(t,t')$ and $b=b(s,s')$. Thus separability decomposes the four-dimensional covariance kernel into two kernels, each of the same dimensionality as the mean function.

Though separability cannot always be justified in the context of applied problems \cite{rougier2017},  it is nevertheless frequently imposed because of the computational advantages it entails \cite{gneiting2006,genton2007,pigoli2018}. And, while wrongfully assuming separability introduces a bias, it may still lead to improved estimation and prediction due to an implicit bias-variance trade-off. What is certain is that separability greatly reduces computational costs of both the estimation and prediction tasks.

To the best of our knowledge, the aforementioned properties of separability have been mostly explored in the case of densely observed data. In the sparse regime, there is no procedure for the non-parametric estimation of a separable covariance. Methodology designed for the densely observed surfaces does not apply, and methods designed for sparsely observed curves cannot feasibly be extended to surfaces, as we will later argue. Since sparse data are generally associated with higher computational complexity (extra costs associated with smoothing) as well as higher statistical complexity (variance is inflated due to sparse measurements burdened by noise), it appears that the gains stemming from separability could be even larger in the sparse regime. 

The goal of this paper is to leverage the separability assumption to reduce complexity of covariance estimation down to that of mean estimation when working under the sparse regime. The method of choice for sparsely observed functional data on one-dimensional domains is the PACE approach \cite{yao2005}, which is based on kernel regression smoothers. A naive generalization of PACE to a two-dimensional domain would entail a computationally infeasible local linear smoothing step in four-dimensional space. Instead, we demonstrate how to make careful use of separability to collapse the four-dimensional smoother into several two-dimensional surface smoothers. Consequently, the estimation of the mean, covariance, and noise level all become of similar computational complexity. Our asymptotic theory shows that the estimators match the minimax optimal convergence rates for two dimensional smoothing problems. Moreover, our simulation study demonstrates the enormous computational gains offered by separability, the statistical gains when the ground-truth is separable, and a favorable bias-variance trade-off when the truth is in fact not separable.
%Our approach is inspired by the PACE methodology \cite{yao2005}, which is based on the kernel regression smoother and is the prevalent approach to sparse data on a univariate domain. While a naive generalization of PACE to a multi-dimensional domain, say of dimension two, would require a computationally infeasible local linear smoothing in the four-dimensional space, we utilize separability to reduce this step into calculation of several two-dimensional surface smoothers. Therefore the mean, the covariance, and the noise level estimation tasks are all of similar computational complexity. Moreover, since all the estimators are based on two-dimensional smoothing, their asymptotic convergence rate matches the known minimax optimal convergence rate for two dimensional smoothing problems, as our asymptotic theory shows. Our simulation study demonstrates the enormous computational gains offered by separability, the statistical gains when the ground-truth is separable, and a favorable bias-variance trade-off when the truth is in fact not separable.

Finally, we illustrate the benefits of our approach on a qualitative analysis of implied volatility surfaces corresponding to call options. Here, one surface corresponds to a fixed asset (e.g. a stock), for which the right to buy it for an agreed-upon \emph{strike} price (first dimension) at a future \emph{time to expiration} (second dimension) is traded. The value of implied volatility at given strike and time to expiration is derived directly from observed market data (the option prices), by the well-known and commonly used Black-Scholes formula \cite{black1973pricing,merton1973theory}. The implied volatilities are preferred over the option prices, since they are dimensionless quantities, allow for a direct comparison of different assets, and are well familiar to practitioners. Since the options are traded only for a finite number of strikes and times to expiration, which vary across different surfaces, the observed data consist of a sparse ensemble.
The interpolation of such an ensemble is a typical objective in financial mathematics as the latent implied volatility surfaces are of interest for other tasks such as prediction (forecasting) of option prices \cite{hull2006options}.
The common practice is to interpolate or smooth the measurements for every surface independently. For example, \citet{cont2002} utilized pre-smoothing by local polynomial regression, evaluating an individual smoother for every single surface. This approach may, however, pose issues when the available sparse measurements for a given surface are concentrated only on a subset of the domain, which is often the case. Once the predicted surfaces are fed into a subsequent predictive models, the naively extrapolated parts of the surface are given the same weight as the more reliable interpolated parts. Instead, we advocate for the idea of ``borrowing strength'' for the purpose of predicting the latent surfaces via best linear unbiased prediction using the information from the entire data set, which also allows for uncertainty quantification \cite{yao2005}. Under separability, we only need to use two-dimensional surface smoothing, which is the case of the pre-smoothing approach as well. At the same time, the proposed methodology outperforms the pre-smoothing approach in terms of prediction error.
 
\section{Methodology}
\label{sec:methodology}

\subsection{Model and Observation Scheme}\label{sec:model}

We assume the existence of i.i.d. latent surfaces $X_n \in \Ltwosq$, $n=1,\ldots,N$, which are mean-square continuous with continuous sample paths. We denote the mean function as $\mu = \mu(t,s)$, where
\[
\mu(t,s) = \E X_1(t,s) , \qquad t,s \in [0,1],
\]
and the covariance kernel as $c=c(t,s,t',s')$, where
\[
c(t,s,t',s') = \Ez{ \left(X_1(t,s)-\mu(t,s)\right)\left( X_1(t',s')-\mu(t',s')\right) }, \qquad t,s,t',s' \in [0,1] .
\]
We think of the first dimension as being temporal, denoted by variable $t$, and the second dimension as being spatial, denoted by variable $s$, though this convention is only made for the purposes of presentation.

%\rubin{Note that I've formulated this as a referencable assumption (A1). In fact, I don't need this when estimating the mean function. }
The crucial assumption in this paper is that of separability of the covariance:
\begin{enumerate}[label=(A{\arabic*})]
\item \label{assumption:A1}
The covariance kernel of the random surfaces $X_1,\dots,X_N$ satisfies
\begin{equation}\label{eq:separability}
c(t,s,t',s') = a(t,t') b(s,s'), \qquad t,s,t',s' \in [0,1],
\end{equation}
for some purely temporal covariance $a=a(t,t')$ and some purely spatial covariance $b=b(s,s')$.
\end{enumerate}

The process $X \in \Ltwosq$ is separable if it is, for example, an outer product of two independent univariate processes (a purely temporal one and a purely spatial one). In that case, apart from the covariance, the mean function also separates into a product of a purely temporal and a purely spatial functions. However, a process can have a separable covariance even when it is not itself separable \cite{rougier2017}, for example the mean function may not be separable. We do not assume separability of the latent process itself, we only assume separability of its covariance in the sense of \eqref{eq:separability}.

We work under the sparse sampling regime, where every surface is observed only at a finite number of irregularly distributed locations, and those measurements are corrupted by independent additive errors. For the $n$-th latent surface $X_n$, the number of measurements $M_n$ as well as the locations of the measurements $\{ (t_{nm}, s_{nm}) \;|\; m=1,\ldots,M_n \} \subset [0,1]^2$ are considered random, and the observations are given by the errors-in-measurements model \cite{yao2005, li2010,zhang2016}:
\begin{equation}\label{eq:measurement_scheme}    
Y_{nm} = X_n(t_{nm}, s_{nm}) + \varepsilon_{nm}, \quad m=1,\ldots,M_n,\; n=1,\ldots,N,
\end{equation}
where $\varepsilon_{nm}$ are i.i.d. with $\E\, \varepsilon_{nm} = 0$ and $\var(\varepsilon_{nm}) = \sigma^2 > 0$ being the noise level.

\subsection{Motivation}\label{sec:motivation}

In this section, we assume for simplicity that the mean is zero, and provide a heuristic description of how one might estimate the separable covariance \eqref{eq:separability}. Note that we have
\[
\cov(Y_{nm},Y_{nm'}) = a(t_{nm},t_{nm'}) b(s_{nm},s_{nm'}) + \sigma^2 \mathds{1}_{[m=m']}, \quad n=1,\ldots,N, \;m,m'=1,\ldots,M_n.
\]

Consider the ``raw covariances" $G_{nmm'} := Y_{nm} Y_{nm'}$. Ignoring the assumption of separability, one could attempt to ``lift" the PACE approach \cite{yao2005} up to higher dimensions. This amounts to plotting the raw covariances as a scatter plot in a four-dimensional domain, discarding the diagonal covariances burdened by noise, and using a surface smoother to obtain an estimator of the covariance. We refer to this procedure as \textit{4D smoothing}, see Section \ref{sec:implementation_details} for details. However, there are two issues with 4D smoothing. The curse of dimensionality results in an estimator of poor quality, unless surfaces are observed relatively densely and many replications are available. And, especially when the latter is true, the computational costs of smoothing in a higher dimension can be excessive. We make the assumption of separability mainly to cope with these two issues, which is often the case in the literature already when working with fully observed data \cite{gneiting2006,pigoli2018}. We do not see separability as a critical modeling assumption, but rather as a regularization, which possibly introduces a bias. Separability reduces both the statistical and the computational complexity of the covariance estimation task. This is always important when working with random surfaces, whatever their mode of observation, but becomes particularly crucial when working with sparsely observed surfaces.

In the following, we provide a heuristic on how separability can be used to our advantage in the sparse observation regime. Assuming zero mean for now, we have
\[
\E G_{nmm'} = \E Y_{nm} Y_{nm'} = a(t_{nm},t_{nm'}) b(s_{nm}, s_{nm'}) + \sigma^2 \mathds{1}_{[m=m']}.
\]
Imagine for a moment that the spatial kernel $b=b(s,s')$ is known, and consider the set of values
\begin{equation}\label{eq:set_of_scatterpoints}
\left\{ \frac{Y_{nm} Y_{nm'}}{b(s_{nm}, s_{nm'})} \;\Bigg|\; m,m'=1,\ldots,M_n, m \neq m', n=1,\ldots,N  \right\} .
\end{equation}
The expectation of every point in this set is $a(t_{nm},t_{nm'})$, so we can chart these points in a scatter plot as
\[
\left(t_{nm},t_{nm'}, \frac{Y_{nm} Y_{nm'}}{b(s_{nm}, s_{nm'})} \right),
\]
and use a two-dimensional surface smoother to obtain an estimator of $a=a(t,t')$.

Correspondingly, if one knew the temporal kernel $a=a(t,t')$ instead, the set of values \eqref{eq:set_of_scatterpoints} (with $a$ in the denominator instead of $b$) could be arranged against $s_{nm}$ and $s_{nm'}$ to obtain an estimator of $b=b(s,s')$.
When neither the temporal kernel $a$ nor the spatial $b$ is known, one can start with a fixed $b$ and iterate between updates of $a$ and $b$, smoothing a scatterplot once per every single update. 

%Given this observation, one could envisage a two-step approach, starting from an initial estimate for one covariance kernel, thus obtaining an estimate for the other, and then updating both the covariance kernels.
However, there are two issues with such an approach. Firstly, for small denominators, the corresponding points on the scatterplot are not reliable, and using them as they are can have a severe negative impact on estimation quality. Secondly, unless very few observations per surface are available, the procedure above is computationally very demanding, see Section \ref{sec:additional_simulations}.
    
To cope with these issues, we use weights for the surface smoother and utilize gridding, i.e. split the domain into disjoint intervals and work on a grid. While, gridding can significantly reduce computations already on a univariate domain \cite{yao2005}, the gains are much bigger in higher dimensions. In the following section, we introduce our methodology in full from the theoretical perspective, while computational aspects are deferred to Section \ref{sec:implementation_details}.

\subsection{Estimation of the Model Components}\label{sec:estimation}

We use local linear regression surface smoothers \cite{fan1996} to formalise the heuristic described in the previous section and estimate the components of the model from Section \ref{sec:model}, i.e. the mean $\mu=\mu(t,s)$, the temporal kernel $a=a(t,t')$, the spatial kernel $b=b(s,s')$, and the noise level $\sigma^2$. 

By applying a surface smoother to the set $\{ (x_k,y_k,z_k) \; | \; k = 1,\ldots,M\} \subset \R^3$ with given weights $\{ w_k \; | \; k=1,\ldots,M \}$, we understand calculating $\widehat{\gamma}_0 = \widehat{\gamma}_0(x,y)$ as the minimizer of the weighted sum of squares
\begin{equation}\label{eq:generic_smoother}
(\widehat{\gamma_0}, \widehat{\gamma_1}, \widehat{\gamma_2}) = \argmin_{\gamma_0, \gamma_1, \gamma_2} \sum_{k=1}^M K \left( \frac{x - x_k}{h_{1}} \right) K \left( \frac{y - y_k}{h_{2}} \right) w_k \Big[ z_k - \gamma_0 - \gamma_1 (x - x_{k}) - \gamma_2 (y - y_{k}) \Big]^2    
\end{equation}
for every fixed $(x,y) \in [0,1]^2$, where $K(\cdot)$ is a smoothing kernel function and $h_{1},h_{2}>0$ are bandwidths. Throughout the paper, we use the Epanechnikov kernel, utilize cross-validation to select the bandwidths, and mention weights only when they are not all equal.

First, we estimate the mean by applying the surface smoother to the set
\begin{equation}\label{eq:smoother_mu}
\{ (t_{nm},s_{nm},Y_{nm}) \; | \; m=1,\ldots,M_n, \; n=1,\ldots,N \}.    
\end{equation}
Denote the resulting estimator by $\widehat \mu = \widehat \mu (t,s)$.

Next, consider the ``raw'' covariances
$G_{nmm'} = \big[Y_{nm} - \widehat{\mu}(t_{nm},s_{nm}) \big] \big[Y_{nm'} - \widehat{\mu}(t_{nm'},s_{nm'})\big].$ We begin by applying the surface smoother to the set 
\[
\{ (t_{nm},t_{nm},G_{nmm'}) \; | \; m,m'=1,\ldots,M_n, \, m \neq m',\; n=1\ldots,N \}
\]
to obtain a preliminary estimator of $a=a(t,t')$, denoted by $\widehat{a}_0 = \widehat{a}_0(t,t')$. Then we use this preliminary estimator to calculate a proxy of $b = b(s,s')$, like we described in the previous section. Namely, we apply the surface smoother to the set
\begin{equation}\label{eq:set_for_B}
\left\{ \left(s_{nm},s_{nm'},\frac{G_{nmm'}}{\widehat{a}_0(t_{nm},t_{nm'})}\right) \; \Bigg| \; m,m'=1,\ldots,M_n, \, m \neq m',\; n=1\ldots,N \right\}
\end{equation}
using weights $\{ \widehat{a}_0^2(t_{nm},t_{nm'}) \}$ to obtain $\widehat{b}_0 = \widehat{b}_0(s,s')$. If the denominator in the set above is ever zero, we remove the corresponding point from the set. Note that since the weights are equal exactly to the denominators squared, it makes sense even formally that such a point is never considered for the surface smoother.
As the next step, we refine our estimator of $a$ by applying the surface smoother to the set
\begin{equation}\label{eq:set_for_A}
\left\{ \left(t_{nm},t_{nm'},\frac{G_{nmm'}}{\widehat{b}_0(s_{nm},s_{nm'})}\right) \; \Bigg| \; m,m'=1,\ldots,M_n, \, m \neq m',\; n=1\ldots,N \right\}
\end{equation}
using weights $\{ \widehat{b}_0^2(s_{nm},s_{nm'}) \}$ from which we obtain $\widehat{a} = \widehat{a}(t,t')$. Finally, we refine the estimator of $b$ by applying the surface smoother to set \eqref{eq:set_for_B} with $\widehat{a}_0$ replaced by $\widehat{a}$, resulting in the estimator $\widehat{b}=\widehat{b}(s,s')$.

Finally, once both the mean and the separable covariance have been estimated, it remains to estimate the noise level $\sigma^2$, which is of interest e.g. for the purposes of prediction. We begin by applying the surface smoother to the set
\[
\left\{ \left(t_{nm},s_{nm},G_{nmm} \right) \; | \; m=1,\ldots,M_n,\; n=1\ldots,N \right\}
\]
to obtain $\widehat{V} = \widehat{V}(t,s)$. Note that since $\E G_{nmm} \approx a(t_{nm},t_{nm}) b(s_{nm},s_{nm}) + \sigma^2$, we can estimate $\sigma^2$ by
\[
\widehat{\sigma}^2 = 4 \int_{1/4}^{3/4} \int_{1/4}^{3/4} \big[\widehat{V}(t,s) - \widehat{a}(t,t)\widehat{b}(s,s) \big] d t d s ,
\]
where (similarly to \cite{yao2005}) we integrate only along the middle part of the domain to mitigate boundary issues.

%Firstly, we estimate the mean at a fixed location by $\widehat{\mu}(t,s) = \widehat{\alpha}_0$ by minimizing the weighted sum of squares
%\[
%(\widehat{\alpha_0}, \widehat{\alpha_1}, \widehat{\alpha_2}) = \argmin_{\alpha_0, \alpha_1} \sum_{n=1}^N \sum_{m=1}^{M_n} K \left( \frac{t_{nm} - t}{h_{\mu 1}} \right) K \left( \frac{s_{nm} - s}{h_{\mu 2}} \right) \Big[ Y_{nm} - \alpha_0 - \alpha_1 (t_{nm} - t) - \alpha_2 (s_{nm}-s) \Big]^2 ,
%\]
%where $K(\cdot)$ is a kernel and $h_{\mu 1},h_{\mu 2}>0$ are bandwidths. Throughout the paper, we use the Epanechnikov kernel and utilize cross-validation to select the bandwidths.

\begin{figure}[!t]
   \centering
   \includegraphics[width=0.8\textwidth]{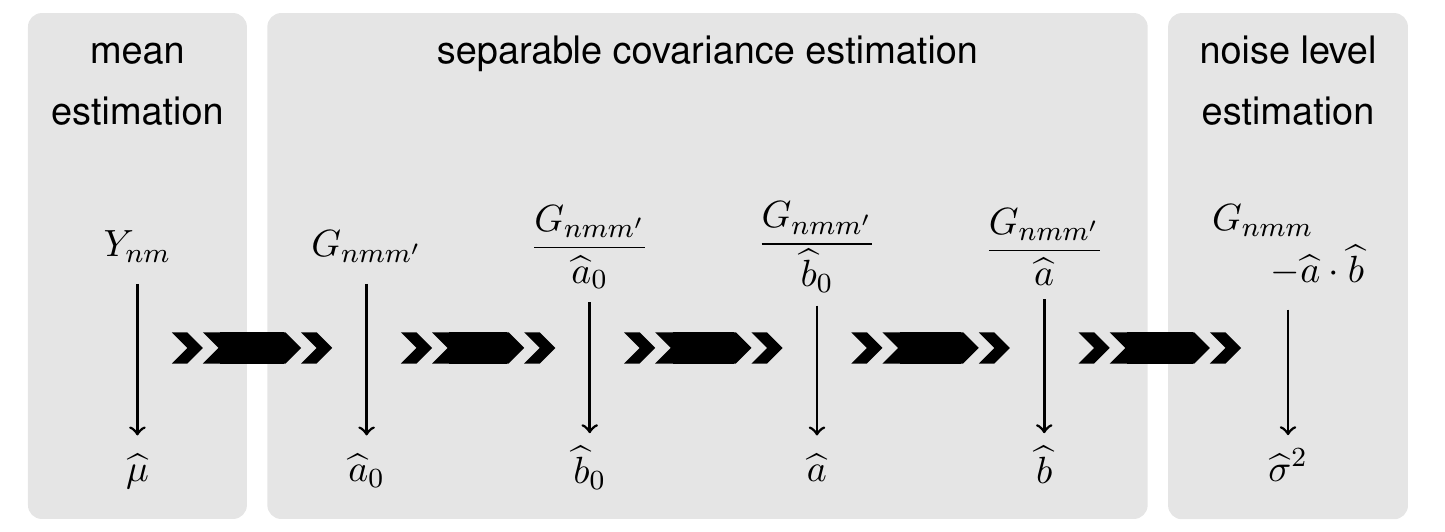}
   \caption{Workflow of the proposed estimation procedure. We estimate firstly the mean from the data $\{ Y_{nm} \}$, then the separable covariance (in several steps) from the raw covariances $\{G_{nmm'} \}$, and finally the noise level. A surface smoother over a 2D domain is utilized in every step (once per a single thin arrow).}
    \label{fig:estimation_procedure} 
\end{figure}

The workflow of the estimation scheme described above is visualised in Figure \ref{fig:estimation_procedure}. The main novelty of our approach lies in the part where the separable covariance is estimated. Separability allows us to reduce dimensionality of the problem. Hence only two-dimensional smoothing is required, while a straightforward multi-dimensional generalization of e.g. the PACE approach \cite{yao2005} not utilizing separability would require four-dimensional smoothing to estimate the covariance.

The estimation of the separable terms can be viewed as an iterative procedure, where either $a$ or $b$ is kept fixed while the other term is being updated. The initializing estimate $\widehat{a}_0$ is also obtained in this way, starting from $\widehat{b}_0 \equiv 1$. A natural question is whether one should iterate this process until convergence, or simply stop after a single step, and use e.g. $\widehat{a}_0$ as the estimator of $a$. The approach we advocate for uses exactly two steps (for both $a$ and $b$). The reason is the following. As will be shown in Section \ref{sec:asymptotics}, the asymptotic distribution of $\widehat{b}$ does not depend on $\widehat{a}_0$, and similarly for $\widehat{a}$. This fact follows from separability. However, given the motivation in the previous section, one can anticipate the finite sample performance of $\widehat{a}$ to be better than that of $\widehat{a}_0$, which is verified in our simulation study. One can think of the first step as estimating the optimal weights consistently, and the second step as using those consistently estimated weights to produce the estimators.

One of the distinctive features of our methodology is the use of this explicit weighing scheme, where the weights for each of the two covariance kernels depend on the other covariance kernel. %Any weighting scheme can also be understood as an alternative to local weighting via the smoothing kernel. 
We explain why this is done in the following example. A more precise justification for the specific (quadratic) form of the weights, which requires some additional background, is deferred to Section \ref{sec:connection_full_obs}.

\begin{example}
Assume we observe zero-mean surfaces sparsely, and one of these surfaces, say $X_1$, is observed at four locations only, as depicted in Figure \ref{fig:smoothing_comparison}. First, let us describe the 4D smoothing estimator at a fixed location $(t,s,t',s')$, when the bandwidths are also fixed. $X_1$ contributes to $\widehat{c}(t,s,t',s')$ only if there is a pair of two locations, where $X_1$ is observed, such that one of the locations is close to $(t,s)$, and the other is close to $(t',s')$. In this case, closeness in time, resp. space, is measured by $h_t$, resp. $h_s$, which control the bandwidth of the smoothing kernel. In Figure \ref{fig:smoothing_comparison} (left), only a single pair of locations contributes to estimation at $(t,s,t',s')$.

Now, let us contrast this to a single step in the proposed estimating procedure. Assume that $b=b(s,s')$ is fixed in the current step, and we are estimating $a=a(t,t')$. In this step, the spatial dimensions are not explicitly considered, we are performing smoothing only in the temporal dimensions. Hence the product of any two locations, where $X_1$ is observed, contributes to $\widehat{a}$, as long as the locations are close to $(t,s)$ and $(t',s')$ in the temporal domain. In the situation displayed in Figure \ref{fig:smoothing_comparison} (right), this leads to four contributing raw covariance pairs. In other words, when estimating the temporal part of the covariance at $(t,s,t',s')$, we can consider even raw covariances, \emph{which are spatially far}  from $(t,s,t',s')$. This is meaningful due to separability. The adopted weighting scheme then ensures that raw covariances arising from points which are spatially distant are appropriately weighted.

To sum up, 4D smoothing can be understood as averaging over information about $c(t,s,t',s')$ captured in raw covariances, whose locations are close to $(t,s,t',s')$. Under separability, however, the proposed methodology borrows information in a different manner, always allowing for more freedom in one dimension or the other, depending on which dimension is currently held fixed.
\end{example}

\begin{figure}[!t]
   \centering
   \begin{tabular}{cccc}
   \includegraphics[height=6cm]{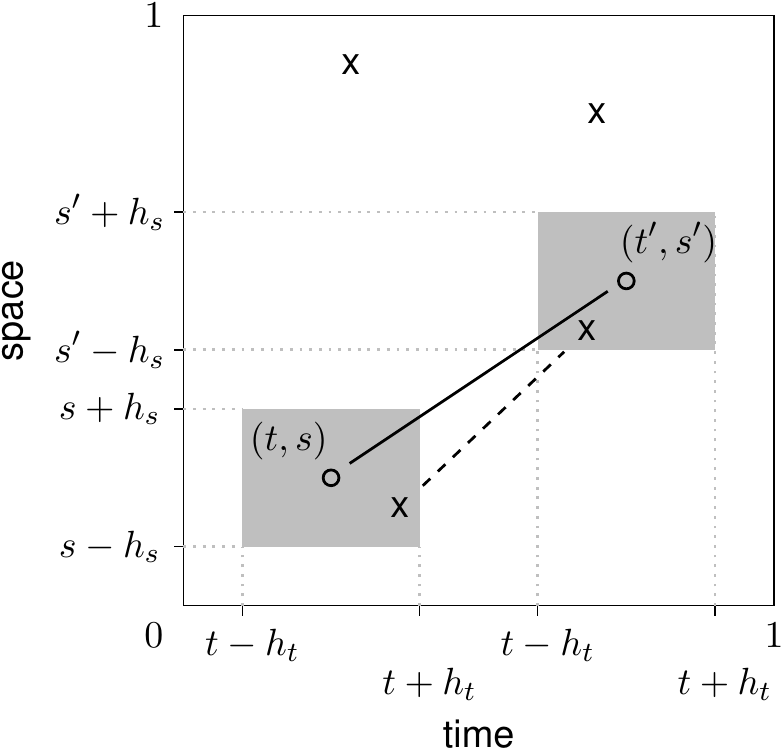} & & &
   \includegraphics[height=6cm]{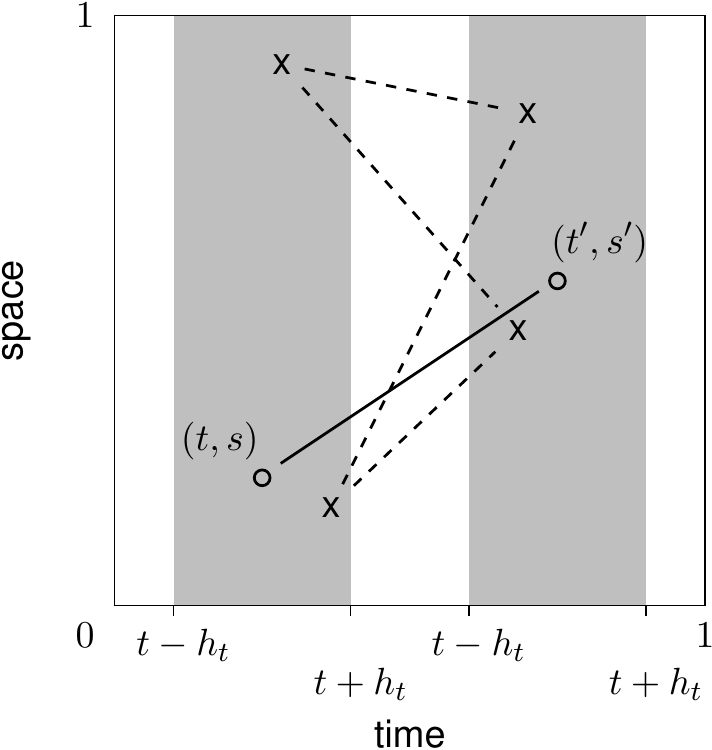}
   \end{tabular}  
   \caption{A single random surface is observed at four locations (depicted by ``$\mathsf{x}$''), and these observations contribute differently to estimation of the covariance at a fixed location $(t,s,t',s')$ in the case of 4D smoothing \emph{(\textbf{left})} and one step of the proposed approach leading to the estimator of $a(t,t')$ \emph{(\textbf{right})}. The gray areas depict active neighborhoods and the dashed lines depict the contributing raw covariances (products of the values at the connected locations).   }
    \label{fig:smoothing_comparison} 
\end{figure}

%%%%%%%%%%%%%%%%%%%%%%%%%%%%%%%%%%%%%%%%%%%%%%%%%%%%%%%%%%%%%%%%%%%%%%%%%%%%%%%%%%%%%%%%
%%%%%%%%%%%%%%%%%%%%%%%%%%%%%%%%%%%%%%%%%%%%%%%%%%%%%%%%%%%%%%%%%%%%%%%%%%%%%%%%%%%%%%%%
\subsection{Prediction}
\label{sec:prediction}

Another objective of our methodology is the recovery  of the latent surfaces based on the sparse and noisy observations thereon.
The prediction method we are going to present in this section follows the principle of \textit{borrowing strength} across the entire data set, an expression framed by \citet{yao2005}. Specifically, consider the training data set composed of the sparse observations made on random surfaces $X_1,\dots,X_N$, i.e. $\{Y_{nm} : m=1,\dots,M_n,\, n=1,\dots,N\}$, and a new random surface $X^{new}$ observed under the same sampling protocol as \eqref{eq:measurement_scheme}:
\begin{equation}\label{eq:prediction_Ynew_sparse_measurements}
   Y^{new}_m = X^{new}( t^{new}_m, s^{new}_m) + \varepsilon^{new}_{m},\qquad m=1,\dots,M^{new}. 
\end{equation}
We assume that $X^{new}$ comes from the same population as $X_1,\dots,X_N$. In fact, it may be set (together with its sparse measurements) to one of the training surfaces $X_1,\dots,X_N$ if the task is to predict one of those.

Our prediction method is calibrated on all the observations, and this information is used for the prediction of $X^{new}$. This contrasts to the pre-smoothing step often used in the functional data literature \citep{ramsay2005,ramsay2007applied}, typically in the dense regime, where the prediction (usually some kind of a smoother) of $X^{new}$ is based only on $\{Y^{new}_m : m=1,\dots,M^{new}\}$.

Let $X^{new}$ be a random surface with the mean $\mu(t,s),\,t,s\in[0,1],$ and the separable covariance kernel $a(t,t')b(s,s')$, $t,t',s,s'\in[0,1]$, observed through sparse measurements \eqref{eq:prediction_Ynew_sparse_measurements}.
Then the best linear unbiased predictor of the latent surface $X^{new}$ given the sparsely observed data $\mathbb{Y}^{new} = ( Y_1^{new}, \ldots, Y_{M^{new}}^{new}  )$, denoted as $\Pi( X^{new} \vert \mathbb{Y}^{new} )$, is given by (cf. \citep{henderson1975best}):
\begin{equation}\label{eq:blup}
\Pi( X^{new}(t,s) \vert \mathbb{Y}^{new} ) = \mu(t,s) + \cov( X^{new}(t,s), \mathbb{Y}^{new} ) \left[ \var( \mathbb{Y}^{new}) \right]^{-1}( \mathbb{Y}^{new} - \E \mathbb{Y}^{new} ),\qquad t,s\in[0,1],
\end{equation}
where
\begin{align}
    \label{eq:blup_cov}
    \cov( X^{new}(t,s), \mathbb{Y}^{new} ) &= \Big( a(t,t^{new}_m) b(s,s^{new}_m) \Big)_{m=1}^{M^{new}} \in\mathbb{R}^{M^{new}},
    \qquad t,s\in [0,1],\\
    \label{eq:blup_var}
    \var( \mathbb{Y}^{new} ) &= \Big( a(t^{new}_m,t^{new}_{m'}) b(s^{new}_m,s^{new}_{m'}) + \sigma^2 \mathds{1}_{[m=m']}  \Big)_{m,m'=1}^{M^{new}} \in\mathbb{R}^{M^{new}\times M^{new}},\\
    \nonumber
    \E \mathbb{Y}^{new} &= \Big( \mu( t^{new}_m, s^{new}_m ) \Big)_{m=1}^{M^{new}} \in\mathbb{R}^{M^{new}}.
\end{align}

Formula \eqref{eq:blup} contains the unknown mean surface $\mu(\cdot,\cdot)$, covariance kernels $a(\cdot,\cdot)$ and $b(\cdot,\cdot)$ as well as the measurement error variance $\sigma^2$.
In reality, we estimate these quantities from the training data set, say sparsely observed measurements on $X_1,\dots,X_N$, and plug-in the estimates $\widehat{\mu}, \widehat{a}, \widehat{b}$, and $\widehat{\sigma}^2$ into \eqref{eq:blup}.
We shall denote such predictor as $\widehat{\Pi}( X^{new} \vert \mathbb{Y}^{new} )$.
%Moreover, we denote $\cov( X^{new}(t,s), \mathbb{Y}^{new} )$ and 
We show in the following section (Theorem~\ref{thm:prediction_bands}) that the predictor $\widehat{\Pi}( X^{new} \vert \mathbb{Y}^{new} )$ converges to its theoretical counterpart $\Pi( X^{new} \vert \mathbb{Y}^{new} )$ as the number of training samples $N$ grows to infinity.

In the rest of this section, we turn our attention to the construction of  confidence bands under a Gaussian assumption.
\begin{enumerate}[label=(A{\arabic*}),resume]
\item \label{assumption:A2}
The random surface $X^{new}$ is a Gaussian random element in $\Ltwosq$ and the measurement error ensemble $\{ \varepsilon_m^{new} \}_{m=1}^{M^{new}}$ is a Gaussian random vector.
\end{enumerate}
First note that under the assumption \ref{assumption:A2}, the best linear unbiased predictor \eqref{eq:blup} actually corresponds to the conditional expectation $\E [ X^{new}(t,s) \vert \mathbb{Y}^{new}]$. Furthermore, we may calculate the conditional covariance structure as
\begin{multline}\label{eq:blup_variance}
\cov\left( X^{new}(t,s),X^{new}(t',s') \vert \mathbb{Y}^{new} \right) \\=
a(t,t')b(s,s') -
\cov( X^{new}(t,s), \mathbb{Y}^{new} )
\left[\var( \mathbb{Y}^{new} )\right]^{-1}
\left[ \cov(  X^{new}(t',s') , \mathbb{Y}^{new} ) \right]^\top,
\qquad t,t',s,s'\in[0,1].
\end{multline}
Moreover, we denote $\widehat{\cov}\left( X^{new}(t,s),X^{new}(t',s') \vert \mathbb{Y}^{new} \right)$ the empirical counterpart to \eqref{eq:blup_variance}, where the unknown quantities $a$, $b$ and $\sigma^2$ are replaced by their estimators.

Fixing $(t,s)\in[0,1]^2$ and $\alpha\in(0,1)$, the $(1-\alpha)$-confidence interval for
$X^{new}(t,s)$ is given by
\begin{equation}\label{eq:confidence_band_pointwise}
\widehat{\Pi}(X^{new}(t,s) \vert \mathbb{Y}^{new}) \pm u_{1-\alpha/2}
\sqrt{ \widehat{\cov}\left( X^{new}(t,s),X^{new}(t,s) \vert \mathbb{Y}^{new} \right) }
\end{equation}
where $u_{1-\alpha/2}$ is the $(1-\alpha/2)$-quantile of the standard Gaussian law.
The point-wise confidence band is then constructed by connecting the intervals \eqref{eq:confidence_band_pointwise} for all $(t,s)\in[0,1]^2$.

The construction of the simultaneous confidence band is more involved, and we shall use the technique proposed by \citet{degras2011simultaneous}. Define the conditional correlation kernel
\begin{equation}\label{eq:corr_kernel_estimated_Z}
\widehat{\corr}\left( X^{new}(t,s),X^{new}(t',s') \vert \mathbb{Y}^{new} \right)
=
\frac{ \widehat{\cov}\left( X^{new}(t,s),X^{new}(t',s') \vert \mathbb{Y}^{new} \right) }{\sqrt{  \widehat{\var}\left( X^{new}(t,s) \vert \mathbb{Y}^{new} \right)
\widehat{\var}\left( X^{new}(t',s')\vert \mathbb{Y}^{new} \right)
}}        
\end{equation}
if the division on the right-hand side makes sense and zero otherwise, and
where
$\widehat{\var}\left( X^{new}(t,s) \vert \mathbb{Y}^{new} \right)
=
\widehat{\cov}\left( X^{new}(t,s),X^{new}(t,s) \vert \mathbb{Y}^{new} \right)$.

Then, construct the simultaneous confidence band by connecting the intervals
\begin{equation}\label{eq:confidence_band_simul}
    \widehat{\Pi}(X^{new}(t,s) \vert \mathbb{Y}^{new}) \pm \widehat{z}_{1-\alpha}
\sqrt{ \widehat{\cov}\left( X^{new}(t,s),X^{new}(t,s) \vert \mathbb{Y}^{new} \right) },
\end{equation}
where $\widehat{z}_{1-\alpha}$ is the $(1-\alpha)$-quantile of the law of
\begin{equation}\label{eq:prediction_law_hat_W}
\widehat{W} = \sup_{t,s\in[0,1]^2} \left| \widehat{Z}(t,s) \right|
\end{equation}
with $\widehat{Z}$ being a Gaussian random surface with $\cov( \widehat{Z}(t,s),\widehat{Z}(t',s') ) = \widehat{\corr}\left( X^{new}(t,s),X^{new}(t',s') \vert \mathbb{Y}^{new} \right)$ for $t,t',s,s'\in[0,1]$, where $\widehat{\corr}$ is the from \eqref{eq:corr_kernel_estimated_Z}.
Therefore, by the definition, $ \Prob ( \sup_{(t,s)\in[0,1]^2} |\widehat{Z}(t,s)| \leq \widehat{z}_{1-\alpha} ) = 1 - \alpha $.
Numerical calculation of this quantile is explained by \citet{degras2011simultaneous}, who also concludes that $\widehat{z}_{1-\alpha} < u_{1-\alpha/2}$. Therefore the point-wise confidence band is always enveloped by the simultaneous confidence band, as expected. 

The asymptotic coverage of the point-wise band \eqref{eq:confidence_band_pointwise} and the simultaneous band \eqref{eq:confidence_band_simul} is verified in Theorem~\ref{thm:prediction_bands} in the following section.

%%%%%%%%%%%%%%%%%%%%%%%%%%%%%%%%%%%%%%%%%%%%%%%%%%%%%%%%%%%%%%%%%%%%%%%%%%%%%%%%%%%%%%%%%%%%%%%%%%%%%%%%%%%%%%%%%%%%%%%%%%%%%%%%%%%%%%%%%%%%%%%%%%%%%%%%%%%%%%%%
%%%%%%%%%%%%%%%%%%%%%%%%%%%%%%%%%%%%%%%%%%%%%%%%%%%%%%%%%%%%%%%%%%%%%%%%%%%%%%%%%%%%%%%%%%%%%%%%%%%%%%%%%%%%%%%%%%%%%%%%%%%%%%%%%%%%%%%%%%%%%%%%%%%%%%%%%%%%%%%%
%%%%%%%%%%%%%%%%%%%%%%%%%%%%%%%%%%%%%%%%%%%%%%%%%%%%%%%%%%%%%%%%%%%%%%%%%%%%%%%%%%%%%%%%%%%%%%%%%%%%%%%%%%%%%%%%%%%%%%%%%%%%%%%%%%%%%%%%%%%%%%%%%%%%%%%%%%%%%%%%

\section{Asymptotic Properties}\label{sec:asymptotics}

In this section, we establish consistency and convergence rates of the estimators $\widehat{\mu}=\widehat{\mu}(t,s)$, $\widehat{a}=\widehat{a}(t,t')$, and $\widehat{b}=\widehat{b}(s,s')$, as well as the measurement error variance $\sigma^2$.

The following assumptions refine the sparse observation scheme introduced in Section~\ref{sec:model}.

\begin{enumerate}[label=(B{\arabic*})]
\item \label{assumption:B1}
The counts of measurements per surface $M_n$ are independent identically distributed random variables with the law $M_n\sim\mathcal{M}>0$ such that $\Prob( \mathcal{M} > 1 )>0$ and $\mathcal{M} \leq M^{max}$ where $M^{max} \in \mathbb{N}$ is a constant.
\item \label{assumption:B2}
The measurement locations $(t_{nm},s_{nm}),\,m=1,\dots,M_n,\,n=1,\dots,N$, are independent identically distributed random variables generated from the density $f_{(t,s)}(\cdot,\cdot)$ on $[0,1]^2$.
The density $f_{(t,s)}(\cdot,\cdot)$ is assumed to be twice continuously differentiable and positive on $[0,1]^2$.
\item \label{assumption:B3}
The counts $(M_n)$, the locations $(t_{nm},s_{nm})$, and the latent surfaces $(X_n)$ are independent. 
\end{enumerate}

The following two assumptions are required for consistent estimation of the mean surface.

\begin{enumerate}[label=(B{\arabic*}),resume]
\item \label{assumption:B4}
The mean surface $\mu(\cdot,\cdot)$ is twice continuously differentiable on $[0,1]^2$.
\item \label{assumption:B5}
There exists $\rho>2$ such that the random surface $X_1$ and the measurement error $\varepsilon_{11}$ satisfy
$$ \sup_{(t,s)\in[0,1]^2} \Ez{ |X_1(t,s)|^\rho } < \infty, \qquad \Ez{ |\varepsilon_{11}|^\rho } < \infty.$$
\item \label{assumption:B6}
The bandwidths $h_{\mu,1},h_{\mu,2}$ for the mean estimator satisfy
$ (\log N) / (N h_{\mu,1} h_{\mu,2} ) = o(1) $ and, furthermore, we assume that they decay with the same rate: $ h_{\mu,1} \asymp h$ and $ h_{\mu,2} \asymp h$ as $N\to\infty$. The statement $x_n \asymp x'_n$ as $n\to\infty$ for two sequences $\{x_n\}$ and $\{x'_n\}$ is understood as $\lim_{n\to\infty} x_n / x'_n \in (0,\infty)$, i.e. $x_n$ and $x'_n$ differ asymptotically up to a multiplicative constant
%$h_{\mu,1}\to 0,\, h_{\mu,2}\to 0, \, N (h_{\mu,1} h_{\mu,2})^2 \to \infty$.
\end{enumerate}
The common decay rate assumption in \ref{assumption:B6} is not really required for our asymptotic theory; it is imposed to simplify the statements on the convergence rates. The same argument applies to the other bandwidths below. Because all smoothing in our methodology is restricted to two dimensions, the bandwidths are indeed expected to decay with the same rate, and we may assume that they differ asymptotically up to a multiplicative constant.

In order to estimate the covariance kernels $a(\cdot,\cdot)$ and $b(\cdot,\cdot)$ we need the following assumptions:
\begin{enumerate}[label=(B{\arabic*}),resume]
\item\label{assumption:B7}
The covariance kernels $a(\cdot,\cdot)$ and $b(\cdot,\cdot)$ are twice continuously differentiable on $[0,1]^2$.
\item\label{assumption:B8}
There exists $\rho'>2$ such that the random surface $X_1$ and the measurement error $\varepsilon_{11}$ satisfy
$$ \sup_{(t,t',s,s')\in[0,1]^4} \Ez{ |X_1(t,s) X_1(t',s')|^{\rho'} } < \infty, \qquad \Ez{ |\varepsilon_{11}|^{2\rho'} } < \infty.$$
\item\label{assumption:B9}
The bandwidths $h_a$ and $h_b$ used for smoothing the covariance kernel $a(\cdot,\cdot)$ and $b(\cdot,\cdot)$ satisfy
$ (\log N) / (N h_a^2 ) = o(1) $ and $ (\log N) / (N h_b^2 ) = o(1) $ as $N\to\infty$ and, for simplicity of the convergence rates statements, we assume
$h_a \asymp h$ and $h_b \asymp h$ as $N\to\infty$ where $h$ is from assumption~\ref{assumption:B6}.
\item\label{assumption:B10}
The true value of the covariance kernel $b(\cdot,\cdot)$, of the separable model \eqref{eq:separability} satisfies
\begin{equation}\label{eq:assumption:B10_equation}
\Theta \stackrel{\mathrm{def}}{=} \int_0^1\int_0^1 b(s,s') f_s(s)f_s(s') dsds' \neq 0
\end{equation}
where $f_s(s) = \int_0^1 f_{(t,s)}(t,s) dt$ is the marginal density of the random location $s_{11}$.
\end{enumerate}

While the other assumptions are standard in the smoothing literature, assumption \ref{assumption:B10} might be surprising, especially considering it is not symmetric between $a(\cdot,\cdot)$ and $b(\cdot,\cdot)$. The reason behind this asymmetry is that our estimation methodology starts by smoothing the raw covariances $G_{nmm'}$ against $(t_{nm},t_{nm'})$ in order to produce the preliminary estimator $\widehat{a}_0(\cdot,\cdot)$. The condition \eqref{eq:assumption:B10_equation} ensures that the estimator $\widehat{a}_0(\cdot,\cdot)$ converges to a nonzero quantity, see the constant $\Theta$ in Theorem~\ref{thm:estim_a_b}. By contrast, this issue is not present in the follow-up steps. Due to positive semi-definitness of $b(\cdot,\cdot)$, the constant $\Theta$ can only be zero if all eigenfunctions of $b(\cdot,\cdot)$ are orthogonal to the marginal sampling density $f_s$. This cannot happen e.g. unless $B$ is exactly low-rank, and with all the eigenfunctions changing signs. From the practical perspective, the condition is merely a technicality.

The estimation of the noise level $\sigma^2$ furthermore requires the following assumption.
\begin{enumerate}[label=(B{\arabic*}),resume]
\item\label{assumption:B11}
The bandwidths $h_{V,1},h_{V,2}$ for the smoother $\widehat{V}(\cdot,\cdot)$ satisfy
$ (\log N) / (N h_{V,1} h_{V,2} ) = o(1) $
and, for simplicity of the convergence rates statements, we assume
$h_{V,1} \asymp h$ and $h_{V,2} \asymp h$ as $N\to\infty$ where $h$ is from assumption~\ref{assumption:B6}.
\end{enumerate}

The mean surface asymptotic theory is presented as the following proposition.
Note that the separability assumption \ref{assumption:A1} is not required here.
\begin{proposition}
\label{prop:estim_mu}
Under the assumptions \ref{assumption:B1} -- \ref{assumption:B6}:
$$ \sup_{(t,s)\in[0,1]^2} \left| \widehat{\mu}(t,s) - \mu(t,s) \right|
= \bigop{ \sqrt{\frac{\log N}{N h^2}} + h^2 }
\qquad\text{as}\quad N\to\infty.
$$
\end{proposition}

Our main asymptotic result, the consistency and the convergence rates for the separable model components \eqref{eq:separability}, is presented in the following theorem.
\begin{theorem}\label{thm:estim_a_b}
Under the assumptions \ref{assumption:A1}, \ref{assumption:B1} -- \ref{assumption:B10}:
\begin{align}
    \label{eq:thm:estim_a_b:asymp_a}
    \sup_{(t,t')\in[0,1]^2} \left| \widehat{a}(t,t') - \Theta a(t,t') \right| &= \bigop{ \sqrt{\frac{\log N}{N h^2}} + h^2 }, \\
    \label{eq:thm:estim_a_b:asymp_b}
    \sup_{(s,s')\in[0,1]^2} \left| \widehat{b}(s,s') - \frac{1}{\Theta} b(s,s') \right| &= \bigop{ \sqrt{\frac{\log N}{N h^2}} + h^2 },
\end{align}
as $N\to\infty$, where $\Theta$ is defined in \eqref{eq:assumption:B10_equation}.
\end{theorem}

The separable decomposition \eqref{eq:separability} is not identifiable, because a constant can multiply one component while dividing the other, i.e. $a(t,t')b(s,s') = \left[ \lambda a(t,t') \right] \left[ (1/\lambda) b(s,s') \right]$, $t,t',s,s'\in[0,1]$, for any $\lambda\in(0,\infty)$.
Therefore we can only aim to recover the covariance kernels $a(\cdot,\cdot)$ and $b(\cdot,\cdot)$ up to a multiplicative constant and its reciprocal, respectively. 
The number $\Theta$ in statements \eqref{eq:thm:estim_a_b:asymp_a} and \eqref{eq:thm:estim_a_b:asymp_b} plays the role of such a constant and depends on the initialization of the algorithm, in our case on the fact that the first estimator $\widehat{a}_0$ smooths the raw covariances $G_{nmm'}$ without any weighting.
Still, the product $\widehat{a}(t,t')\widehat{b}(s,s')$, $t,t',s,s'\in[0,1]$,, estimates consistently the covariance structure $c(t,s,t',s') = a(t,t')b(s,s')$, $t,t',s,s'\in[0,1]$,, which is summarised in the following corollary.

\begin{corollary}\label{corollary:estimation_complete_covariance_structure}
Under the assumptions  \ref{assumption:A1}, \ref{assumption:B1} -- \ref{assumption:B10}:
$$ \sup_{(t,s,t',s')\in[0,1]^4} \left| \widehat{a}(t,t')\widehat{b}(s,s') - a(t,t')b(s,s') \right| = \bigop{ \sqrt{\frac{\log N}{N h^2}} + h^2 } $$
as $N\to\infty$.
\end{corollary}

Finally, the asymptotic behaviour of the noise level $\sigma^2$ is given as the following proposition.
\begin{proposition}\label{prop:estim_sigma}
Under the assumptions  \ref{assumption:A1}, \ref{assumption:B1} -- \ref{assumption:B11}:
$$
\widehat{\sigma}^2 = \sigma^2 + \bigop{ \sqrt{\frac{\log N}{N h^2}} + h^2 }\qquad\text{as}\quad N\to\infty.
$$
\end{proposition}

This completes the asymptotic theory for our estimsators, and we now turn to prediction.
The following theorem shows that the predictor $\widehat{\Pi}( X^{new} \vert \mathbb{Y}^{new} )$ defined in Section~\ref{sec:prediction} converges -- as the sample size grows to infinity -- to its oracle counterpart \eqref{eq:blup}, which assumes the knowledge of the true distribution of the data. Moreover, the theorem also proves the asymptotic coverage of the point-wise and simultaneous confidence bands \eqref{eq:confidence_band_pointwise} and \eqref{eq:confidence_band_simul}.
\begin{theorem}
\label{thm:prediction_bands}
Under the assumptions  \ref{assumption:A1}, \ref{assumption:B1} -- \ref{assumption:B11}:
\begin{equation}\label{eq:thm:prediction_bands_statement_1}
\sup_{(t,s)\in[0,1]^2} \left| \widehat{\Pi}( X^{new}(t,s) \vert \mathbb{Y}^{new} ) - \Pi( X^{new}(t,s) \vert \mathbb{Y}^{new} ) \right| = \smallop{1}, \qquad\text{as}\quad N\to\infty,
\end{equation}
conditionally on $\mathbb{Y}^{new}$. %$\mathbb{Y}^{new},(t^{new}_1,s^{new}_1),\dots,(t^{new}_{M^{new}},s^{new}_{M^{new}})$.

Assuming further \ref{assumption:A2} and fixing $\alpha\in(0,1)$:
\begin{align*}
\forall (t,s)\in[0,1]^2 \lim_{N\to\infty} \Prob\left( 
\left| \widehat{\Pi}(X^{new}(t,s) \vert \mathbb{Y}^{new} ) - X^{new}(t,s) \right|  
\leq u_{1-\alpha}
\sqrt{ \widehat{\var}\left( X^{new}(t,s) \vert \mathbb{Y}^{new} \right) }
\,\middle|\, \mathbb{Y}^{new} \right)
 &= 1-\alpha,\\
\lim_{N\to\infty}\Prob\left(
\sup_{(t,s)\in[0,1]^2}
\left[ \widehat{\var}\left( X^{new}(t,s)\vert \mathbb{Y}^{new} \right) \right]^{-1/2}
\left| \widehat{\Pi}(X^{new}(t,s) \vert \mathbb{Y}^{new} ) - X^{new}(t,s) \right|  
\leq \widehat{z}_{1-\alpha}
\,\middle|\, \mathbb{Y}^{new} \right)
&= 1-\alpha.
\end{align*}

\end{theorem}

The rates established in this section manifest the statistical consequences of separability. Corollary~\ref{corollary:estimation_complete_covariance_structure} shows that the complete covariance structure is estimated with the rate $O_\Prob(\sqrt{(\log N) / (N h^2)} + h^2)$, which is the known optimal minimax convergence rate \citep{fan1996} for two dimensional non-parametric regression. 
By steps similar to our proofs presented in Section~\ref{sec:proofs}, it could be shown that the empirical covariance smoother yields the convergence rate $O_\Prob(\sqrt{(\log N) / (N h^4)} + h^2)$.
The empirical covariance smoother's convergence rate is thus slower than the one found in Corollary~\ref{corollary:estimation_complete_covariance_structure}, achieved via the separable model.

\section{Simulation Study}\label{sec:simulation_study}

We explore the finite sample performance of the proposed methodology by means of a moderate simulation study (total runtime of about one thousand CPU hours). Computational efficiency (relatively small runtimes) is achieved by working on a $20 \times 20$ grid, like described in Section \ref{sec:marginalization}. 
Every surface $X_1,\ldots,X_{100}$ is first sampled fully on this grid as a zero-mean matrix-variate Gaussian with covariance $C$ (to be specified), superposed with noise (zero-mean i.i.d. Gaussian entries with variance $\sigma^2$), and then sub-sampled in a way that only a fraction of the entries, selected at random, is retained. The covariance $C$ is always standardized to have trace one, and $\sigma^2$ is chosen such that the gridded white noise process is also trace one.

\textbf{Methods compared.} We compare the proposed separable estimator $\widehat{c} = \widehat{a} \cdot \widehat{b}$ against the non-separable empirical estimator obtained by local linear smoothing in four dimensions (4D smoothing), and also against the best separable approximation (BSA) \cite{genton2007,masak2020} obtained from the fully observed and noise-free surfaces. We also compare the proposed estimator against its \textit{one-step} version ($\widehat{c} = \widehat{a}_0 \cdot \widehat{b}_0$, cf. Figure \ref{fig:estimation_procedure}), while it is shown in Section \ref{sec:additional_simulations} that more than two steps lead to a similar performance as the proposed two-step estimator.

\textbf{Covariance choices.} We consider four specific choices for the covariance: (a) a separable covariance with Fourier basis eigenfunctions and power decay of the eigenvalues; (b) Brownian sheet covariance, which is also separable; (c) the parametric covariance introduced by \citet{gneiting2002}, which is non-separable; and (d) a superposition of the Fourier covariance from (a) with another separable covariance, having shifted Legendre polynomial eigenfunctions and power decay of the eigenvalues. We discuss these choices in detail in Section \ref{sec:setup_simulations}. However, for the purpose of presenting the simulation results, it is only important to point out the following. Firstly, while the Fourier setup (a) and the Brownian setup (b) are separable, the Gneiting setup (c) and the Fourier-Legendre setup (d) are non-separable. Secondly, while the Brownian setup (b) and Gneiting setup (c) lead to rather flat covariances, the Fourier setup (a) and the Fourier-Legendre setup (d) lead to quite wiggly (though infinitely smooth) covariances.

\textbf{Sparsity.} In all simulation setups, we consider different percentages of the entries observed $p \in \{2, 5, 10, 20, 40, 70\}$. Since the grid size is $20 \times 20$, this means for example that for $p=2$ we have $2/100 \cdot 20^2 = 8$ observations per surface. For all the setups and percentages, we report relative estimation errors $\|\widehat{C} - C\|_2/\| C\|_2$, where $\widehat{C}$ is an estimator obtained by one of the four methods above (4D smoothing, one-step, proposed, or BSA). The results are shown in Figure \ref{fig:simulation_results}. Every reported error was calculated as an average of 100 Monte Carlo runs.

\begin{figure}[!t]
   \centering
   \begin{tabular}{ccc}
   (a) Fourier -- separable, wiggly && (b) Brownian -- separable, flat  \\
   \includegraphics[width=0.45\textwidth]{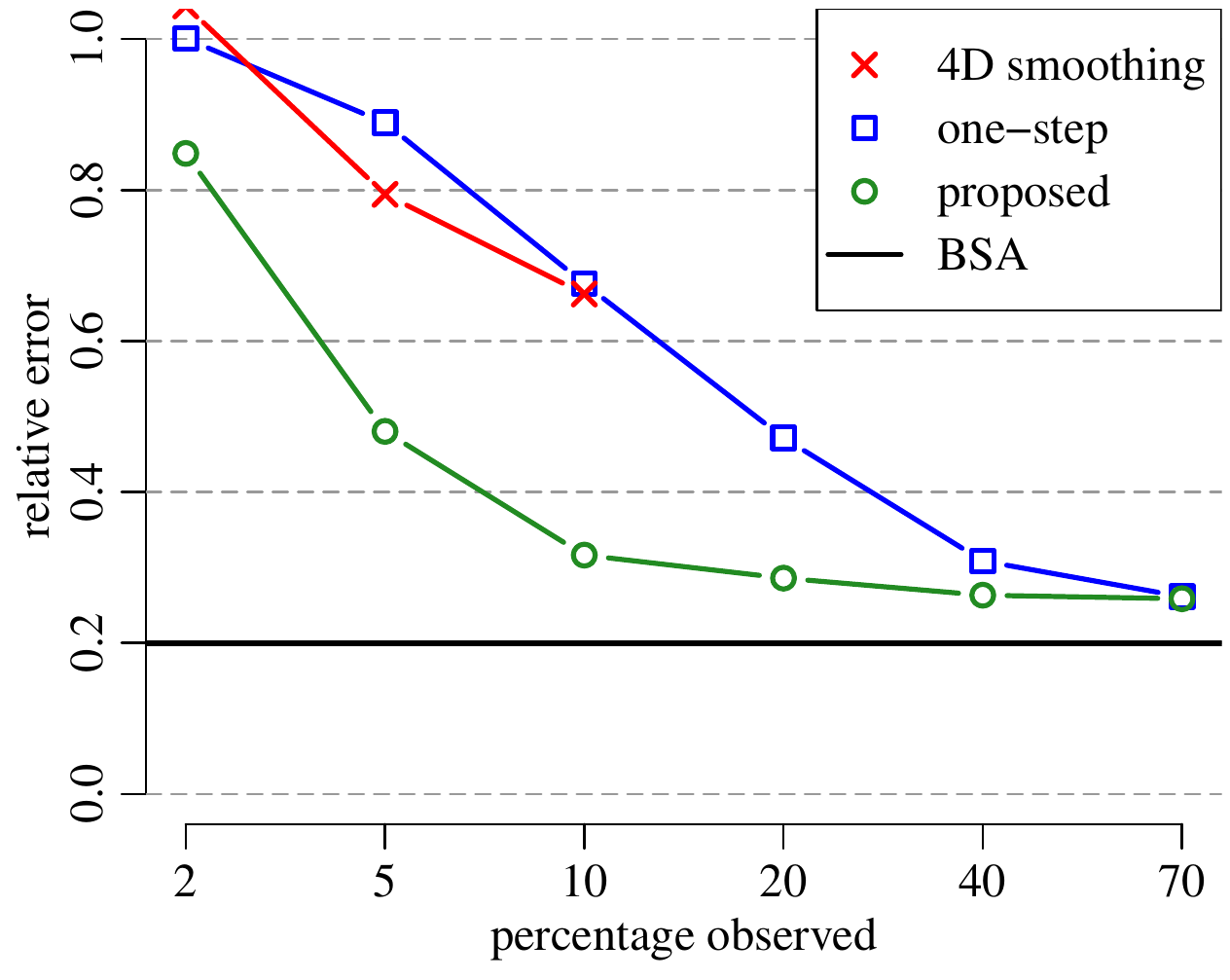} &&
   \includegraphics[width=0.45\textwidth]{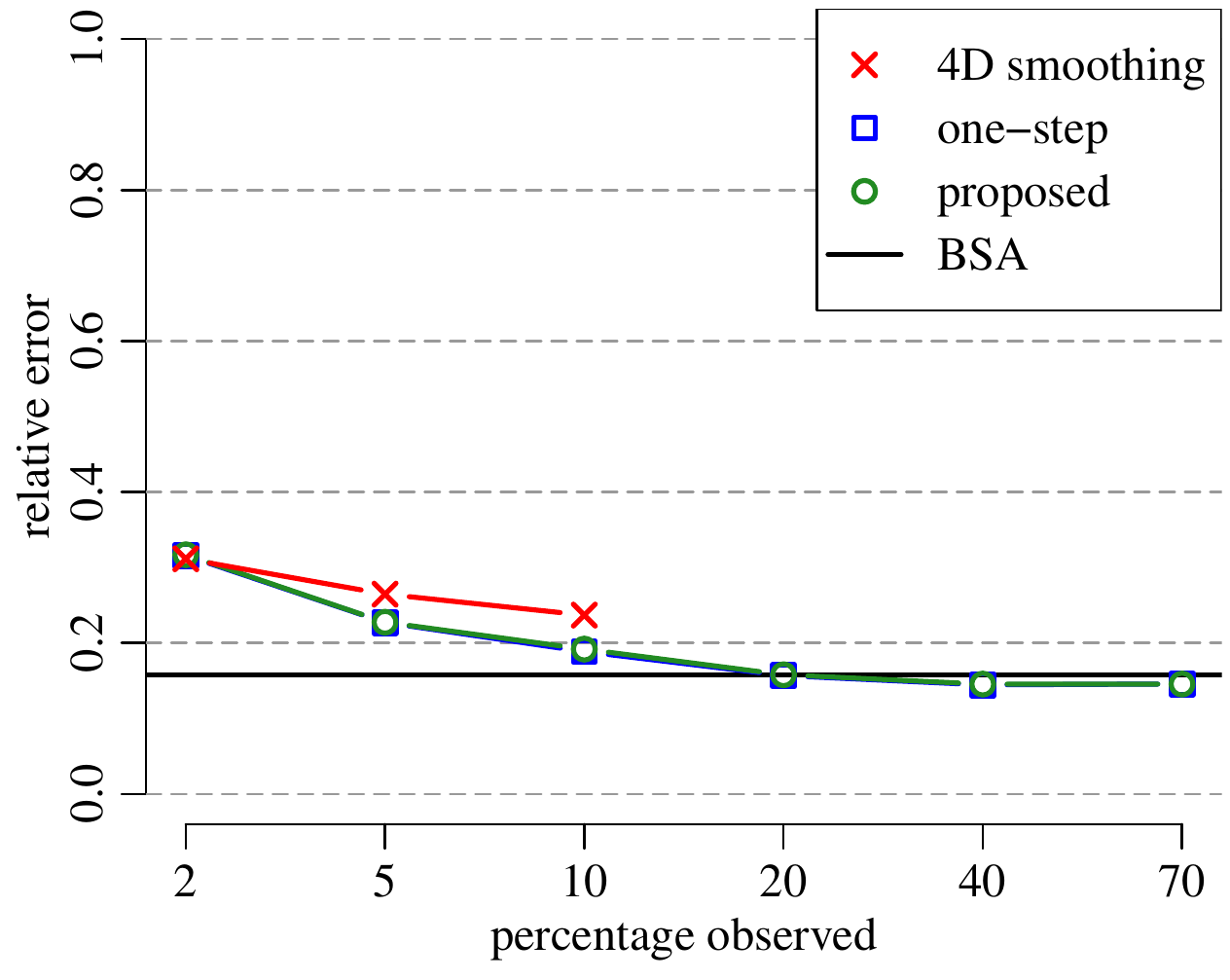} \\\\
   (c) Gneiting -- non-separable, flat && (d) Fourier-Legendre -- non-separable, wiggly \\
   \includegraphics[width=0.45\textwidth]{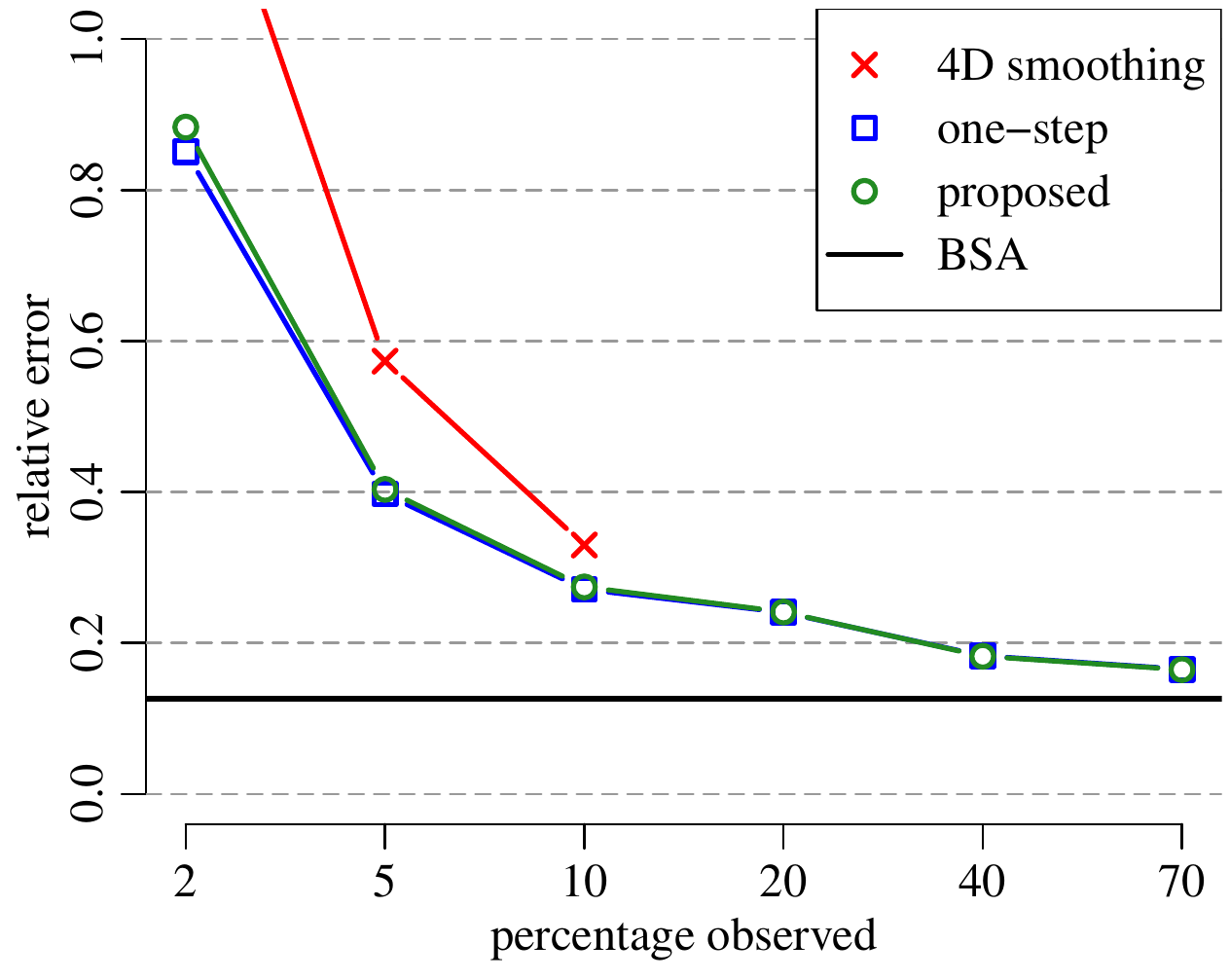} &&
   \includegraphics[width=0.45\textwidth]{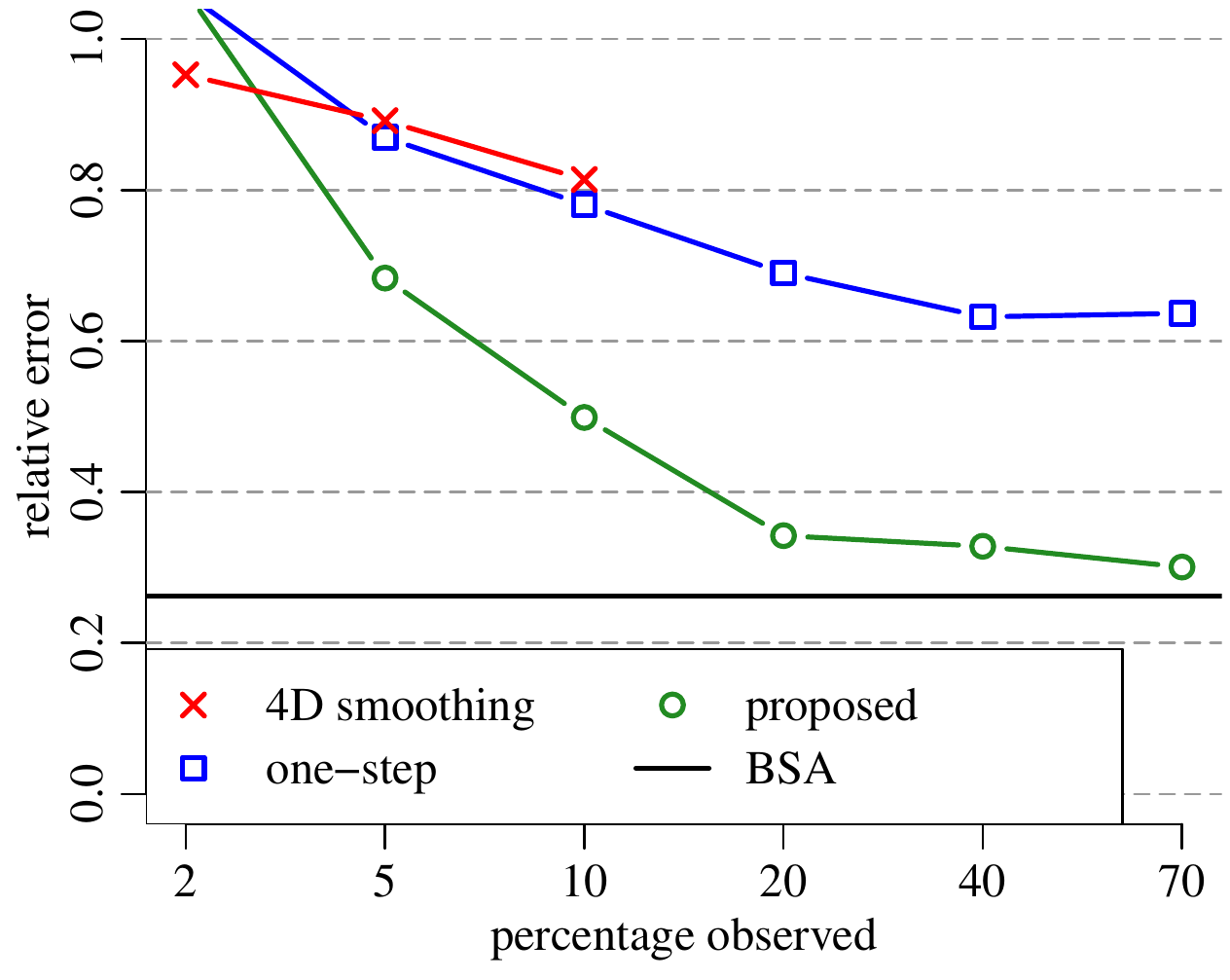}
   \end{tabular}  
   \caption{Relative estimation errors depending on percentages of the surfaces observed $p$ for 4 ground truth covariance choices (a)-(d) and 4 methods compared. BSA provides a baseline, having access to full surfaces and hence not depending on $p$. For 4D smoothing, only results for small $p$ are reported.}
    \label{fig:simulation_results} 
\end{figure}

\textbf{Error components.} The reported estimation errors can be thought of having four components: (i) asymptotic bias, which is zero if the true covariance is separable, i.e. in cases (a) and (b); (ii) error due to finite number of samples ($N=100$); (iii) error due to sparse observations (i.e. not observing the full surfaces); and (iv) noise contribution. BSA errors are always free of the latter two, providing a baseline. The effect of not observing the full surfaces is displayed for different values of $p$. Finally, the noise contamination prevents smoothing approaches to reach the performance of BSA even for $p$ large. Although our methodology explicitly handles noise, the finite sample performance is better with noise-free data, which only BSA has access to.

\textbf{Results.} There is a number of comments to be made about the results in Figure \ref{fig:simulation_results}:
\begin{enumerate}
\item In the setups where the covariance is flat, i.e. (b) and (c), the one-step and the proposed approaches work the same, and 4D smoothing also works relatively well. These two setups are simple in a sense, because information can be borrowed quite efficiently via smoothing, regardless of whether the truth is separable or not. Still, the proposed approach utilizing separability does not perform worse than 4D smoothing even in the non-separable case (c), having the advantage of being much faster to obtain. For $p=10$ the proposed estimator takes only a couple of seconds while 4D smoothing takes about 40 minutes even at this relatively small size of data.
\item When the true covariance is wiggly, the proposed methodology clearly outperforms 4D smoothing, both for the separable truth (a) and the non-separable truth (d). The reason is that smoothing procedures are not very efficient in this case, and borrowing strength via separability is imperative.
\item The reason why error curves for 4D smoothing are only calculated up to $p=10$ is the computational cost of smoothing in higher dimensions. In fact, performing 4D smoothing when $p=10$ took more time than calculationg of all the remaining results combined (cross-validations included). The runtimes are reported in Section \ref{sec:additional_simulations}.
\end{enumerate}

\begin{remark}\label{rem:cross-validation}
We used a standard cross-validation strategy to choose bandwidths for the separable model. However, for 4D smoothing, cross validation is not feasible -- using it would increase the total runtime of our simulation study to over one thousand CPU \emph{days}. Therefore, we chose bandwidths for 4D smoothing based on those cross-validated for a separable model. This approach is taken throughout the paper. See Section \ref{sec:implementation_details} for details and Section \ref{sec:additional_simulations} for an empirical demonstration of the effectiveness of this choice.
\end{remark}

\section{Data Analysis: Implied Volatility Surfaces}\label{sec:data_analysis}

A \textit{European call option} is a contract granting its holder the right, but not the obligation, to buy an underlying asset, for example a stock, for an agreed-upon strike price at a defined expiration time.
Finding a model and deriving a pricing formula for the fair price of a European call option was a milestone problem in quantitative finance and stochastic calculus.  \cite{black1973pricing,merton1973theory} solved this problem and under the so-called the \textit{Black-Scholes-Merton} model they shown that the fair price of the European call option on a non-divident paying asset is given by the \textit{Black-Scholes formula} \citep{hull2006options}:
\begin{equation}\label{eq:Black_Scholes}
C_t^{BS}(m,\tau,\sigma_S) = S_t F_{N(0,1)}(d_1) - \kappa e^{-rt} F_{N(0,1)}(d_2),   
\end{equation}
$$ d_1 = \frac{-\log m + \tau ( r + \sigma_S^2/2 )}{\sigma_S \sqrt{t}}, \qquad
d_2 = \frac{-\log m + \tau(r-\sigma_S^2/2) }{\sigma_S\sqrt{t}},$$
where $m=\kappa/S_t$ is the \textit{moneyness} defined as the ratio of the strike  $\kappa$ and the current underlying asset price $S_t$ at the current time $t$,
$\tau=T-t$ denotes the time to expiration,
$\sigma_S$ is the volatility parameter in the Black-Scholes-Merton model,
$r$ is the risk-free interest rate,
and $F_{N(0,1)}(\cdot)$ denotes the cumulative distribution function of the standard normal distribution.
The only unknown quantity among the inputs in \eqref{eq:Black_Scholes} is the volatility $\sigma_S$.

Besides calculating the fair price of an option given the (estimated/realised) volatility, the Black-Scholes formula \eqref{eq:Black_Scholes}
can be used  in reverse: having observed the market price of the option, denoted as $C_t^*(m,\tau)$, find the value of $\sigma_t^{BS}(m,\tau)$ that solves the equation
$$ C_t^{BS}( m,\tau, \sigma_t^{IV}(m,\tau) ) = C_t^*(m,\tau) .$$
It can be shown that such value $\sigma_t^{IV}(m,\tau)>0$, called the \textit{implied volatility}, exists uniquely for each triplet of $m>0$, $\tau>0$, and $C_t^*(m,\tau)>0$.
Now, if the market indeed followed the Black-Scholes-Merton model and the investors were rational, the implied volatility $\sigma_t^{IV}(m,\tau)$ for various $m$ and $\tau$ would be constant.
However, this is not true for real market data pointing to the shortcomings of the Black-Scholes-Merton model.
Despite these shortcomings, the Black-Scholes formula \eqref{eq:Black_Scholes} is widely used for \textit{transforming} the observed option prices into an ensemble of implied volatilities in a bijective manner.
The advantage of considering the implied volatilities as opposed to the market prices of the options is that the implied volatility surfaces tend to be smoother and comparable across assets. Thus we can take advantage of the functional data analysis framework.

In contrast to the European options, an \textit{American  call option} grants the right to buy the underlying asset anytime until the expiration time $T$. The pricing of American options on possibly dividend paying stocks is more complicated because the pricing involves the optimal stopping problem. In general, no closed form solution exists and numerical algorithms are required \citep{cox1979option}. Likewise, the observed market option prices can be transformed into implied volatilities.

In this section we consider the options data offered by DeltaNeutral \citep{DeltaNeutral}.
This free data set contains the end-of-day prices as well as the calculated implied volatilities for options on U.S. Equities markets. The data covers the period from January 2003 until April 2019 but limits each month to contain the daily options data on only one symbol (a stock or an index). The currently included symbol changes every month and the options on some of the symbols are American while some are European. For each month we pick randomly only one trading day with the data on the currently available symbol and discard the other trading days. Therefore the sample we analyse contains 196 snapshots with option prices and implied volatilities. We discard the non-liquid options and consider the contracts with the log-moneyness $\log m=\log(K/S_t) \in [-0.5, 0.5]$ and the time to expiration $\tau=T-t \in [14, 365]$ (in days). Moreover, we take the logarithm of the implied volatilities to transform them from the domain $(0,\infty)$ onto the real line.
%We focus on the options with the log-moneyness around zero and with the time to expiration less than a year because those options are the most traded. We exclude the options with the expiration less than 2 weeks because they are highly volatile.
To reduce computational costs we round  the log-moneyness and the time to expiration to fall on a common $50\times 50$ grid, c.f. Section~\ref{sec:computational_details}.

Figure~\ref{fig:iv_snapshots_mean} shows two observations in our samples. The snapshot of Qualcomm Inc (QCOM) features a cummulation  of observation at the short expiration. Here, it happens five times that two raw observations fall in the same pixel on the common $50\times 50$ grid. In these few cases we calculate the average of the two observations in each pair. Due to smoothness, the option prices (and hence the implied volatilities) attain very similar values and thus this rounding and averaging does not change the conclusions of our analysis. We have observed this fact by fitting the model on finer grids while the estimates remained similar.

\begin{figure}[!t]
   \centering
   \begin{tabular}{ccc}
   (a) Dell on 01/19/2006 & (b) Qualcomm Inc on 02/07/2018 & (c) mean surface  \\\\
   \includegraphics[width=0.3\textwidth]{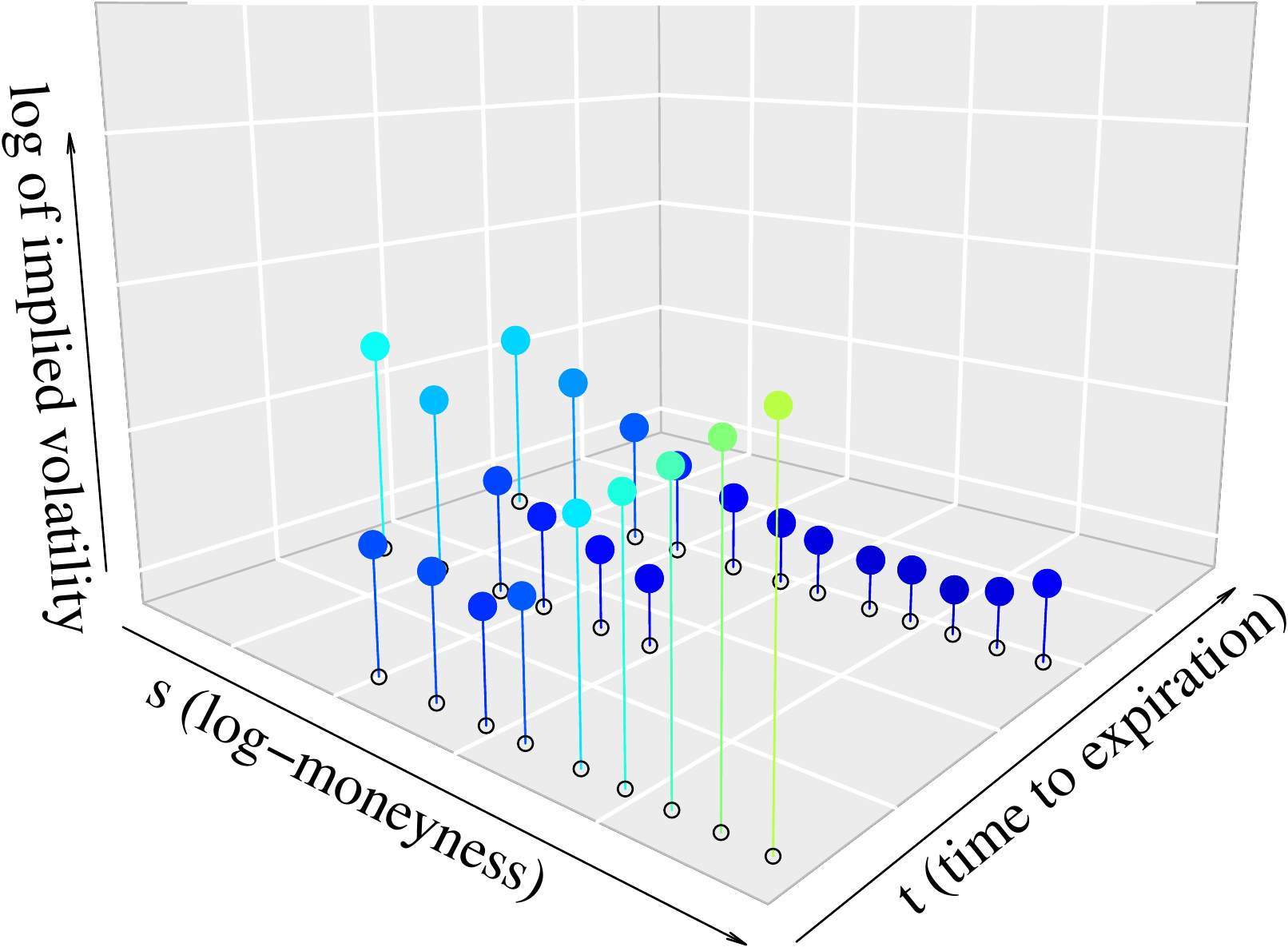}&
   \includegraphics[width=0.3\textwidth]{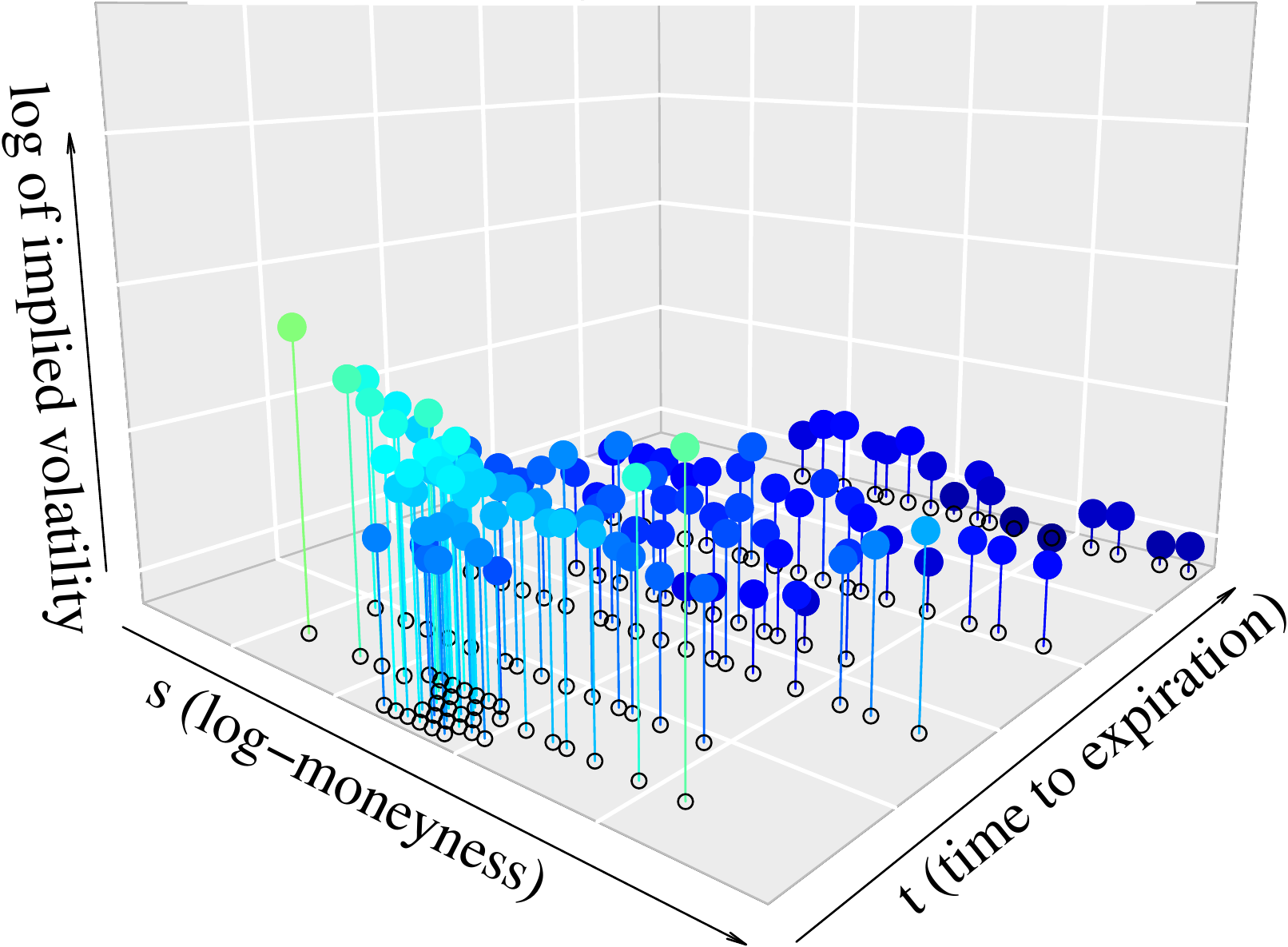}&
   \includegraphics[width=0.3\textwidth]{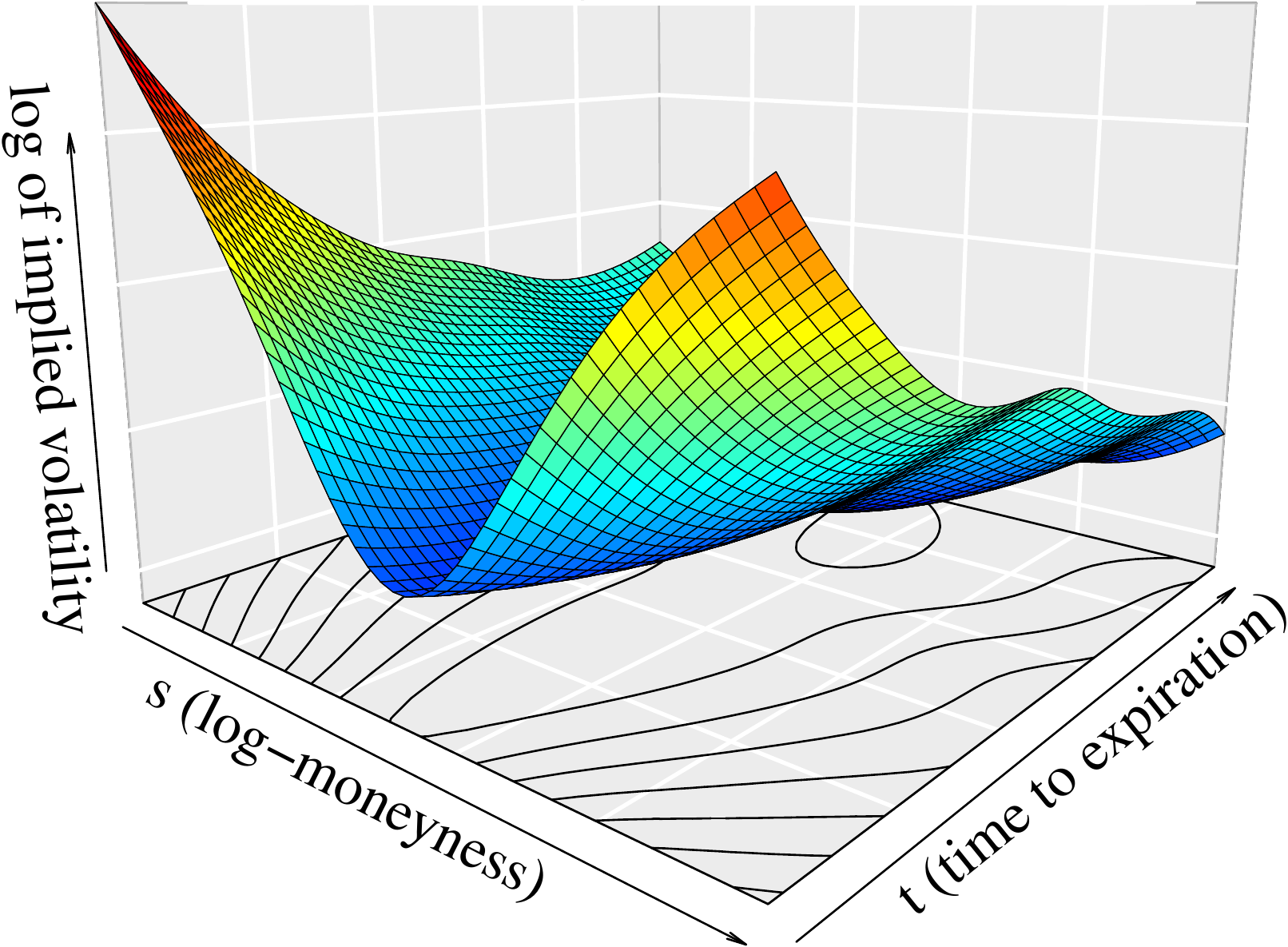} \\\\
   \multicolumn{3}{c}{
   \includegraphics[width=0.9\textwidth]{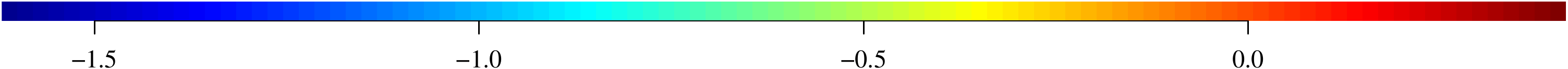}}
   \end{tabular}  
   \caption{Two sample snapshots of the considered log implied volatility surfaces corresponding to the call options on the stocks of Dell Technologies Inc on 01/19/2006 \emph{(\textbf{a})} and Qualcomm Inc on 02/07/2018 \emph{(\textbf{b})}, and the mean surface of the implied volatility gained from pooling all the data together \emph{(\textbf{c})}. }
    \label{fig:iv_snapshots_mean} 
\end{figure}

The estimated mean surface is displayed on the right-hand side of Figure~\ref{fig:iv_snapshots_mean}. The mean surface captures the typical feature of the implied volatility surfaces: the volatility smile \citep{hull2006options}. The implied volatility is typically greater for the options with moneyness away from 1, while this aspect is more significant for shorter times to expiration. 

Figure~\ref{fig:iv_a_b} displays the estimates of the separable covariance components by our methodology presented in Section~\ref{sec:estimation}. The moneyness component demonstrates the highest marginal variability at the center of the covariance surface, meaning that the log implied volatility oscillates the most for the options with the log-moneyness around 0 (i.e. moneyness 1). The marginal variance is lower as the log-moneyness departs from 0.
The eigendecomposition plot of the moneyness covariance kernel shows that the most of variability is explained by a nearly constant function with a small bump at the log-moneyness 0. The second leading eigenfunction adjusts the peak at the log-moneyness around 0 to a greater extent than the first eigenfunction.
The covariance kernel corresponding to the time to expiration variable is smoother and demonstrates slightly higher marginal variability at shorter expiration. This phenomenon is well known for implied volatility \citep{hull2006options}. The eigendecomposition of this covariance kernel indicates that the log implied volatility variation is mostly driven by the constant function while the second leading eigenfunction adjusts the slope of the surface for varying time to expiration.

\begin{figure}[!t]
   \centering
   \begin{tabular}{cc}
   \includegraphics[width=0.47\textwidth]{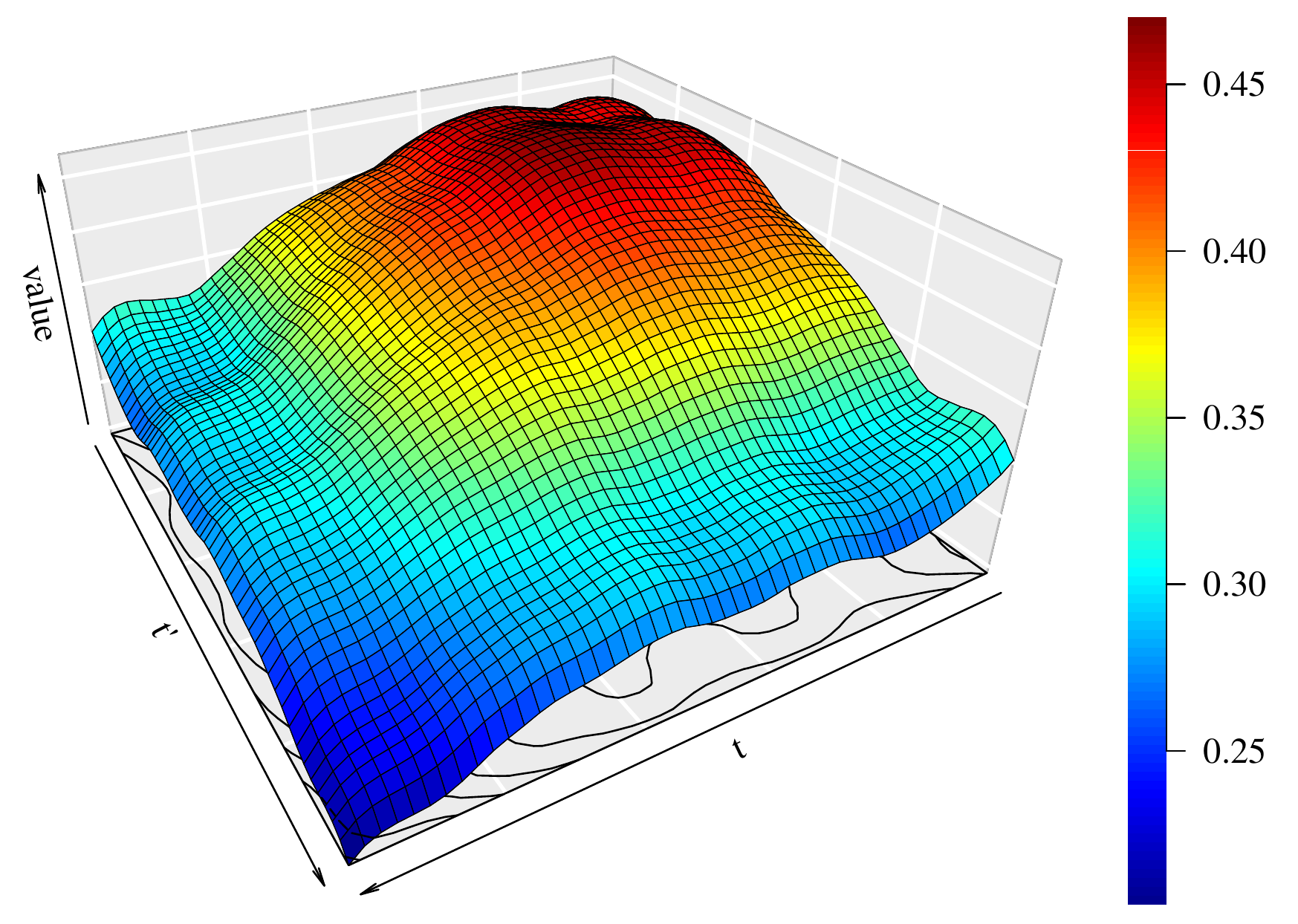}&
   \includegraphics[width=0.47\textwidth]{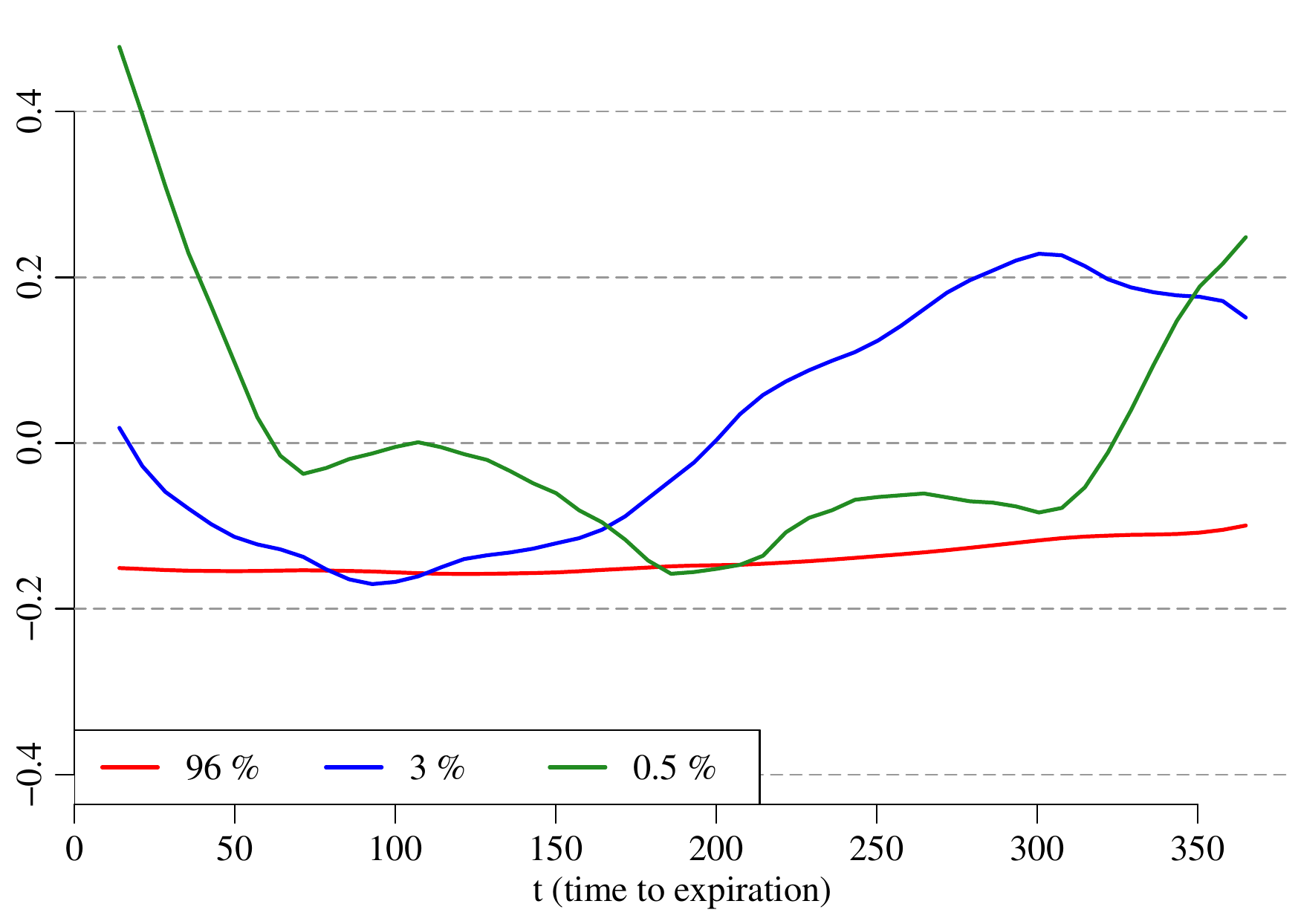}\\\\
   \includegraphics[width=0.47\textwidth]{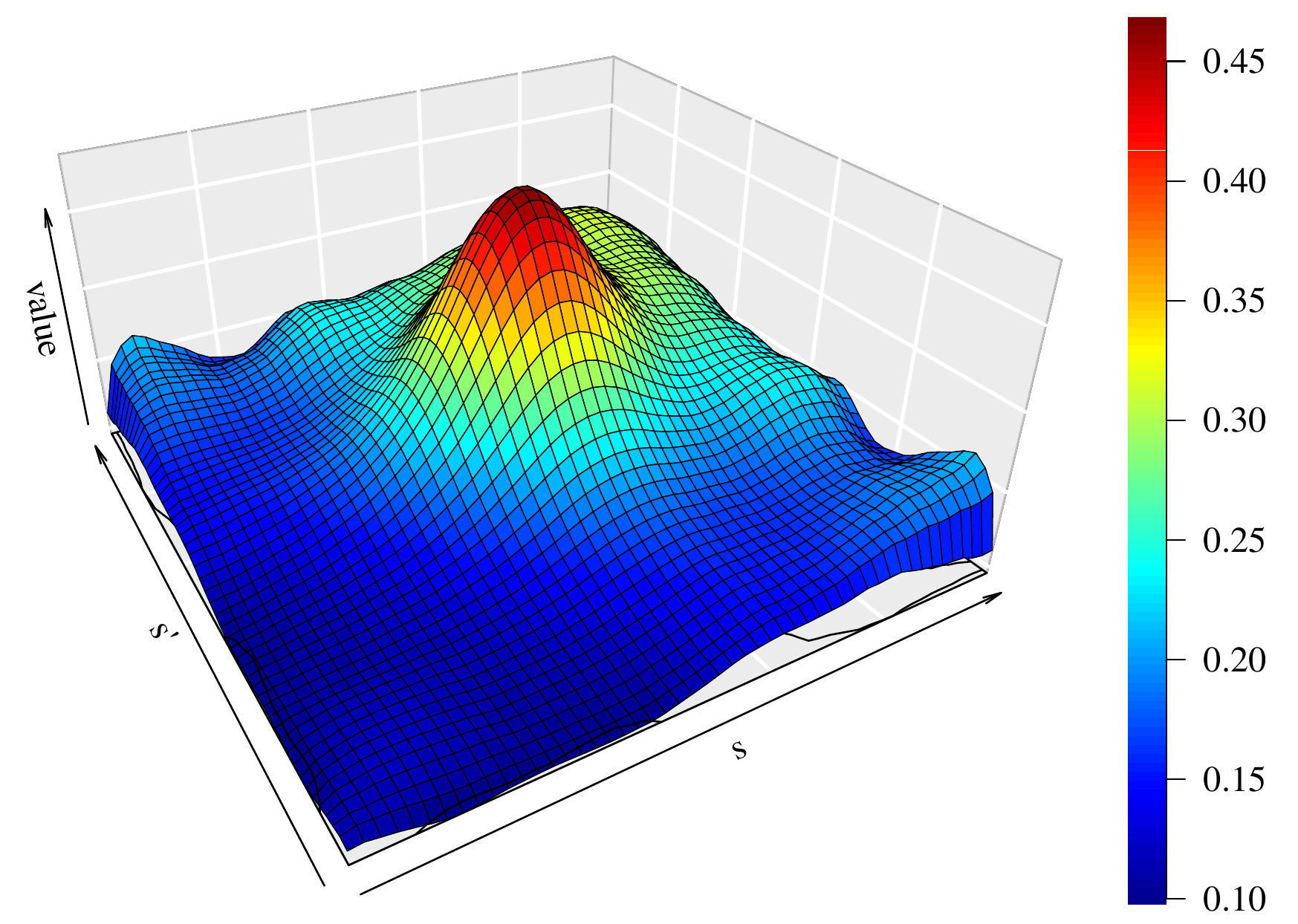}&
   \includegraphics[width=0.47\textwidth]{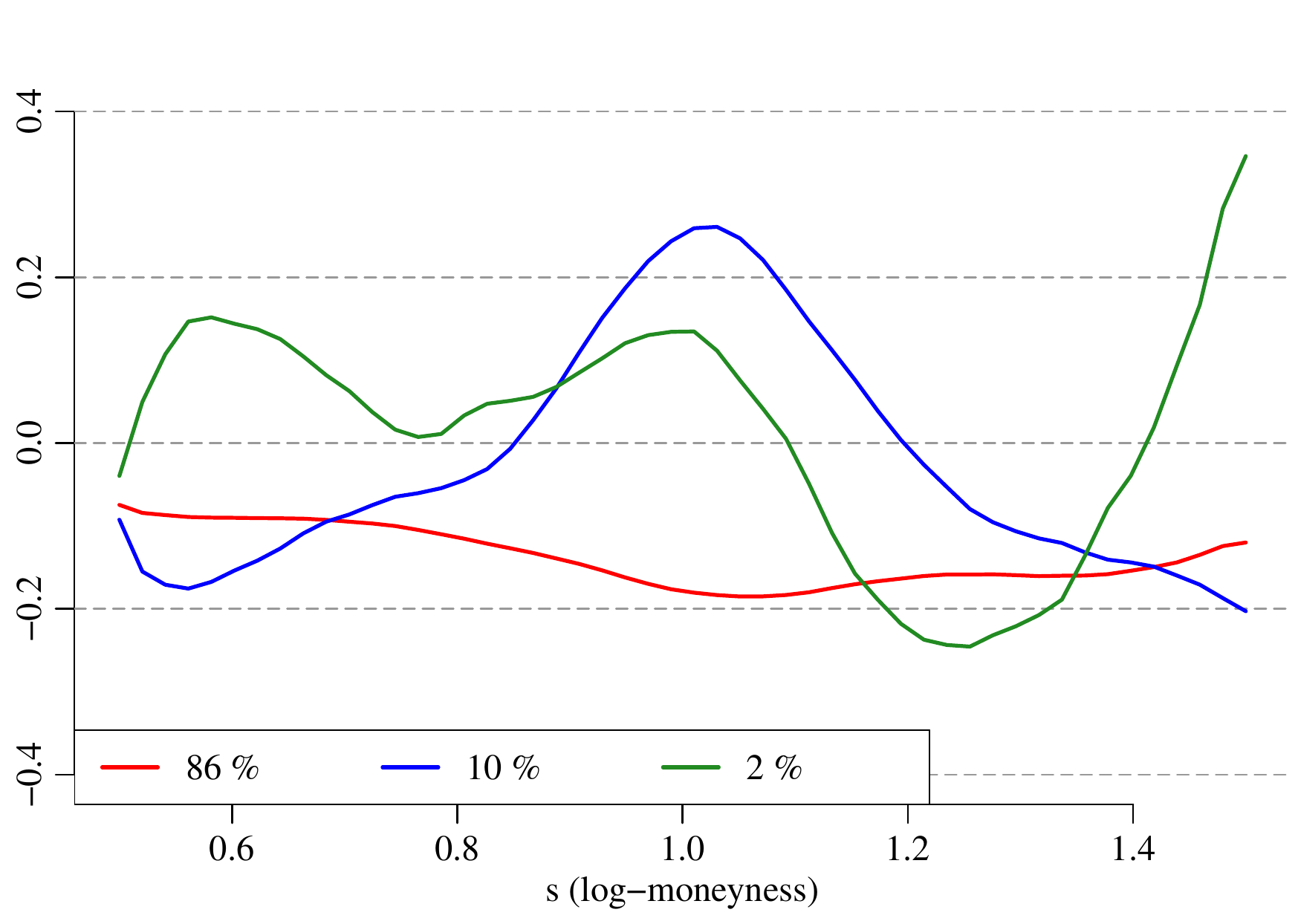}
   \end{tabular}  
   \caption{\textbf{Top-left:} The estimated covariance kernel $\widehat{a}=\widehat{a}(t,t')$ corresponding to the time to expiration variable. \textbf{Top-right:} The three leading eigenfunctions of the spectral decomposition of the covariance kernel $\widehat{a}=\widehat{a}(t,t')$. \textbf{Bottom-left and botton-right:} The same as above but for the estimated covariance kernel $\widehat{b}=\widehat{b}(s,s')$ corresponding to the log-moneyness variable.}
    \label{fig:iv_a_b} 
\end{figure}

Figure~\ref{fig:iv_prediction_bands} demonstrates our prediction techniques presented in Section~\ref{sec:prediction} together with the 95 \% simultaneous confidence band. We recall that the confidence band aims to capture the latent smooth random surface itself, while our raw observations are modelled by adding an error term. Therefore, the raw data are not guaranteed to be covered in the confidence band.

\begin{figure}[!t]
   \centering
   \begin{tabular}{cc}
   \includegraphics[height=0.375\textwidth]{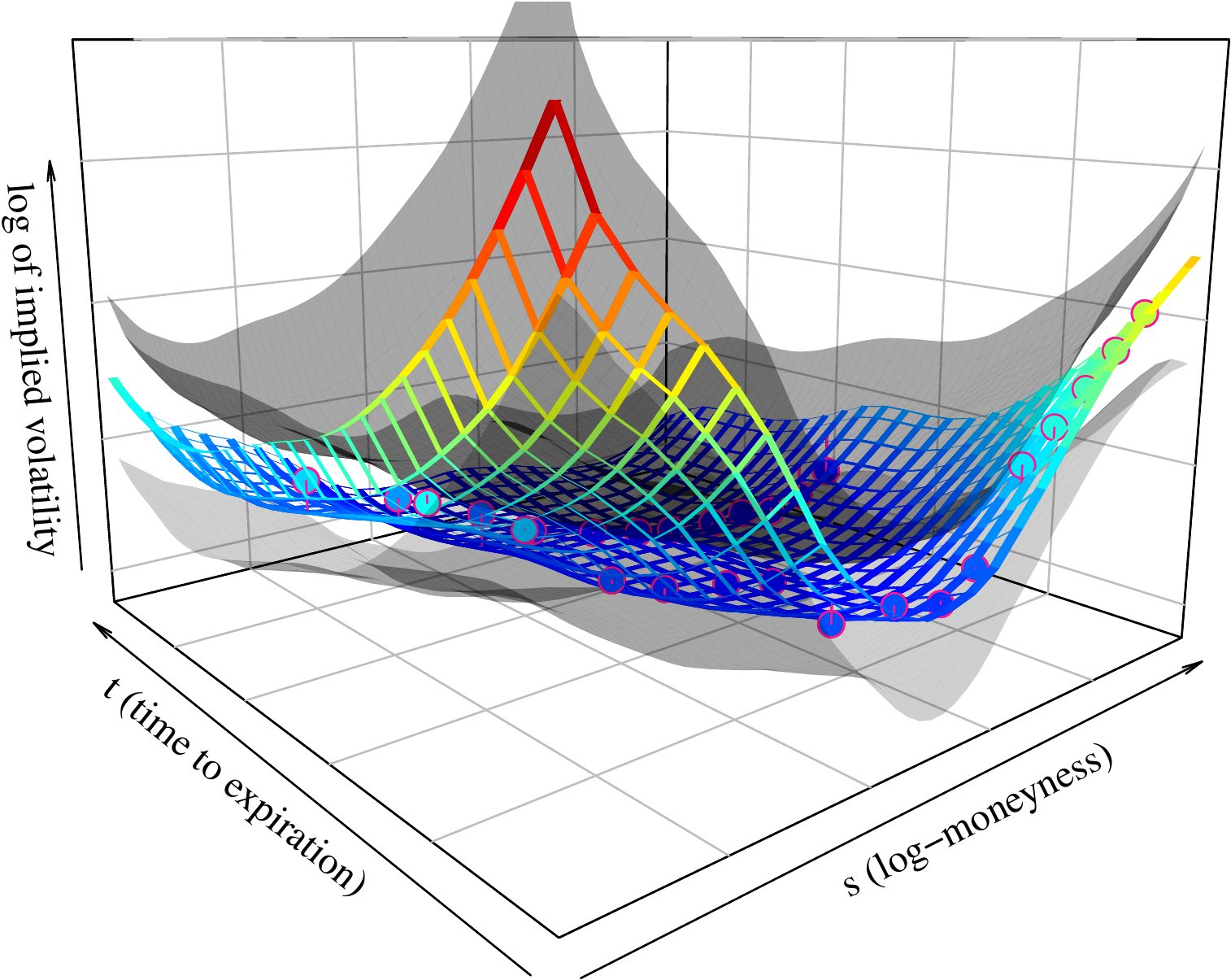}&
   \includegraphics[height=0.375\textwidth]{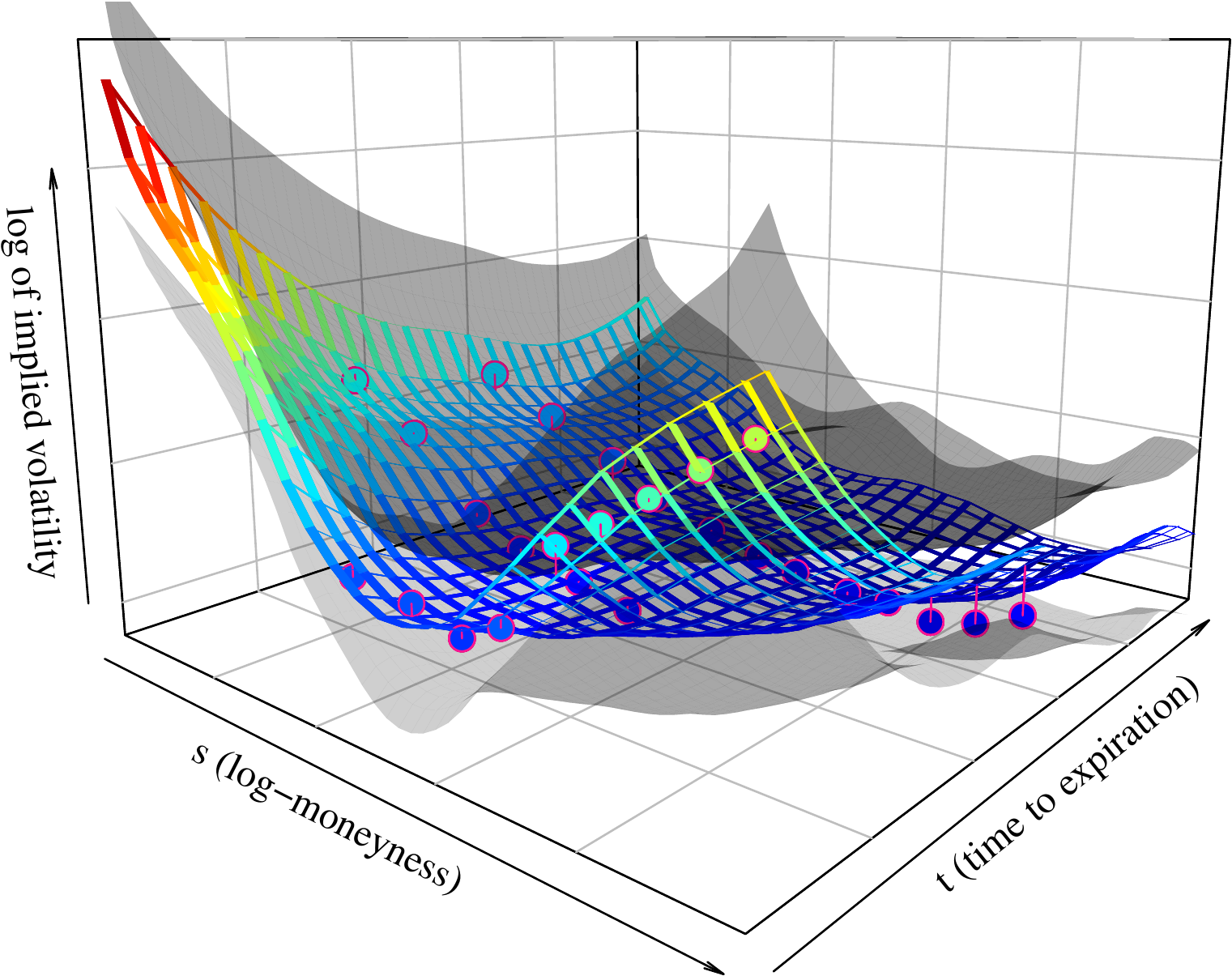} \\\\
   \multicolumn{2}{l}{\hspace{0.5cm}
   \includegraphics[width=0.9\textwidth]{plots/colkey1.pdf}}
   \end{tabular}  
  \caption{Two views on prediction based on the call options written on the stock of Dell Technologies Inc on 01/19/2006. The circles depict the available sparse observations, the ribbons depict the predicted latent surface by the method of Section~\ref{sec:prediction}, where the covariance structure was assumed separable, and finally the transparent gray surfaces depict the 95 \% simultaneous confidence band for the latent log implied volatility surface.
  }
    \label{fig:iv_prediction_bands}  
\end{figure}

\subsection{Quantitative Comparison}
\label{subsec:quantitative_comparison}

The prediction method outline in Section~\ref{sec:prediction} requires as an input the pairwise covariances regardless whether they have been estimated by the separable estimator $\widehat{a}(t,t')\widehat{b}(s,s')$ or the 4D smoother $\widehat{c}(t,s,t',s')$.
As the benchmark for our comparison, we choose the locally linear kernel smoother \citep{fan1996} applied individually for each surface as such smoothers constitute a usual pre-processing step \citep{cont2002}. We will refer to this predictor as \textit{pre-smoothing}.
In this section, we demonstrate that the predictive performance is comparable for  both covariance estimator strategies (the separable and the 4D smoother), and that both of these approaches are superior to pre-smoothing. Moreover, the separable smoother is substantially faster than the 4D smoother.

We compare the prediction error by performing a 10-fold cross-validation, where the covariance structure is fitted always on varying 90 \% of the surfaces, with the remaining 10 \% used for out-of-sample prediction. In the set that is held out for prediction, we select some of the sparse observations and predict them based on the remaining observations on that surface. We use the following hold-out patterns:
\begin{enumerate}[label=({\alph*})]
    \item \textit{Leave one chain out.} Since the options are quotes always for a range of strikes, they constitute features known as option chains (c.f. Figure~\ref{fig:iv_snapshots_mean}) where multiple option prices (or equivalently implied volatilities) are available for a fixed time to expiration. For those surfaces that include at least two such chains, we remove gradually each chain and predict it based on the other chains. Therefore the number of prediction tasks performed on a single surface is equal to the number of chains observed per that surface.
    \item \textit{Predict in-the-money.} Predict implied volatilities for below-average moneyness (i.e. moneyness $m \leq 1$)  based on the out-of-the-money observations (moneyness $m>1$).
    \item \textit{Predict out-of-the-money.} Predict implied volatilities for above-average moneyness (i.e. moneyness $m \geq 1$) based on the in-the-money observations (moneyness $m<1$).
    \item \textit{Predict short maturities.} Predict the implied volatility for options with the time to maturity $\tau < 183 $ [days] based on the implied volatility of the options with the time to maturity $\tau \geq 183 $ [days].
    \item \textit{Predict long maturities.} Predict the implied volatility for options with the time to maturity $\tau > 183 $ [days] based on the implied volatility of the options with the time to maturity $\tau \leq 183 $ [days].
\end{enumerate}
All the prediction strategies are performed only for those surfaces where both the discarded part and the kept part are non-empty.
We measure the prediction error on surface with the index $n$ (in the test partition within the $K$-fold cross-validation) by the following root mean square error criterion, relative to the pre-smoothing benchmark:
\begin{equation}
\label{eq:quantitative_comparison_RMSE}
RMSE^{\text{method}}(n) = 
\sqrt{
\frac{
\sum_{m \in M_n^{\text{discarded}} } \left( (\widehat\Pi^{\text{method}}( X(t_{nm},s_{nm}) | \mathbb{Y}_n^{\text{kept}})) -  Y_{nm}  \right)^2
}{
\sum_{m \in M_n^{\text{discarded}} } \left( (\widehat\Pi^{\text{pre-smooth}}( X(t_{nm},s_{nm}) | \mathbb{Y}_n^{\text{kept}})) -  Y_{nm}  \right)^2
}}
\end{equation}
where $Y_{nm}, m=1,\dots,M_n$ are the implied volatility observations on the $n$-th surface,
$M_n^{\text{discarded}} \subset \{1,\dots,M_n\}$ denotes the set of observations' indexes discarded for the $n$-th surface,
$\mathbb{Y}_n^{\text{kept}}$ are the vectorized implied volatility observations that were kept to be conditioned on.
The predictor $\Pi^{\text{pre-smooth}}$ denotes the pre-smoothing based on the observations $\mathbb{Y}_n^{\text{kept}})$.
The predictors $\Pi^{\text{method}}$ for method being either the separable smoother or the 4D smoother constitute the proposed predictors in this article with the covariance structure estimated by either of the two smoothers. These predictors are always trained only on the training partition (90 \% of the surfaces) within the $10$-fold cross-validation scheme. Note that this out-of-sample comparison is adversarial for the proposed approach, because when predicting a fixed surface, the measurements on that surface (and another 10\% of measurements total) are not used for the mean and covariance estimation. In practice, we naturally utilize all available information. However, for the hold-out comparison study here, that would require frequent re-fitting of the covariance, which would not be computationally feasible, in particular for 4D smoothing.

\begin{figure}[!t]
   \centering
   \includegraphics[width=0.9\textwidth]{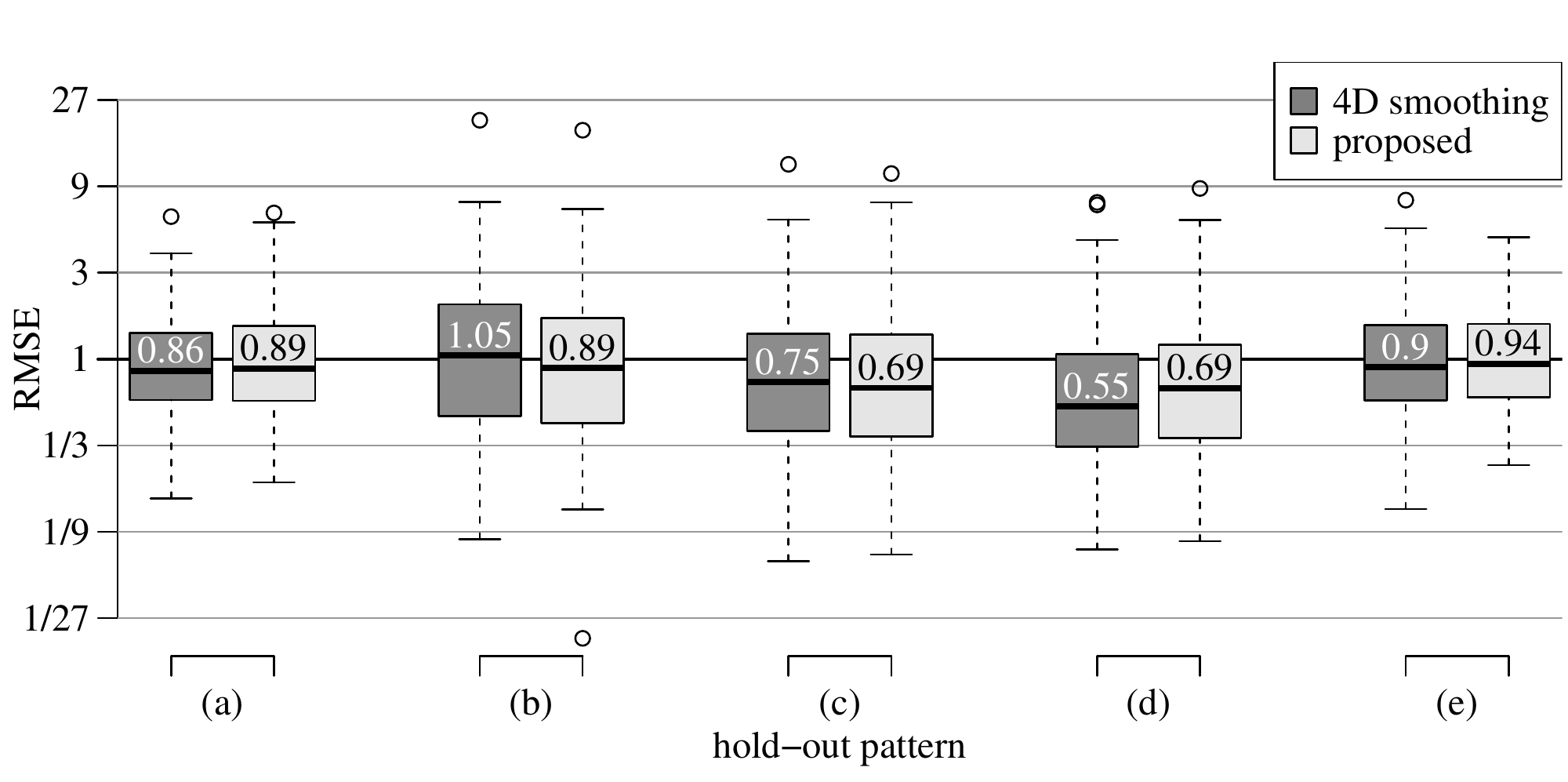}
   \caption{Boxplots of root mean square errors relative to the pre-smoothing benchmark \eqref{eq:quantitative_comparison_RMSE} for prediction method of Section \ref{sec:prediction} with the covariance estimated by 4D smoothing or the proposed separable approach, and different hold-out patterns: (a) leave one chain out; (b) predict in-the-money; (c) predict out-of-the-money; (d) predict short maturities; and (e) predict long-maturities. Numbers inside the boxes provide numerical values of the median. For a given method, RMSE value 3 means that the given method is 3-times worse than the benchmark, while RMSE value 1/3 corresponds to 3-fold improvement.}
    \label{fig:comparison} 
\end{figure}

Figure~\ref{fig:comparison} presents the results under the five hold-out patterns in form of boxplots created from the relative errors \eqref{eq:quantitative_comparison_RMSE}. We see that the prediction errors based on estimated covariances, be it the separable smoother or the 4D smoother, are typically smaller than the pre-smothing benchmark with the only exception of the 4D smoother in the hold-out pattern (b). The predictive performances of the separable and the 4D smoother are comparable, but they differ a lot in terms of runtime.
It typically takes 30 seconds to calculate the separable smoother (including a cross-validation based selection of the smoothing bandwidths), while the 4D smoother takes around 3 hours. The latter runtime is moreover without considering any automatic selection of the bandwidths, because such would be computationally infeasible. Hence we use the bandwidths selected by the separable model, as described in Remark \ref{rem:cross-validation}. 
The calculations are performed on a quite coarse grid of size $20\times 20$. The calculations on a dense grid, such as $50 \times 50$ used in the qualitative analysis in Section~\ref{sec:data_analysis} are not feasible for the 4D smoother.

Therefore we conclude that the separable smoother approach enjoys a better predictive performance than the pre-smoothing benchmark and -- while having having similar predictive performance as the predictor based on 4D smoothing -- is computationally much faster than the said competitor. In fact, it requires two-dimensional smoothers only, just as the pre-smoothing benchmark.

\section{Conclusions and Future Directions}

In practice, covariances are often non-separable \cite{gneiting2006,aston2017,bagchi2017,dette2020,pigoli2018,rougier2017}, and assuming separability induces a bias. The variance stemming from sparse measurements and noise contamination is, however, often of a of larger magnitude, thus sanctioning separability as a means to achieve a better bias-variance trade-off. Moreover, separability entails faster computation and lower storage requirements. As demonstrated above, these advantages are pronounced in the sparse regime.

%When data are observed quite sparsely, our approach outperforms a general covariance simply due to the curse of dimensionality (our approach borrows information more efficiently). Once data are observed relatively densely and many replications are available, our approach is \emph{not} expected to outperform the smoother of the empirical covariance per se in terms of estimator's quality. However, the computational gains of our estimator compared to the empirical covariance smoother are huge in this case. 

When data are observed fully and no smoothing is used (i.e. simple averages are calculated at the grid points instead), calculation of the estimator takes about 0.03 seconds in case of separability and 0.12 seconds in general, with our grid size $20$ and sample size $N=100$. With larger grid size, the acceleration gained by assuming separability naturally magnifies. But still, separability is assumed mostly to save memory rather than time, in the fully observed case \cite{aston2017}. On the other hand, with sparsely observed data and kernel smoothing deployed, the speed-up factor stemming from separability can be of the order of hundreds already with rather small grid sizes. The point we advocate is that -- with sparse observations, kernel smoothing deployed, and compared to observations on a grid -- the computational savings offered by separability are the same in terms of memory, but are much more profound in terms of runtime. This is due to the extra costs associated with smoothing.

Our analysis of implied volatility surfaces provides some insights into the statistical dependencies of such sparsely observed data, and our quantitative comparison demonstrated the prediction performance of our prediction method over the pre-smoothing benchmark.
%We conclude our data analysis with a comment on the assumption that all implied volatility surfaces come from a common population. The implied volatility, as opposed to the option prices, represents a dimensionless quantity \citep{hull2006options} related to the volatility (variance) of the Black-Scholes model.
We have formed our data set collecting options across various symbols and timestamps with the reasoning that the volatility across these come also from one population, which is debatable. This study should be seen rather as a proof-of-concept on how to ``borrow strength'' across the data set offered by DeltaNeutral \citep{DeltaNeutral}, and how this approach can be used to decrease the prediction error.

Having said that, we believe that our methodology can provide even better results when considering a more homogeneous population such as the time series of the implied volatility surfaces related to a single fixed symbol/asset. In this case, pre-smoothing is usually required \cite{cont2002,kearney2018} for forecasting of such time series. Our methodology could avoid the pre-smoothing step in this case by predicting the surfaces while borrowing the information across the entire data set. Furthermore, our methodology could be easily tailored to predicting principal components scores by conditional expectation (similarly to \cite{yao2005}), which could be used for forecasting by a vector autoregression.

Combining separable covariances and time series of sparsely observed surfaces hints at other directions of future work. For example, \citet{rubin2020} showed how to estimate the spectral density operator non-parametrically from sparsely observed functional time series. Estimating separable spectral density operator for sparsely observed surface-valued time series seems to be within reach and is likely to provide predictions that benefit from the information across time, and thus reducing the prediction error even more.

\appendix

\section{Computational Details}\label{sec:computational_details}

\subsection{Weighting Scheme}\label{sec:connection_full_obs}

In this section, we provide a heuristic justification for the quadratic choice of weights in the smoothers for the estimation of the covariance kernels $a(\cdot,\cdot)$ and $b(\cdot,\cdot)$, such as \eqref{eq:set_for_B} or \eqref{eq:set_for_A}. Then, we provide a more precise justification, showing that the quadratic choice of weights corresponds to the optimal choice, when data are observed densely.

We begin with the heuristic justification. The quadratic weights can be motivated by the connection to \textit{weighted least squares}.
We recall that weighted least squares are used for  linear regression models where the model errors are not necessarily i.i.d. Their covariance matrix is assumed to be a diagonal matrix known up to a multiplicative constant:
\begin{equation}\label{eq:generalized_least_squares}
\mathbf{y} = \mathbf{X} \beta + \varepsilon, \qquad \E [ \varepsilon \vert \mathbf{X} ] = 0, \qquad \var( \varepsilon \vert \mathbf{X} ) = \sigma_\varepsilon^2 \diag(\mathbf{v}),\qquad \sigma_\varepsilon^2>0,     
\end{equation}
where $\mathbf{y} = (y_1,\dots,y_I)\transpose$ is the response, $\mathbf{X} = (\mathbf{x}_1\transpose,\dots,\mathbf{x}_I\transpose)\transpose$ is the model matrix, and $\diag(\mathbf{v})$ denotes the diagonal matrix with the known vector $\mathbf{v}=(v_1,\dots,v_I)\transpose$ on its diagonal.
The regression coefficients $\beta$ in the model \eqref{eq:generalized_least_squares} are estimated by the weighted least squares:
\begin{equation}\label{eq:weighted_least_squares}
\widehat\beta = \argmin_{\beta} \sum_{i=1}^I w_i \left( y_i - \mathbf{x}_i \beta \right)^2,\qquad\text{where}\quad w_i=\frac{1}{v_i},\,\,i=1,\dots,I.    
\end{equation}

Consider the surface smoother of
\[
\left\{ \left(s_{nm},s_{nm'},\frac{G_{nmm'}}{\alpha(t_{nm},t_{nm'})}\right) \; \Bigg| \; m,m'=1,\ldots,M_n, \, m \neq m',\; n=1\ldots,N \right\}
\]
where $\alpha(t,t'),\,t,t'\in[0,1]$, is a fixed deterministic kernel.
The kernel smoothing technique we deploy is based on fitting a linear regression locally. In view of model \eqref{eq:generalized_least_squares} we want to assess the variance of the response $G_{nmm'} / \alpha(t_{nm},t_{nm'})$ to improve the estimation procedure:
\begin{equation}\label{eq:variance_G_nmm}
\var\left( \frac{G_{nmm'}}{\alpha(t_{nm},t_{nm'})} \right) =
\frac{1}{\alpha^2(t_{nm},t_{nm'})} \var\left( G_{nmm'} \right).
\end{equation}
The variance of $G_{nmm'}$ is unknown and therefore cannot be used to improve the estimation. Still, we observe in  equation~\eqref{eq:variance_G_nmm} that the variance is multiplied by the reciprocal of $\alpha^2(t_{nm},t_{nm'})$. Therefore, we would define the weights for the weighted least squares \eqref{eq:weighted_least_squares} as $w_i = \alpha^2(t_{nm},t_{nm'})$, to utilize the knowledge we actually have.

Let us now describe the connection between the quadratic weighting scheme and the case of fully observed surfaces. With fully observed surfaces, the separable model can be estimated via the generalized power iteration method, where a single step is given by the partial inner product between the empirical covariance and the previous step \cite{masak2020}. For example, when $b=b(s,s')$ is fixed, one step of the power iteration method is given by \cite[Proposition 1]{masak2020}
\begin{equation}\label{eq:PIP_continuous}
\widehat{a}(t,t') = \int_0^1 \int_0^1 b(s,s') \widehat{c}_N(t,s,t',s') d s d s' \bigg/  \int_0^1 \int_0^1 b^2(s,s') d s d s' ,
\end{equation}
where $\widehat{c}_N$ is the empirical covariance estimator.
In this section, we demonstrate that, with fully observed data and with no smoothing conducted, the estimation methodology of Section \ref{sec:estimation} corresponds to the power iteration step \eqref{eq:PIP_continuous}.

Firstly, assume that $b$ is fixed, and we are using the surface smoother on the set of points \eqref{eq:set_for_A} to obtain $\widehat{a}$. Assume that the $n$-th surface is observed twice at the temporal location $t$, i.e. at two locations $(t,s_1)$ and $(t,s_2)$, and once more in a general location $(t', s')$. Let us denote the raw covariance corresponding to the $n$-th surface and locations $(t,s)$ and $(t',s')$ explicitly by $G_n(t,s,t',s')$. Then, two values are available for the location $(t,t')$ in set \eqref{eq:set_for_A}:
\[
\frac{G_n(t,s_1,t',s')}{b(s_1,s')} \quad \& \quad \frac{G_n(t,s_2,t',s')}{b(s_2,s')}.
\]
The corresponding weights are $b^2(s_1,s')$ and $b^2(s_2,s')$, respectively. If the bandwidth is small enough, and no other observations are available for this location, $\widehat{a}(t,t')$ is calculated as a weighted average:
\[
\left[ b^2(s_1,s') \frac{G_n(t,s_1,t',s')}{b(s_1,s')} + b^2(s_2,s') \frac{G_n(t,s_2,t',s')}{b(s_2,s')} \right] \Big/ \big[ b^2(s_1,s') + b^2(s_2,s') \big]
\]
Also, for the purposes of the surface smoother, using the two points separately with their separate quadratic weights is equivalent to using the weighted average with the weight $b^2(s_1,s') + b^2(s_2,s')$.

When the temporal slice $t$ of the $n$-th surface is observed fully, the weighted averaging can be done continuously:
\[
\widehat{a}(t,t') = \int_0^1 b(s,s') G_n(t,s,t',s') d s \bigg/ \int_0^1 b^2(s,s') d s .
\]
When the temporal slice $t'$ of the $n$-th surface is also observed fully, the weighted average becomes
\[
\widehat{a}(t,t') = \int_0^1\int_0^1 b(s,s') G_n(t,s,t',s') d s d s' \bigg/ \int_0^1\int_0^1 b^2(s,s') d s d s' .
\]
When this is true for all $N$ surfaces, the result is averaged over all the independent realizations, and we arrive directly to \eqref{eq:PIP_continuous}, since $\widehat{c}_N(t,s,t',s') = \frac{1}{N} \sum_{n=1}^N G_n(t,s,t',s')$.

Altogether, our estimation procedure can be thought of (due to the specific weighting scheme used) as a sparse version of the generalized power iteration method of \cite{masak2020}. This link has important computational implications, which are discussed in the following section.

\subsection{Marginalization on a Grid}\label{sec:marginalization}

By marginalization, we mean preparation of the raw covariances for the 2D smoothing step, i.e. charting the raw covariances either in time or in space and weighting them as in formulas \eqref{eq:set_for_A} and \eqref{eq:set_for_B}. In this section, we will show that scatter points can be pooled together during the marginalization process to save computations during the subsequent smoothing step, when the data are observed (or rounded to) a common grid. The associated computational advantages are discussed in the following section.

Assume for the remainder of this section that 
data from the measurement scheme \eqref{eq:measurement_scheme} arrive as matrices $\mathbf{Y}_1, \ldots, \mathbf{Y}_N \in \R^{d_1 \times d_2}$ with only some of their entries known, i.e. most of the entries are missing. The marginal covariance kernels $a=a(t,t')$ and $b=b(s,s')$ are replaced by matrices $\mathbf A \in \R^{d_1 \times d_1}$ and $\mathbf B \in \R^{d_2 \times d_2}$, respectively. We assume again for simplicity that the mean $\mu = \mu(t,s)$ is zero. The raw covariances then form a tensor $\mathbf{G}_n = \mathbf Y_n \otimes \mathbf Y_n \in \R^{d_1 \times d_2 \times d_1 \times d_2}$ with entries $\mathbf{G}_n[i,j,i',j'] = \mathbf Y_n[i,j] \mathbf Y_n[i',j']$.

Again, like in the previous section, assume that $\mathbf B$ is fixed, and we are using the surface smoother on the discrete equivalent to set \eqref{eq:set_for_A}, i.e.
\begin{equation}\label{eq:set_discrete}
\bigg\{ \left(i,i', \frac{\mathbf G_n[i,j,i',j']}{\mathbf B[j,j']} \right) \bigg| \; \mathbf Y_n \text{ observed at } (i,j) \text{ and } (i',j'), \; (i,j) \neq (i',j'), \; n=1,\ldots,N \bigg\} 
\end{equation}
to obtain $\widehat{A}$. Like in the previous section, assume $\mathbf{Y}_n$ was observed at locations at $[i,j_1]$, $[i,j_2]$ and $[i',j']$ where no two locations are the same. As explained in the previous section, it is equivalent for the surface smoother to replace the corresponding two values from \eqref{eq:set_discrete}, i.e.
\[
\frac{\mathbf G_n[i,j_1,i',j']}{\mathbf B[j_1,j']} \quad \& \quad \frac{\mathbf G_n[i,j_2,i',j']}{\mathbf B[j_2,j']}
\]
with weights $\mathbf B[j_1,j']$ and $\mathbf B[j_2,j']$, by a single value
\begin{equation}\label{eq:pooled_point}
\big( \mathbf B[j_1,j'] \mathbf G_n[i,j_1,i',j'] + \mathbf B[j_2,j'] \mathbf G_n[i,j_2,i',j']\big) \big/ \big( \mathbf B^2[j_1,j']  + \mathbf B^2[j_2,j'] \big)
\end{equation}
with the aggregated weight $\mathbf B^2[j_1,j']  + \mathbf B^2[j_2,j']$.

Let $\mathbf y_{n,i}$ (resp. $\mathbf y_{n,i'}$) denote the $i$-th (resp. $i'$-th) column of $\mathbf Y_n$. Let $\mathbf q_i$ denote the identifier of whether the entries of $\mathbf y_{n,i}$ were observed (and similarly $\mathbf q_{i'}$), i.e.
\[
\mathbf q_i[l] = \begin{cases}
1, \; l \in \{ j_1, j_2 \} \\
0, \; \text{otherwise}
\end{cases} \quad \& \quad 
\mathbf q_{i'}[l] = \begin{cases}
1, \; l = j' \\
0, \; \text{otherwise}.
\end{cases}
\]
Then, value \eqref{eq:pooled_point} can be calculated as $\mathbf y_{n,i}^\top \mathbf B \mathbf y_{n,i} / \mathbf q_i^\top \mathbf B_2 \mathbf q_{i'}$ with the aggregated weight given by $\mathbf q_i^\top \mathbf B_2 \mathbf q_{i'}$, where $\mathbf B_2$ is the entry-wise square of $\mathbf B$. Naturally, this can be generalized to the case when arbitrary number of entries in the $i$-th and $i'$-th columns of $\mathbf Y_n$ are observed. But more importantly, it can be also generalized to account for different pairs of columns of $\mathbf Y_n$ at the same time.

Let $\mathbf Q_n$. Then the contribution of the $n$-th surface $\mathbf Y_n$ into set \eqref{eq:set_discrete} can be calculated at once as
\begin{equation}\label{eq:pooling_formula}
    \mathbf Y_{n}^\top \widetilde{\mathbf B} \mathbf Y_{n} / \mathbf Q^\top \widetilde{\mathbf B}_2 \mathbf Q
\end{equation}
where $\widetilde{\mathbf B}$ is $\mathbf{B}$ with the diagonal values replaces by zeros and $\widetilde{\mathbf B}_2$ is the entry-wise square of $\widetilde{\mathbf B}$. The diagonal values of $\mathbf B$ are replaced by zeros as described in order to discard products of the type $\mathbf Y^2[i,j]$, which are burdened with noise.

The situation is analogous in the other step, when $\mathbf{A}$ is fixed and $\mathbf{B}$ is calculated. The whole procedure of estimating the separable covariance based on gridded sparse measurements is outlined in Algorithm \ref{alg:sparse_separable}.

\begin{algorithm}[t]
\caption{Estimation of separable model from sparsely observed zero-mean surfaces.}\label{alg:sparse_separable}
\begin{description}
\item[Input] $\mathbf{Y}_1,\ldots,\mathbf{Y}_N \in (\R \cup \{ \diamond \})^{d_1 \times d_2}$, where $\diamond$ represents a missing value
\item[] $ \mathbf Q_n := \mathds{1}_{[\mathbf Y_n \neq \diamond]} \in \{ 0, 1\}^{d_1 \times d_2}$, for $n=1,\ldots,N$
\item[] replace all $\diamond$ entries in $\mathbf{Y}_1,\ldots,\mathbf{Y}_N$ by zeros
\item[] $\mathbf A = \big(1 \big)_{i,j=1}^{d_1 \times d_1}$
\item[repeat]\mbox{}
\begin{description}
\item[for] $n=1,\ldots,N$
  \begin{description}
  \item[] $ \widetilde{\mathbf B} := \mathbf B$ with diagonal entries replaced by zeros
  \item[] $ \widetilde{\mathbf B}_2 :=$ entry-wise square of $\widetilde{\mathbf B}$
  \item[] $\mathbf W_n := \mathbf Q_n \widetilde{\mathbf B}_2 \mathbf Q_n^\top$
  \item[] $\mathbf Z_n := \mathbf Y_n \widetilde{\mathbf B} \mathbf Y_n^\top$ entry-wise divided by $\mathbf W_n$
  \end{description}
\item[end for]
\item[] $\mathbf A :=$ surface smoother of $\{ \mathbf Z_1, \ldots, \mathbf Z_N \}$ with $\{ \mathbf W_1, \ldots, \mathbf W_N \}$ as the smoothing weights
\item[for] $n=1,\ldots,N$
  \begin{description}
  \item[] $ \widetilde{\mathbf A} := \mathbf A$ with diagonal entries replaced by zeros
  \item[] $ \widetilde{\mathbf A}_2 :=$ entry-wise square of $\widetilde{\mathbf A}$
  \item[] $\mathbf W_n := \mathbf Q_n^\top \widetilde{\mathbf A}_2 \mathbf Q_n$
  \item[] $\mathbf Z_n := \mathbf Y_n^\top \widetilde{\mathbf A} \mathbf Y_n$ entry-wise divided by $\mathbf W_n$
  \end{description}
\item[end for]
\item[] $\mathbf B :=$ surface smoother of $\{ \mathbf Z_1, \ldots, \mathbf Z_N \}$ with $\{ \mathbf W_1, \ldots, \mathbf W_N \}$ as the smoothing weights
\end{description}
\item[until convergence (or only twice)]
\item[Output] $\mathbf{A}, \mathbf{B}$
\end{description}
\end{algorithm}

\subsection{Implementation Details}\label{sec:implementation_details}

Separability offers reductions in both time and memory complexities already when data are observed fully \cite{aston2017,masak2020}. In this section, we argue that computational gains of separability are even greater, when data are observed sparsely and kernel smoothing is used.

Kernel smoothers are known to be computationally demanding. To directly evaluate a kernel smoother in $d_1$ locations using $d_2$ observations takes $\O(d_1 d_2)$ operations. Table \ref{tab:complexities} shows these quadratic complexities in our situation, explained below. Assume that $N$ surfaces were observed on a grid of size $d \times d$ relatively densely (i.e. a fixed percentage of the grid was observed -- this is not unrealistic since one often chooses the grid size in such a way), and an unbounded kernel was used. The quadratic complexity of kernel smoother translates into estimating a general covariance by a surface smoother in $\O(N d^8)$ operations, because all $\O(N d^4)$ raw covariances have to be accessed at every single one of $\O(d^4)$ grid points. When we consider $N$ fixed, the resulting complexity in $d$, i.e. $\O(d^8)$, is huge. Under separability, not using the marginalization procedure, the complexity is $\O(d^6)$, because $\O(d^4)$ raw covariances have to be accessed at $\O(d^2)$ grid points. With marginalization, i.e. using formula \eqref{eq:pooling_formula}, the time complexity drops down to $\O(d^4)$, because the number of raw covariances that has to be accessed at every grid point decreases to $\O(d^2)$.

In practice, the quadratic complexity of kernel smoothers becomes intractable and there exist many computational approaches to reduce the burden. Most notably, the fast Fourier transform can be used on equispaced domains to reduce quadratic complexity to log-linear \cite{silverman1982}, effectively cutting down the powers of $d$ in the first row of Table \ref{tab:complexities}.
%Still, we will try to make a point that computational gains stemming from separability are far more profound when data are observed sparsely and smoothing step is used, compared to when data are fully observed on a grid.
Many other accelerating approaches exist, see e.g. \cite{raykar2010,langrene2019}, and references therein. However, software availability utilizing these computationally efficient approaches is rather limited, and this is particularly true for multi-dimensional problems.

We do not provide our own implementation of kernel smoothing. For the proposed approach, we implement Algorithm 1, which uses a ``surface smoother''. To this end, we utilize local linear smoothers provided in the \textsf{fdapace} package \cite{yao2005,fdapace2020}. We also utilize internal functions from \textsf{fdapace} to perform cross-validation for the choice of bandwidths.

\begin{table}[]
    \centering
    \caption{Complexities for covariance estimation of a random surface observed on a $d \times d$ grid.}
    \label{tab:complexities}
    \begin{tabular}{l c c c}
    \toprule
Complexity & Separability & Separability w/o & 4D smoothing \\
 & & Marginalization \\
 \midrule
 Time &    $\O(d^4)$ & $\O(d^6)$ & $\O(d^8)$ \\
 Memory &  $\O(d^2)$ & $\O(d^4)$ & $\O(d^4)$ \\
 \bottomrule
    \end{tabular}
\end{table}

For 4D smoothing, which we consider only for comparison, we use the \textsf{np} package \cite{racine2008}, which is to the best of our knowledge the only \textsf{R} \cite{R} package able to perform local linear polynomial regression surface smoothing in more than two dimensions. The 4D smoothing estimator requires smoothing in four dimensions. We provide the formula here for completeness. For the set of points $\{ (x_k, y_k, x_k',y_k') | k=1,\ldots,M \}$ and weights $\{ w_k | k=1,\ldots,M \}$, the smoothed surface $\widehat{\gamma}_0 = \widehat{\gamma}_0(x,y,x',y')$ is calculated as
\begin{equation}\label{eq:4d_smoothing}
\begin{split}
    (\widehat{\gamma_0}, \widehat{\gamma_1}, \widehat{\gamma_2}, \widehat{\gamma_3}, \widehat{\gamma_4}) = \argmin_{\gamma_0, \gamma_1, \gamma_2, \gamma_3, \gamma_4} \sum_{k=1}^M &K \left( \frac{x - x_k}{h_{1}} \right) K \left( \frac{y - y_k}{h_{2}} \right) K \left( \frac{x' - x_k'}{h_{3}} \right) K \left( \frac{y' - y_k'}{h_{4}} \right) w_k \cdot \\
    &\cdot \Big[ z_k - \gamma_0 - \gamma_1 (x - x_{k}) - \gamma_2 (y - y_{k}) - \gamma_3 (x' - x_{k}') - \gamma_4 (y' - y_{k}') \Big]^2  
\end{split}
\end{equation}
for every fixed $(x,y,x',y') \in [0,1]^4$, where $K(\cdot)$ is a smoothing kernel function and $h_1,h_2,h_3,h_4 > 0$ are bandwidths. Even though the \textsf{np} package implements cross-validation to choose the bandwidths, we found the computational burden to be huge and the performance rather poor in our setups, see Section \ref{sec:additional_simulations}. Hence, whenever we use 4D smoothing, we fix the unknown bandwidths as chosen by cross-validation for the proposed separable model. While this intuitively leads to smaller than optimal bandwidths, we found out in our simulation study that bandwidths are governed mainly by smoothness of the underlying covariance rather than the number of points per surface available. Since the smoothness of a four-dimensional covariance and its separable proxy is similar, optimal bandwidths chosen for the proposed (separable) approach seem to be reasonable for 4D smoothing as well, and this was verified in our experiments. Regardless, we can hardly afford other strategy for choosing the four bandwidths for 4D smoothing. Even the sophisticated combination of cross-validation and optimization provided in the \textsf{np} package leads huge runtimes in our setups, see Section \ref{sec:additional_simulations}.

\subsection{Setup of the Simulation Study}\label{sec:setup_simulations}

Here, we fully specify the four covariances chosen for our simulation study (Section \ref{sec:simulation_study}).
\begin{description}
\item[(a)] In the Fourier scenario, $a=a(t,t')$ and $b=b(s,s')$ are chosen to be the same, such that they have the trigonometric basis as their eigenfunctions and power decay of their eigenvalues, resulting in a rather wiggly univariate covariance displayed in Figure \ref{fig:covariance_choices} (left). The covariance is then set as $c(t,s,t',s') = a(t,t') \, b(s,s')$, resulting in a separable covariance.
\item[(b)] In the Brownian scenario, $a=a(t,t')$ and $b=b(s,s')$ are both chosen as the covariance of the Wiener process, i.e. $a(t,t') = \min(t,t')$ and $b(s,s') = \min(s,s')$, resulting in a rather flat covariance displayed in Figure \ref{fig:covariance_choices} (center). The covariance is then set as $c(t,s,t',s') = a(t,t') \, b(s,s')$, i.e. it is separable again.
\item[(c)] In the Gneiting scenario, the covariance has the following parametric form:
\begin{equation}\label{eq:gneiting_formula}
c(t,s,t',s') = \frac{\sigma^2}{(a^2|t-t'|^{2 \alpha}+1)^\tau} \exp\left( \frac{b^2 |s -s' |^{2 \gamma}}{ (a^2|t-t'|^{2 \alpha}+1)^{\beta\gamma} } \right),
\end{equation}
where $a=b=\tau=\alpha=\gamma=\sigma^2=1$ and $\beta=0.7$. This covariance is non-separable \cite{gneiting2002}, but it is rather flat.
\item[(d)] In the Fourier-Legendre scenario, we choose $a_1 = a_1(t,t')$ and $b_1=b_1(s,s')$ as the Fourier univariate covariances specified above. Furthermore, $a_2 = a_2(t,t')$ and $b_2=b_2(s,s')$ are both chosen as rank-4 covariances with shifted Legendre basis as their eigenfunctions, resulting in rather wiggly univariate covariances (see Figure \ref{fig:covariance_choices}, right). The covariance is then chosen as $c(t,s,t',s') = a_1(t,t') \, b_1(s,s') + a_2(t,t') \, b_2(s,s')$, resulting in a non-separable covariance.
\end{description}

\begin{figure}[!t]
  \centering
  \begin{tabular}{ccc}
  Fourier & Brownian & Legendre \\
  \includegraphics[width=0.31\textwidth]{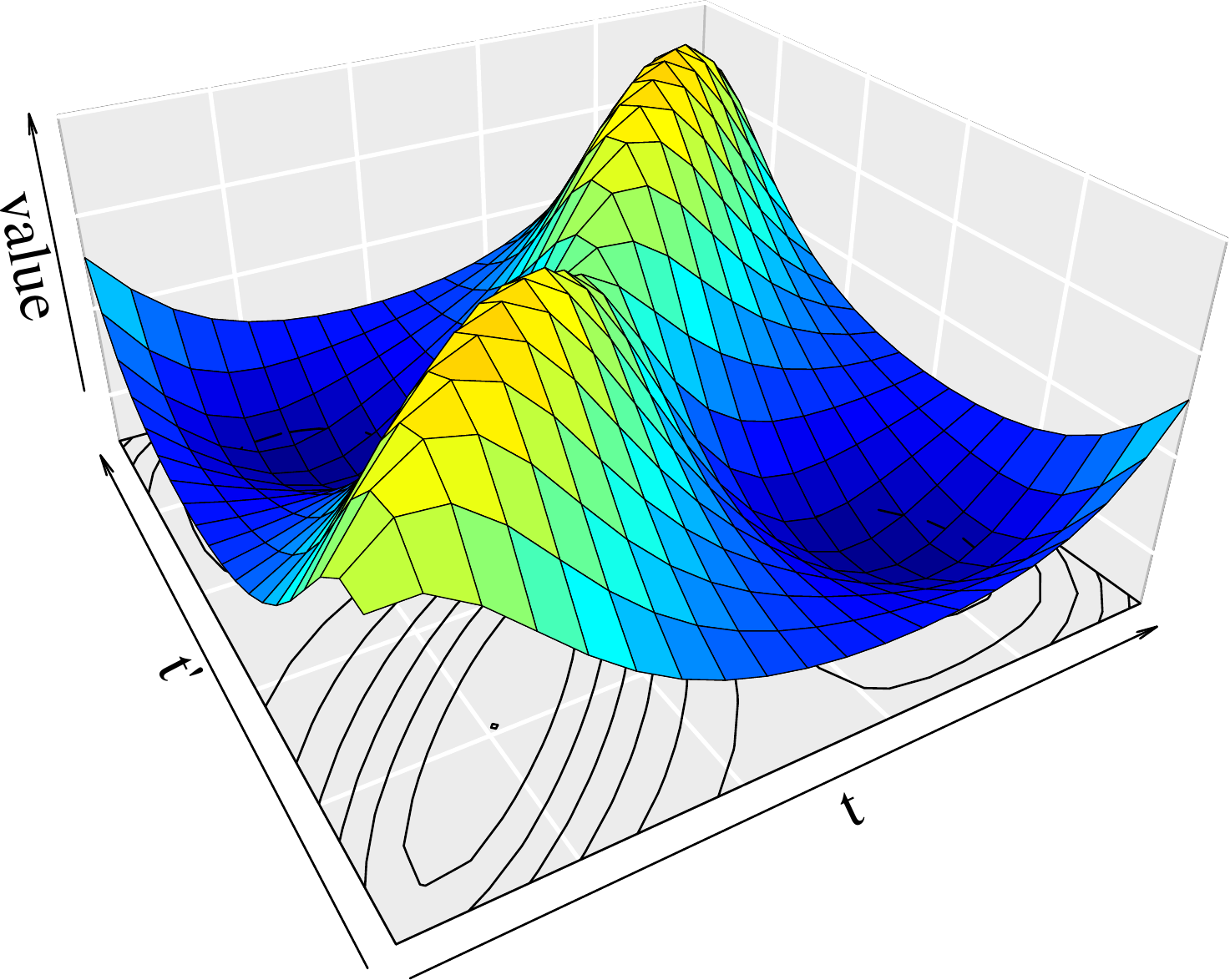} &
  \includegraphics[width=0.31\textwidth]{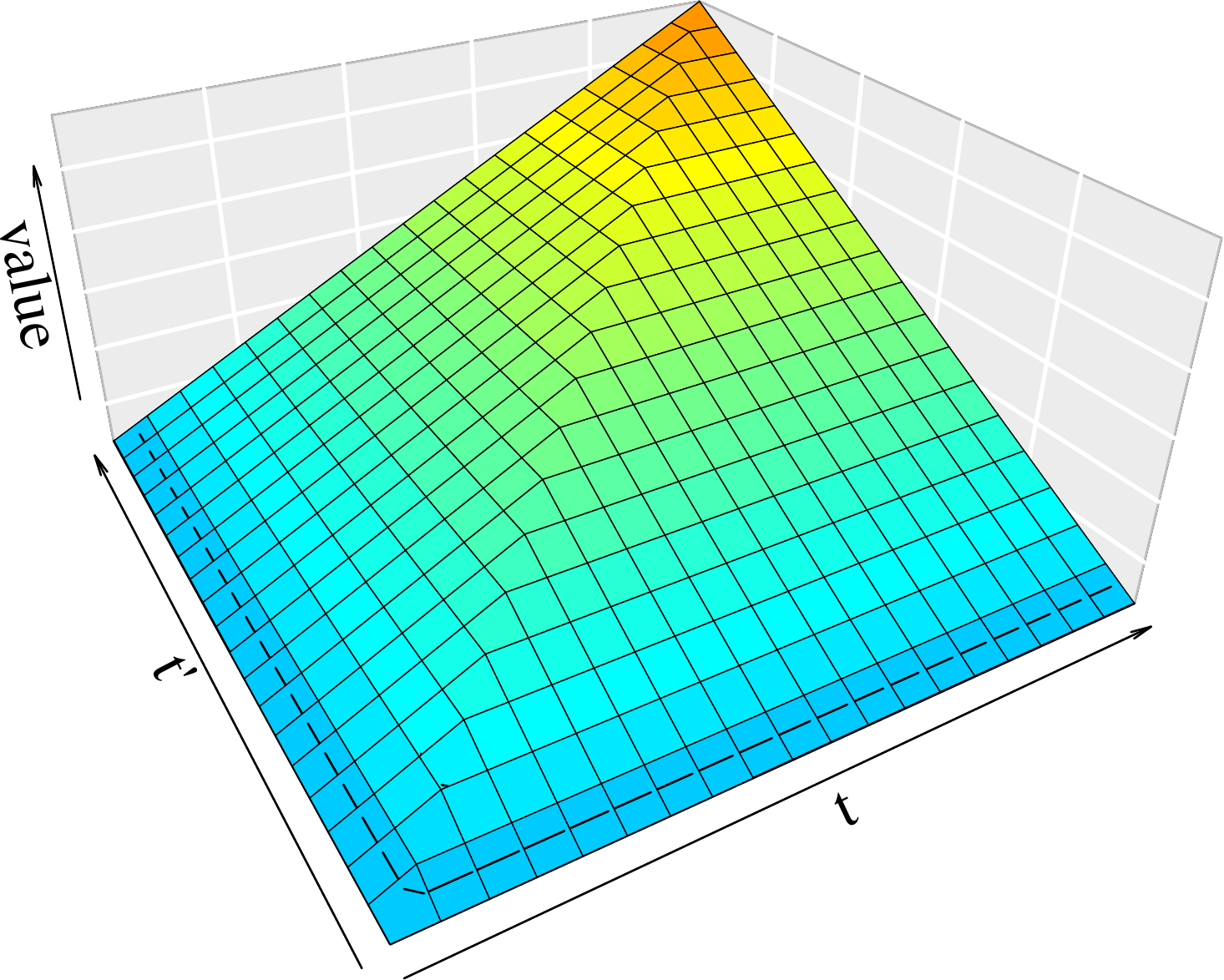} &
  \includegraphics[width=0.31\textwidth]{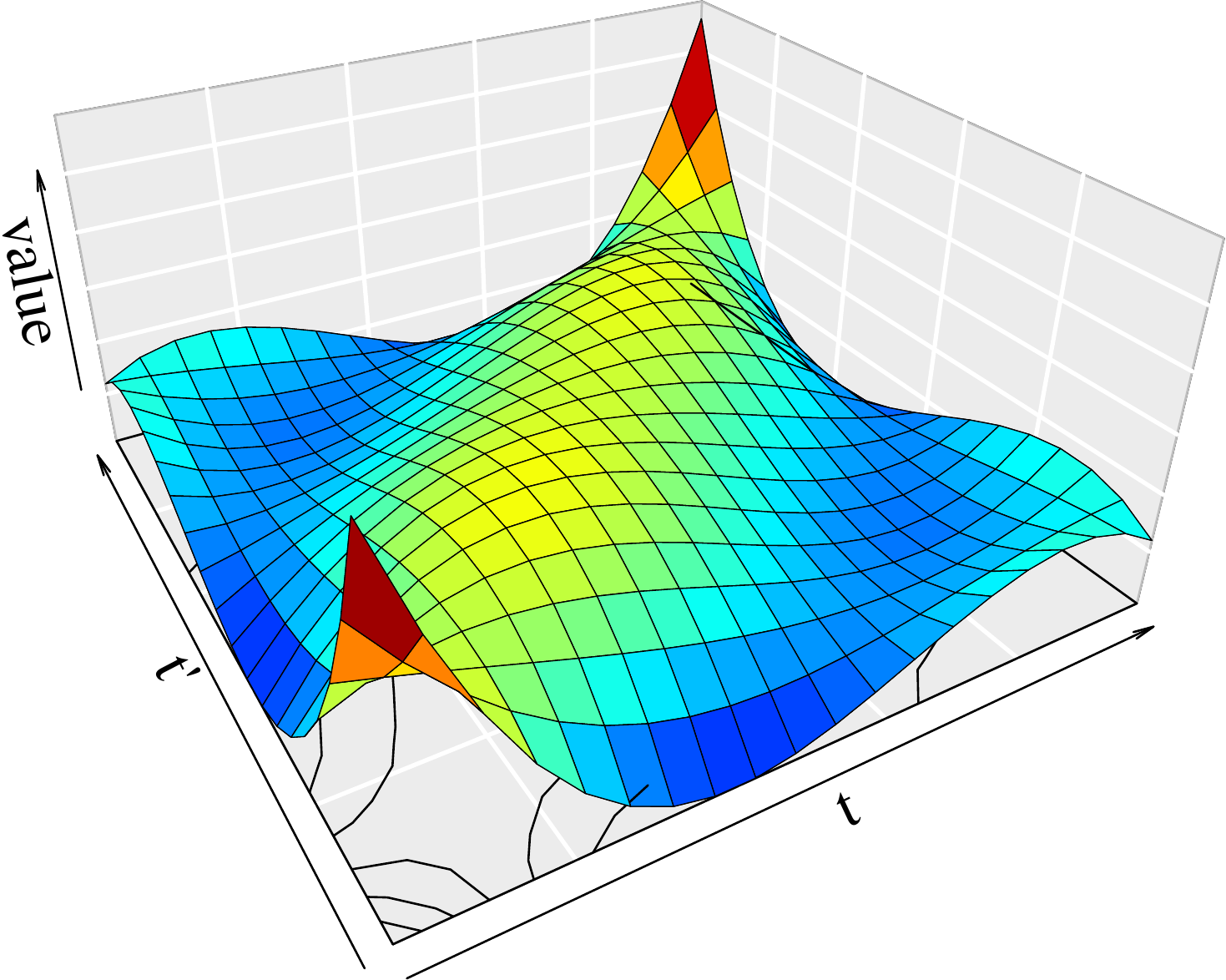} \\\\
  \multicolumn{3}{c}{
  \includegraphics[width=0.9\textwidth]{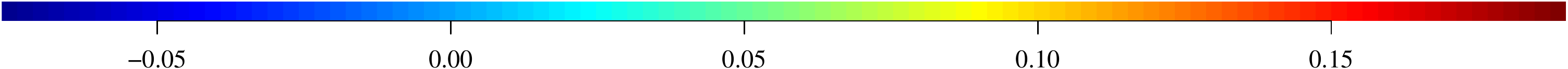}}
  \end{tabular}  
  \caption{Univariate covariances used as building blocks in our simulation study.}
    \label{fig:covariance_choices} 
\end{figure}

% \begin{figure}[!t]
%   \centering
%   \begin{tabular}{cc}
%   \includegraphics[width=0.47\textwidth]{plots/setup_gneiting_a.pdf} &
%   \includegraphics[width=0.47\textwidth]{plots/setup_gneiting_b.pdf}
%   \end{tabular}  
%   \caption{Two-dimensional slices through covariance \eqref{eq:gneiting_formula}, namely $c(t,1/2,t',1/2)$ is displayed (left) and $c(1/2,s,1/2,s')$ is shown (right).}
%     \label{fig:gneiting_cuts} 
% \end{figure}

\subsection{Additional Simulation Results}\label{sec:additional_simulations}

Firstly, we show runtimes for the Fourier scenario from our simulation study in Figure \ref{sec:simulation_study}. The runtimes look similarly for any of the remaining scenarios (not reported). To demonstrate effectiveness of the marginalization procedure described in detail in Section \ref{sec:marginalization}, we also show runtimes for the non-pooled procedure, considering all raw covariances in sets \ref{eq:set_for_A} or \ref{eq:set_for_B} as separate points for the purposes of smoothing. It leads to the same results as the proposed approach but, as more and more points per surface are observed, the number of scatter points supplied to the smoothing procedure increases rapidly, which increases the runtimes. But more importantly, at the edge of computational feasibility for the 4D smoothing approach (i.e. with the percentage $p=10$) the proposed separable estimator is calculated about 200 times faster than the 4D smoothing estimator.

Secondly, it was observed in Section \ref{sec:simulation_study} that the proposed approach outperforms the one-step version of our estimator $\widehat{a}_0 \cdot \widehat{b}_0$, cf. Figure \ref{fig:estimation_procedure}. While the proposed approach can be seen as a two-step version, a natural question arises whether a multi-step version of the estimator could not be much better. Figure \ref{fig:runtimes_3rd_iter} compares the estimation errors achieved by the proposed approach and by a three-step approach in all four scenarios considered in Section \ref{sec:simulation_study}. The third step offers a significant improvement in only one of the scenarios, and even then the improvement is relatively small compared to the improvements achieved in Section \ref{sec:simulation_study}.

\begin{figure}[!b]
   \centering
   \begin{tabular}{cc}
   \includegraphics[width=0.45\textwidth]{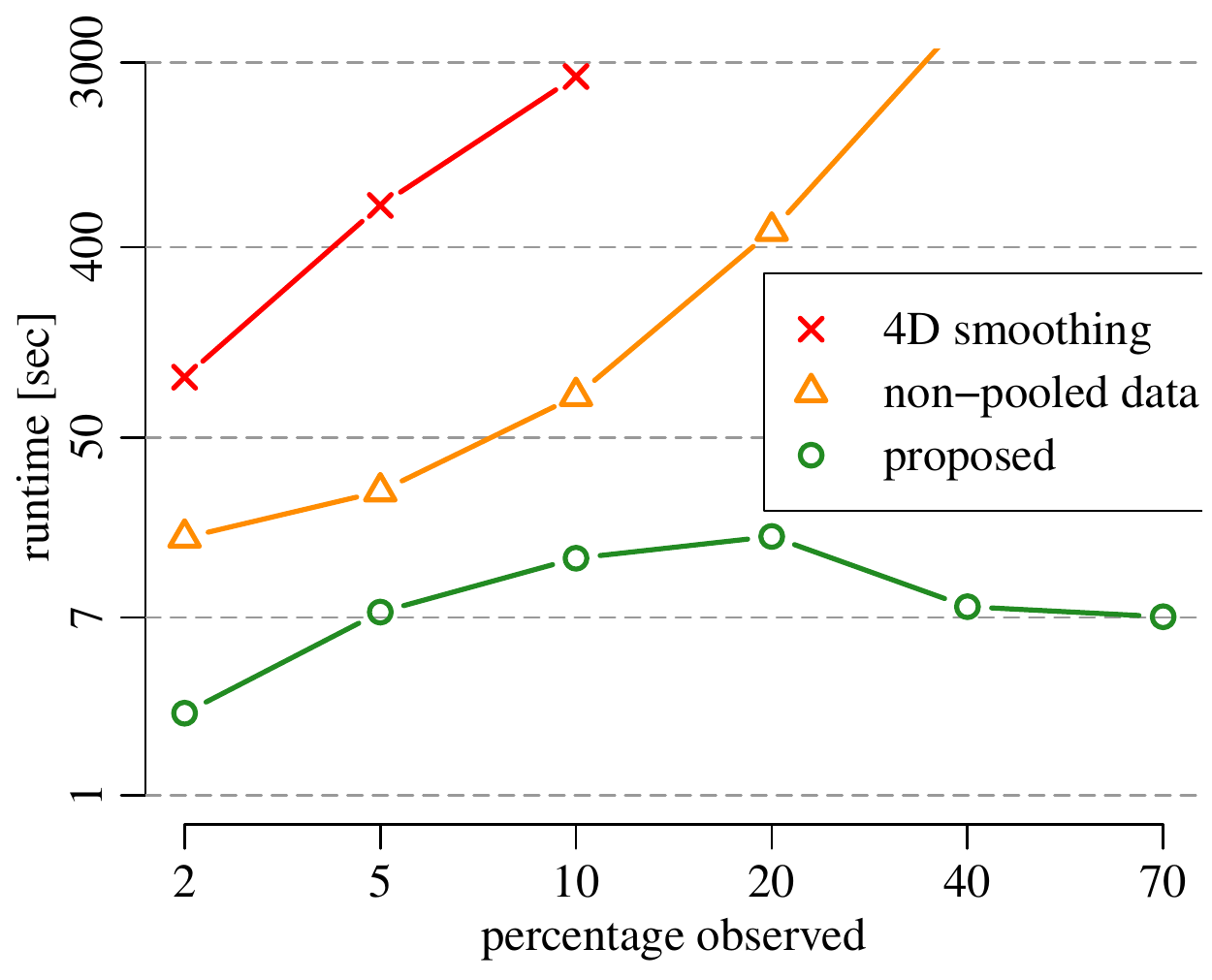} &
   \includegraphics[width=0.45\textwidth]{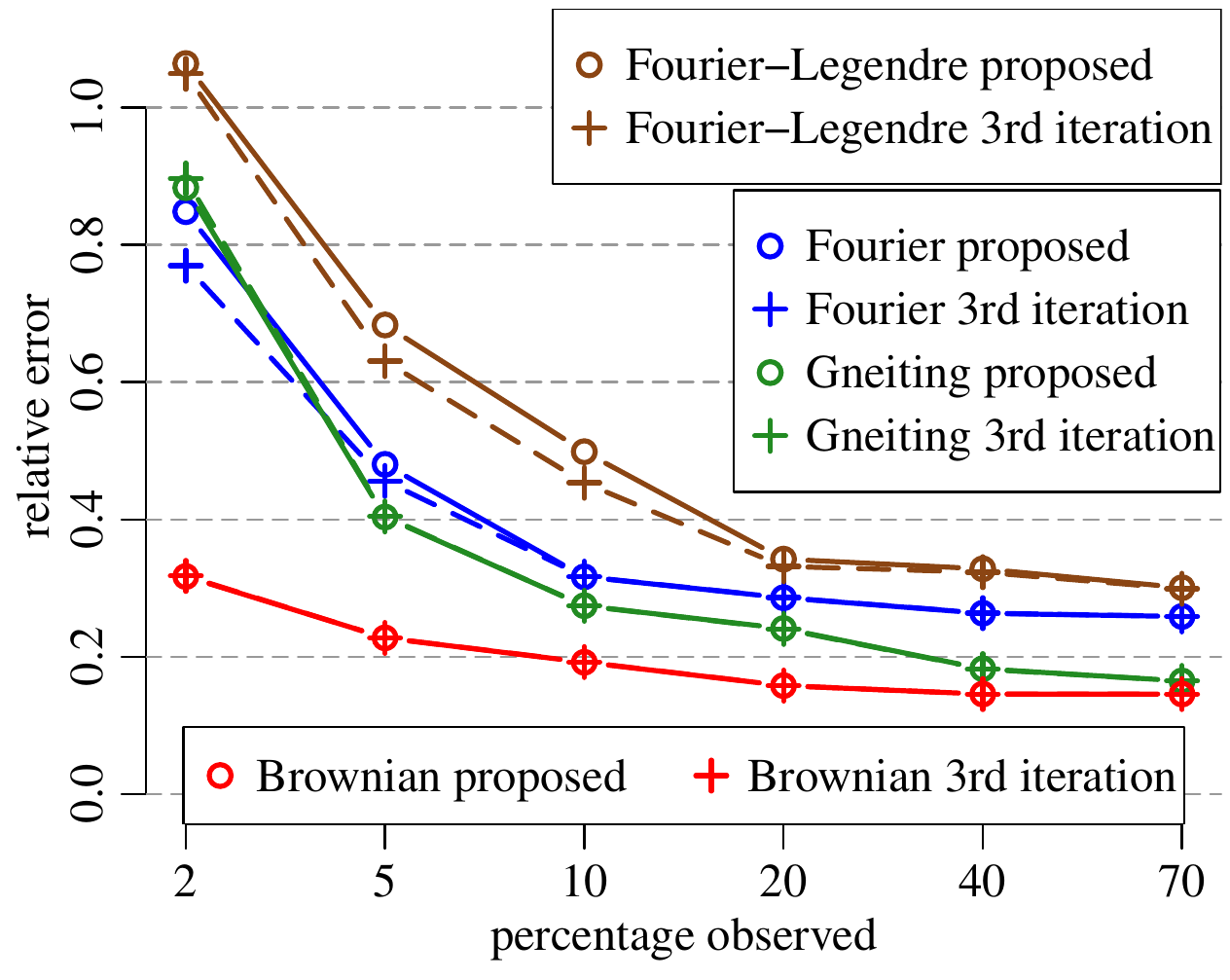}
   \end{tabular}  
   \caption{Left: Runtimes for Fourier simulations. Right: The proposed approach (two-step) compared the proposed approach with included third step for all the four scenarios considered in Section \ref{sec:simulation_study}.}
    \label{fig:runtimes_3rd_iter} 
\end{figure}

Finally, Table \ref{tab:CV} shows relative estimation errors of the 4D smoothing approach in all four scenarios used in Section \ref{sec:simulation_study}, but only with the smallest considered percentage $p=2$. The cross-validated choice of bandwidths is compared against the choice of bandwidth suggested by the separable model (which is used in Section \ref{sec:simulation_study}). It is clear that cross-validation fails here, because even in the simplest Brownian scenario, cross-validated relative error is larger that one. The reason for that is likely the following. Cross-validation for 4D smoothing, as implemented in the \textsf{np} package, does not evaluate its objective function on a grid. In order to reduce the computational burden, the cross-validation objective is optimized is a step-wise manner, until a stopping criterion is met. The stopping criterion is set as a \emph{tolerance} (defaults to $10^{-8}$), and the iterative optimization is stopped once both the change in objective value and change in the bandwidths is smaller than the tolerance for two consecutive iterations. While this saves computation time compared to creating a grid over potential bandwidth values and evaluating the objective function in all the grid points, the optimization can get stuck in a local minimum. This is the reason why the cross-validated errors in Table \ref{tab:CV} are so high. The sampling pattern is very sparse with $p=2$, and there likely are many local minima. However, the cross-validation still requires fitting the covariance for different values of the bandwidths. We tried to obtain results for larger values of $p$ as well, however we ran out of time (with a single task) at 70 hours with $p=5$ even with the tolerance decreased to $10^{-2}$. Hence, we have no other choice but to use the bandwidths suggested by the separable model also for 4D smoothing.

\begin{table}[]
    \centering
    \caption{Relative estimation errors of 4D smoothing with cross-validated bandwidths and bandwidths suggested by the separable model for the four scenarios considered in Section \ref{sec:simulation_study} with $p=2$ percentages of the surfaces observed.}
    \label{tab:CV}
    \begin{tabular}{lcccc}
\toprule
& (a) Fourier & (b) Brownian & (c) Gneiting & (d) Fourier-Legendre \\
    \midrule
cross-validation  & 1.07 & 1.06 & 1.04 & 1.39 \\
separable choice  & 0.04 & 0.26 & 1.34 & 1.1 \\
\bottomrule
    \end{tabular}
\end{table}

%%%%%%%%%%%%%%%%%%%%%%%%%%%%%%%%%%%%%%%%%%%%%%%%%%%%%%%%%%%%%%%%%%%%%%%%%%%%%%%%%%%%%%%%%%%%%%%%%%%%%%%%%%%%%%%%%%%%%%%%%%%%%%%%%%%%%%%%%%%%%%%%%%%%%%%%%%%%%%%%%%%%%%%%%%%%%%%
%%%%%%%%%%%%%%%%%%%%%%%%%%%%%%%%%%%%%%%%%%%%%%%%%%%%%%%%%%%%%%%%%%%%%%%%%%%%%%%%%%%%%%%%%%%%%%%%%%%%%%%%%%%%%%%%%%%%%%%%%%%%%%%%%%%%%%%%%%%%%%%%%%%%%%%%%%%%%%%%%%%%%%%%%%%%%%%
%%%%%%%%%%%%%%%%%%%%%%%%%%%%%%%%%%%%%%%%%%%%%%%%%%%%%%%%%%%%%%%%%%%%%%%%%%%%%%%%%%%%%%%%%%%%%%%%%%%%%%%%%%%%%%%%%%%%%%%%%%%%%%%%%%%%%%%%%%%%%%%%%%%%%%%%%%%%%%%%%%%%%%%%%%%%%%%
\section{Proofs of Formal Statements}
\label{sec:proofs}

\subsection{Explicit Formula for Local Polynomial Regression}
\label{subsec:explicit_formula_smoothers}

Our estimators introduced in Section~\ref{sec:methodology} are based on  local polynomial regression techniques and are defined as minimizers of (weighted) least squares problems  \eqref{eq:generic_smoother}. It turns out that the minimizers for these point-wise optimization problems admit a unique solution given by an explicit formula. In this section we recall this formula for a general local linear polynomial regression with possibly exogenous weights which will be later used in the proofs of the asymptotic behaviour of our estimators.

The local linear surface smoother of the generic set 
$\{ (x_k,y_k,z_k) \; | \; k = 1,\ldots,M\} \subset \R^3$ given weights $\{ w_k \; | \; k=1,\ldots,M \}$
is defined as the solution of the least squares problem \eqref{eq:generic_smoother}. It turns out that this minimizer to this optimization problem admits a unique solution:
\begin{equation}\label{eq:generic_smoother_solution}
\widehat{\gamma}_{0(x,y)} =
 \Psi_1(x,y) 
\left[ \Psi_2(x,y) \right]^{-1},
\qquad (x,y)\in[0,1]^2, 
\end{equation}
where for each $(x,y)\in[0,1]^2$ and $p,q\in\mathbb{N}_0$ we define
\begin{align}
\nonumber
\Phi_1(x,y) &= S_{20}(x,y) S_{02}(x,y) - \left[ S_{11}(x,y) \right]^2,\\
\nonumber
\Phi_2(x,y) &= S_{10}(x,y) S_{02}(x,y) - S_{01}(x,y)S_{11}(x,y),\\
\nonumber
\Phi_3(x,y) &= S_{01}(x,y) S_{20}(x,y) - S_{10}(x,y)S_{11}(x,y),\\
\nonumber
\Psi_1(x,y) &= \Phi_1(x,y) Q_{00}(x,y) - \Phi_2(x,y)Q_{10}(x,y) - \Phi_3(x,y)Q_{01}(x,y),\\
\nonumber
\Psi_2(x,y) &= \Phi_1(x,y) S_{00}(x,y) - \Phi_2(x,y)S_{10}(x,y) - \Phi_3(x,y)S_{01}(x,y),\\
\label{eq:generic_smoother_Spq}
S_{pq} (x,y) &= \frac{1}{M} \sum_{k=1}^M
\left( \frac{x-x_k}{h_1} \right)^p
\left( \frac{y-y_k}{h_2} \right)^q
\frac{1}{h_1 h_2}
K\left( \frac{x-x_k}{h_1} \right)
K\left( \frac{y-y_k}{h_2} \right)
w_m,
\qquad 0\leq p+q\leq 2,\\
\label{eq:generic_smoother_Qpq}
Q_{pq} (x,y) &= \frac{1}{M} \sum_{k=1}^M
\left( \frac{x-x_k}{h_1} \right)^p
\left( \frac{y-y_k}{h_2} \right)^q
\frac{1}{h_1 h_2}
K\left( \frac{x-x_k}{h_1} \right)
K\left( \frac{y-y_k}{h_2} \right)
w_k z_k,
\qquad 0\leq p+q\leq 1,
\end{align}
where $h_1>0$ and $h_2>0$ are smoothing bandwidths in the first and the second dimension respectively.

The formula \eqref{eq:generic_smoother_solution} is derived by differentiating the weighted least squares \eqref{eq:generic_smoother} and finding the solution to the normal equations. It is based on the standard steps used in the local regression literature, e.g. \cite{fan1996}[\S 3.1] or \cite{rubin2020}[\S B.2].

%%%%%%%%%%%%%%%%%%%%%%%%%%%%%%%%%%%%%%%%%%%%%%%%%%%%%%%%%%%%%%%%%%%%%%%%%%%%%%%%%%%%%%%%%%%%%%%%%%%%%%%%%%%%%%%%%%%%%%%%%%%%%%%%%%%%%%%%%%%%%%%%%%%%%%%%%%%%%%%%%%%%%%%%%%%%%%%
%%%%%%%%%%%%%%%%%%%%%%%%%%%%%%%%%%%%%%%%%%%%%%%%%%%%%%%%%%%%%%%%%%%%%%%%%%%%%%%%%%%%%%%%%%%%%%%%%%%%%%%%%%%%%%%%%%%%%%%%%%%%%%%%%%%%%%%%%%%%%%%%%%%%%%%%%%%%%%%%%%%%%%%%%%%%%%%

\subsection{Kernel Averages of $m$-dependent Data}

Thanks the explicit formula \eqref{eq:generic_smoother_solution} we may reduce the asymptotic behaviour assessment to the investigation of the terms \eqref{eq:generic_smoother_Spq} and \eqref{eq:generic_smoother_Qpq}. 
In this section we review the general asymptotic framework for the asymptotics of these kernel averages and hence the framework for the asymptotics of \eqref{eq:generic_smoother_solution}. We shall use the general theory developed by \cite{hansen2008uniform} who derived a toolbox for strong mixing time series data where the regressors attain values in possibly unbounded sets. Here we recall this result and write down a simplified version sufficient for our data.

\begin{theorem}[\cite{hansen2008uniform}]
\label{thm:hansen_special_case}
Let $\{(U_k, V_k, Z_k)\}_{k\in\mathbb{Z}} \in \mathbb{R}^3$ be a strictly stationary sequence of random vectors and consider the averages of the form
\begin{equation}\label{eq:hansen_kernel_average}
    \Xi_k(u,v) = \frac{1}{k h_1 h_2 } \sum_{i=1}^k \left( \frac{u-U_i}{h_1} \right)^p \left( \frac{v - V_i}{h_2} \right)^q  K\left( \frac{u-U_i}{h_1} \right) K\left( \frac{v - V_i}{h_2} \right) Z_i
\end{equation}
where $p,q\in\mathbb{N}_0$.
\begin{enumerate}[label=(C{\arabic*})]
    \item\label{assumption:C1}
    $K(\cdot)$ is the Epanechnikov kernel, i.e. $K(u) = (3/4)(1-u^2) \mathds{1}_{[|u|<1]}$.
    \item\label{assumption:C2}
    $(U_k,V_k)$ attain values in the set $[0,1]^2$.
    \item\label{assumption:C3}
    $\{(U_k,V_k,Z_k)\}$ is an $m$-dependent sequence, i.e. for each $k\in\mathbb{Z}$, the random vectors $(\dots,U_k,V_k,Z_k)$ and $(U_{k+m},V_{k+m},Z_{k+m}, U_{k+m+1}, V_{k+m+1}, Z_{k+m+1},\dots)$ are independent. 
    \item\label{assumption:C4}
    There exists $s>2$ such that
     %$\E [ |Z_1|^s ] < \infty $
     $(u,v) \mapsto \E [ |Z_1|^s \vert U_1=u, V_1=v ]$ is bounded.
     \item\label{assumption:C5}
     The probability density function of $(U_1,V_1)$ is twice continuously differentiable.
     \item\label{assumption:C6}
     The smoothing bandwidth satisfies $(\log k)/( k h_1 h_2 ) = o(1)$ as $k\to\infty$.
\end{enumerate}
Then the kernel averages \eqref{eq:hansen_kernel_average} satisfy
\begin{equation}\label{eq:hansen_kernel_average_convergence_stochastic}
\sup_{(u,v)\in [0,1]^2} \left| \Xi_k(u,v) - \E \Xi_k(u,v) \right| = \bigop{ \sqrt{\frac{\log k}{k h_1 h_2}} },\qquad\text{as}\quad k\to\infty.
\end{equation}
\end{theorem}

Theorem~\ref{thm:hansen_special_case} is essentially a special case of \cite{hansen2008uniform}[Thm 2] when the consider sequence is defined on a bounded domain and is $m$-dependent.
We also assume that the two dimensional smoothing kernel \eqref{eq:hansen_kernel_average} is the product of two one dimensional Epanechnikov kernels. Note that 
the function $(u,v)\mapsto u^p v^q K(u)K(v)$ satisfies Hansen's conditions on the smoothing kernel.

The only generalisation where Theorem~\ref{thm:hansen_special_case} deviates from \cite{hansen2008uniform}[Thm 2] is that we allow the smoothing bandwidth $(h_1,h_2)$ to attain different values in different directions. The proof of such generalisation, while having the smoothing kernel as a product of one-dimensional kernels, follows the lines of the proof \cite{hansen2008uniform}[Thm 2] where $h^d$, with $d=2$, is replaced by $h_1 h_2$.

The following corollary goes one step further and incorporates the convergence of $\E \Xi_k(u,v)$ into the statement \eqref{eq:hansen_kernel_average_convergence_stochastic}.

\begin{corollary}\label{corollary:hansen}
Under the assumption of Theorem~\ref{thm:hansen_special_case} consider the function
\begin{equation}\label{eq:corollary:hansen_def_M}
M(u,v) = c_p c_q \Ez{ Z_1 \middle| U_1 = u, V_1=v } f_{U_1,V_1}(u,v) , \qquad x,y\in[0,1],    
\end{equation}
where $f_{U_1,V_1}(\cdot,\cdot)$ denotes the probability density function of $(U_1,V_1)$ and
$c_r = \int x^r K(x) dx$ for $r\in\mathbb{N}_0$. Moreover:
\begin{enumerate}[label=(D{\arabic*})]
    \item\label{assumption:D}
    Assume that the function $M(\cdot,\cdot)$ is twice continuously differentiable on $[0,1]^2$.
\end{enumerate}
Then the kernel averages \eqref{eq:hansen_kernel_average} satisfy
\begin{equation}\label{eq:hansen_kernel_average_convergence_full}
\sup_{(u,v)\in [0,1]^2} \left| \Xi_k(u,v) - M(u,v) \right| = \bigop{ \sqrt{\frac{\log k}{k h_1 h_2}} + h_1^2 + h_2^2 },\qquad\text{as}\quad k\to\infty.
\end{equation}
\end{corollary}
\begin{proof}
We start by decomposing the supremum \eqref{eq:hansen_kernel_average_convergence_full} into a stochastic and a deterministic part
\begin{equation}\label{eq:hansen_corollary_eq1}
 \sup_{(u,v)\in[0,1]^2} \left| \Xi_k(u,v) - M(u,v)  \right|
\leq
\sup_{(u,v)\in[0,1]^2} \left| \Xi_k(u,v) - \E \Xi_k(u,v)\right|
+
\sup_{(u,v)\in[0,1]^2} \left| \E \Xi_k(u,v) - M(u,v)  \right|.
\end{equation}
The first term on the right-hand side of \eqref{eq:hansen_corollary_eq1} is of order $O_\Prob( \sqrt{(\log n)/( n h_1 h_2 )} )$ by Theorem~\ref{thm:hansen_special_case}.
The expectation in the second term on the right-hand side of \eqref{eq:hansen_corollary_eq1} is developed as
\begin{align}
\nonumber
\E \Xi_k(u,v) &=
\Ez{
\left( \frac{u-U_{1}}{h_{1}} \right)^p
\left( \frac{v-V_{1}}{h_{2}} \right)^q
\frac{1}{h_{1} h_{2}}
K\left( \frac{u-U_{1}}{h_{1}} \right)
K\left( \frac{v-V_{1}}{h_{2}} \right)
\Ez{ Z_{11} \middle| U_{1}, V_{1} }
} \\
\nonumber
&=
\Ez{
\left( \frac{u-U_{1}}{h_{1}} \right)^p
\left( \frac{v-V_{1}}{h_{2}} \right)^q
\frac{1}{h_{1} h_{2}}
K\left( \frac{u-U_{1}}{h_{1}} \right)
K\left( \frac{v-V_{1}}{h_{2}} \right)
M( U_{1}, V_{1} )
} \\
\nonumber
&= \iint
\left( \frac{u-x}{h_{1}} \right)^p
\left( \frac{v-y}{h_{2}} \right)^q
\frac{1}{h_{1} h_{2}}
K\left( \frac{u-x}{h_{1}} \right)
K\left( \frac{v-y}{h_{2}} \right)
M( x, y )
f_{U_1,V_1}(x,y) dx dy \\
\label{eq:hansen_corollary_eq2}
&= \iint
\left( \tilde{x} \right)^p
\left( \tilde{y} \right)^q
K\left( \tilde{x} \right)
K\left( \tilde{y} \right)
M( u + h_{1} \tilde{x}, v + h_{2} \tilde{y} )
f_{U_1,V_1}( u + h_{1} \tilde{x}, v + h_{2} \tilde{y} ) d\tilde{x} d\tilde{y}.
\end{align}
Applying the Taylor expansion of order 2 in the right-hand side of \eqref{eq:hansen_corollary_eq2} and using the assumptions \ref{assumption:C5} and \ref{assumption:D},
the second term on the right-hand side of \eqref{eq:hansen_corollary_eq1} is of order $O( h_1^2+h_2^2 )$.
\end{proof}

%%%%%%%%%%%%%%%%%%%%%%%%%%%%%%%%%%%%%%%%%%%%%%%%%%%%%%%%%%%%%%%%%%%%%%%%%%%%%%%%%%%%%%%%%%%%%%%%%%%%%%%%%%%%%%%%%%%%%%%%%%%%%%%%%%%%%%%%%%%%%%%%%%%%%%%%%%%%%%%%%%%%%%%%%%%%%%%
%%%%%%%%%%%%%%%%%%%%%%%%%%%%%%%%%%%%%%%%%%%%%%%%%%%%%%%%%%%%%%%%%%%%%%%%%%%%%%%%%%%%%%%%%%%%%%%%%%%%%%%%%%%%%%%%%%%%%%%%%%%%%%%%%%%%%%%%%%%%%%%%%%%%%%%%%%%%%%%%%%%%%%%%%%%%%%%

\subsection{Proof of Proposition \ref{prop:estim_mu}}

Tailoring the generic smoother \eqref{eq:generic_smoother} to the mean surface estimator \eqref{eq:smoother_mu}, we arrive at the customized versions of \eqref{eq:generic_smoother_Spq} and \eqref{eq:generic_smoother_Qpq}:
\begin{align}
\label{eq:mu_smoother_Spq}
S_{pq}^\mu (t,s) &=
\frac{1}{\sum_{n=1}^N M_n}
\sum_{n=1}^N \sum_{m=1}^{M_n}
\left( \frac{t-t_{nm}}{h_{\mu,1}} \right)^p
\left( \frac{s-s_{nm}}{h_{\mu,2}} \right)^q \cdot\\
&\qquad\cdot
\frac{1}{h_{\mu,1} h_{\mu,2}}
K\left( \frac{t-t_{nm}}{h_{\mu,1}} \right)
K\left( \frac{s-s_{nm}}{h_{\mu,2}} \right),
\quad 0\leq p+q\leq 2.\\
\label{eq:mu_smoother_Qpq}
Q_{pq}^\mu (t,s) &=
\frac{1}{\sum_{n=1}^N M_n}
\sum_{n=1}^N \sum_{m=1}^{M_n}
\left( \frac{t-t_{nm}}{h_{\mu,1}} \right)^p
\left( \frac{s-s_{nm}}{h_{\mu,2}} \right)^q\cdot\\
&\qquad\cdot
\frac{1}{h_{\mu,1} h_{\mu,2}}
K\left( \frac{t-t_{nm}}{h_{\mu,1}} \right)
K\left( \frac{s-s_{nm}}{h_{\mu,2}} \right)
Y_{nm},
\quad 0\leq p+q\leq 1.
\end{align}

We assess the asymptotic behaviour of \eqref{eq:mu_smoother_Spq} and \eqref{eq:mu_smoother_Qpq} in the following lemmas.

\begin{lemma}\label{lemma:smoother_mu_Qpq_asymptotics}
Under the assumptions \ref{assumption:B1} -- \ref{assumption:B6}, 
\begin{equation}\label{eq:smoother_mu_Qpq_asymptotics}
\sup_{(t,s)\in[0,1]^2} \left| Q^{\mu}_{pq}(t,s) - M_{[ Q^{\mu}_{pq} ]}(t,s) \right| 
= \bigop{  \sqrt{\frac{\log N}{N h_1 h_2}} + h_{\mu,1}^2 + h_{\mu,2}^2   },
\qquad\text{as}\quad N\to\infty,
\end{equation}
for each $0\leq p+q \leq 1$ and where
\begin{equation}\label{eq:smoother_mu_Qpq_asymptotics_def_M}
M_{[ Q^{\mu}_{00} ]}(t,s) = \mu(t,s) f_{(t,s)}(t,s),
\qquad
M_{[ Q^{\mu}_{10} ]}(t,s) = M_{[ Q^{\mu}_{01} ]}(t,s) = 0,
\qquad t,s\in[0,1].
\end{equation}
\end{lemma}

\begin{proof}
Define the sequence of random vectors $\{(U_k,V_k,Z_k)\}_{k=1}^\infty$ by putting $\{t_{nm}\}$, $\{s_{nm}\}$ and $\{Y_{nm}\}$ in order such that
\begin{align}
    \nonumber
    \{U_1,U_2,\dots \} &= \{ t_{11}, t_{12}, \dots, t_{1m_1}, t_{21}, \dots, t_{2m_2}, t_{31}, \dots \}, \\
    \nonumber
    \{V_1,V_2,\dots \} &= \{ s_{11}, s_{12}, \dots, s_{1m_1}, s_{21}, \dots, s_{2m_2}, s_{31}, \dots \}, \\
    \label{eq:smoother_mu_Qpq_asymptotics:def_Z}
    \{Z_1,Z_2,\dots \} &= \{ Y_{11}, Y_{12}, \dots, Y_{1m_1}, Y_{21}, \dots, Y_{2m_2}, Y_{31}, \dots \}.
\end{align}

The sequence$\{(U_k,V_k,Z_k)\}_{k=1}^\infty$ satisfies the assumption of Theorem~\ref{thm:hansen_special_case} and Corollary~\ref{corollary:hansen}, namely strict stationarity is by assumptions \ref{assumption:B1} -- \ref{assumption:B3}, is $M^{max}$-dependent by assumption \ref{assumption:B1}, and the conditions \ref{assumption:C4}, \ref{assumption:C5}, \ref{assumption:C6}, \ref{assumption:D} are satisfied by the assumptions \ref{assumption:B5}, \ref{assumption:B2}, \ref{assumption:B6}, \ref{assumption:B5} respectively.
Therefore the sequence of kernel averages
$$ \Xi^\mu_{pq,k}(t,s) = \frac{1}{k h_{\mu,1}h_{\mu,2}} \sum_{i=1}^n
\left( \frac{t-U_i}{h_1} \right)^p \left( \frac{s - V_i}{h_2} \right)^q  K\left( \frac{t-U_i}{h_1} \right) K\left( \frac{s - V_i}{h_2} \right) Z_i
$$
satisfies
$$
\sup_{(t,s)\in[0,1]^2} \left| \Xi^\mu_{pq,k}(t,s) - M_{[ Q^{\mu}_{pq} ]}(t,s) \right| = \bigop{  \sqrt{\frac{\log k}{k h_1 h_2}} + h_{\mu,1}^2 + h_{\mu,2}^2  },
\qquad\text{as}\quad k\to\infty,
$$
and the formulae \eqref{eq:smoother_mu_Qpq_asymptotics_def_M} follow from the definition of $\{(U_k,V_k,Z_k)\}_{k=1}^\infty$ and definition \eqref{eq:corollary:hansen_def_M}.
Since the sequence $\{ Q^{\mu}_{pq}(t,s) \}_{N=1}^\infty$ is a subsequence of $\{ \Xi^\mu_{pq,k}(t,s) \}_{k=1}^\infty$ and $k = k(N) \asymp N$ as $N\to\infty$,
the convergence rate \eqref{eq:smoother_mu_Qpq_asymptotics} holds as well.
\end{proof}

\begin{lemma}\label{lemma:smoother_mu_Spq_asymptotics}
Under the assumptions \ref{assumption:B1} -- \ref{assumption:B3}, 
\begin{equation}\label{eq:smoother_mu_Spq_asymptotics}
\sup_{(t,s)\in[0,1]^2} \left| S^{\mu}_{pq}(t,s) - M_{[ S^{\mu}_{pq} ]}(t,s) \right| 
= \bigop{  \sqrt{\frac{\log N}{N h_1 h_2}} + h_{\mu,1}^2 + h_{\mu,2}^2   },
\qquad\text{as}\quad N\to\infty,
\end{equation}
for each $0\leq p+q \leq 1$ and where
\begin{align}\label{eq:smoother_mu_Spq_asymptotics_def_M_row1}
M_{[ S^{\mu}_{00} ]}(t,s) &= \mu(t,s) f_{(t,s)}(t,s),&
M_{[ S^{\mu}_{11} ]}(t,s) &= M_{[ S^{\mu}_{10} ]}(t,s) = M_{[ S^{\mu}_{01} ]}(t,s) = 0,\\
\label{eq:smoother_mu_Spq_asymptotics_def_M_row2}
M_{[ S^{\mu}_{20} ]}(t,s) &= M_{[ S^{\mu}_{02} ]}(t,s) = c_2 f_{(t,s)}(t,s),&
c_2 &= \int x^2 K(x) dx.
\qquad t,s\in[0,1].
\end{align}
\end{lemma}
\begin{proof}
The proof of this lemma follows essentially the same lines as the proof of Lemma~\ref{lemma:smoother_mu_Qpq_asymptotics}.
In the definition of the sequence $\{(U_k,V_k,Z_k)\}_{k=1}^\infty$ we put $Z_k = 1$, for all $k\in\mathbb{N}$, on the line \eqref{eq:smoother_mu_Qpq_asymptotics:def_Z}.
The formulae \eqref{eq:smoother_mu_Spq_asymptotics_def_M_row1} and \eqref{eq:smoother_mu_Spq_asymptotics_def_M_row2} are verified analogously by the definition \eqref{eq:corollary:hansen_def_M}.
\end{proof}

\begin{proof}[Proof of Proposition~\ref{prop:estim_mu}]
We are now ready to combine the above and proof Proposition~\ref{prop:estim_mu}.
Following the formulae presented in Section~\ref{subsec:explicit_formula_smoothers}, we have for
\begin{align*}
    \Phi^{\mu}_1(t,s) &= S^{\mu}_{20}(t,s) S^{\mu}_{02}(t,s) - \left[ S^{\mu}_{11}(t,s) \right]^2,\\
    \Phi^{\mu}_2(t,s) &= S^{\mu}_{10}(t,s) S^{\mu}_{02}(t,s) - S^{\mu}_{01}(t,s)S^{\mu}_{11}(t,s),\\
    \Phi^{\mu}_3(t,s) &= S^{\mu}_{01}(t,s) S^{\mu}_{20}(t,s) - S^{\mu}_{10}(t,s)S^{\mu}_{11}(t,s),\\
    \Psi^{\mu}_1(t,s) &= \Phi_1(t,s) Q^{\mu}_{00}(t,s) - \Phi_2(t,s)Q^{\mu}_{10}(t,s) - \Phi_3(t,s)Q^{\mu}_{01}(t,s),\\
    \Psi^{\mu}_2(t,s) &= \Phi_1(t,s) S^{\mu}_{00}(t,s) - \Phi_2(t,s)S^{\mu}_{10}(t,s) - \Phi_3(t,s)S^{\mu}_{01}(t,s),\\
\end{align*}
their asymptotic behaviour
\begin{align*}
    \Phi^{\mu}_1(t,s) &= \left( c_2 f_{(t,s)}(t,s) \right)^2 + \bigop{r^{\mu}_N},&
    \Phi^{\mu}_2(t,s) &= \bigop{r^{\mu}_N},\qquad\qquad
    \Phi^{\mu}_3(t,s) = \bigop{r^{\mu}_N},\\
    \Psi^{\mu}_1(t,s) &= \left(c_2 \right)^2 \left(f_{(t,s)}(t,s) \right)^3 \mu(t,s) + \bigop{r^{\mu}_N},&
    \Psi^{\mu}_2(t,s) &= \left(c_2 \right)^2 \left(f_{(t,s)}(t,s) \right)^3 + \bigop{r^{\mu}_N},
\end{align*}
uniformly in $(t,s)\in[0,1]^2$, as $N\to\infty$, where $r^\mu_N = \sqrt{ (\log N)/(N h_{\mu,1} h_{\mu,2})} + h_{\mu,1}^2 + h_{\mu,2}^2$.
Hence
$$ \widehat\mu(t,s) = \Psi^{\mu}_1(t,s) / \left[ \Psi^{\mu}_2(t,s) \right]^{-1} = \mu(t,s) + \bigop{r^{\mu}_N},\qquad\text{as}\quad N\to\infty,$$
by the uniform version of Slutsky's theorem and by the fact that $f_{(t,s)}(\cdot,\cdot) \neq 0$ on $[0,1]^2$ by \ref{assumption:B2}.
Thanks to the assumption \ref{assumption:B6}, the convergence rate simplifies to the common rate $h$ and
$$ \widehat\mu(t,s) = \mu(t,s) + \bigop{ \sqrt{\frac{\log N}{N h^2}} + h^2 },\qquad\text{as}\quad N\to\infty.$$
\end{proof}

%%%%%%%%%%%%%%%%%%%%%%%%%%%%%%%%%%%%%%%%%%%%%%%%%%%%%%%%%%%%%%%%%%%%%%%%%%%%%%%%%%%%%%%%%%%%%%%%%%%%%%%%%%%%%%%%%%%%%%%%%%%%%%%%%%%%%%%%%%%%%%%%%%%%%%%%%%%%%%%%%%%%%%%%%%%%%%%
%%%%%%%%%%%%%%%%%%%%%%%%%%%%%%%%%%%%%%%%%%%%%%%%%%%%%%%%%%%%%%%%%%%%%%%%%%%%%%%%%%%%%%%%%%%%%%%%%%%%%%%%%%%%%%%%%%%%%%%%%%%%%%%%%%%%%%%%%%%%%%%%%%%%%%%%%%%%%%%%%%%%%%%%%%%%%%%

\subsection{Proof of Theorem \ref{thm:estim_a_b} and Corollary~\ref{corollary:estimation_complete_covariance_structure}}

\begin{lemma}\label{lemma:proof_smoother_A_lemma_asymptotics}
Assume the conditions \ref{assumption:A1}, \ref{assumption:B1} -- \ref{assumption:B9} and
fix a deterministic twice continuously differentiable kernel $\beta(s,s'),\,s,s'\in[0,1]$ such that $\iint [\beta(s,s')]^2 dsds'>0$. Then the smoother $\widehat{\alpha}(t,t'),\,t,t'\in[0,1]$, obtained by smoothing the set
\begin{equation}\label{eq:proof_smoother_A}
\left\{ \left(t_{nm},t_{nm'},\frac{G_{nmm'}}{\beta(s_{nm},s_{nm'})}\right) \; \Bigg| \; m,m'=1,\ldots,M_n, \, m \neq m',\; n=1\ldots,N \right\}
\end{equation}
using weights $\{ \beta^2(s_{nm},s_{nm'}) \}$
admits the following asymptotics
\begin{equation}
\label{eq:proof_smoother_A_lemma_asymptotics}
\widehat{\alpha}(t,t') =  a(t,t')
\frac{ \iint \beta(s,s') b(s,s')  f_s(s)f_s(s') dsds' }{ \iint \left[\beta(s,s')\right]^2  f_s(s)f_s(s') dsds' } +  \bigop{ \sqrt{\frac{\log N}{N h^2}} + h^2 }
\end{equation}
uniformly in $(t,t')\in[0,1]^2$ as $N\to\infty$.
We recall that $a(t,t'),\,t,t'\in[0,1],$ and $b(s,s'),\,s,s'\in[0,1]$, on the right-hand side of \eqref{eq:proof_smoother_A_lemma_asymptotics} are the true separable covariance structure components \eqref{eq:separability}.
\end{lemma}

\begin{proof}
We start the proof by the analysis of a simplified case. Suppose that we know the mean surface $\mu(\cdot,\cdot)$ and define the raw covariances accordingly
\begin{equation}\label{eq:proofs_raw_covariance_tilde}
\tilde{G}_{nmm'} = (Y_{nm} - \mu(t_{nm},s_{nm}))( Y_{nm'} - \mu(t_{nm'},s_{nm'})).    
\end{equation}
Construct the smoother of the set \eqref{eq:proof_smoother_A} where we replace $G_{nmm'}$ by $\tilde{G}_{nmm'}$.
Such smoother, denoted as $\tilde{\alpha}(\cdot,\cdot)$, is given by the formula
\begin{align}
    \label{eq:proof_smoother_A_lemma_asymptotics_formula_smoother}
    \tilde{\alpha}(t,t') &= \Psi^{\tilde{\alpha}}_1(t,t') / \left[ \Psi^{\tilde{\alpha}}_2(t,t') \right]^{-1},\\
    \nonumber
    \Psi^{\tilde{\alpha}}_1(t,t') &= \Phi_1(t,t') Q^{\tilde{\alpha}}_{00}(t,t') - \Phi_2(t,t')Q^{\tilde{\alpha}}_{10}(t,t') - \Phi_3(t,t')Q^{\tilde{\alpha}}_{01}(t,t'),\\
    \nonumber
    \Psi^{\tilde{\alpha}}_2(t,t') &= \Phi_1(t,t') S^{\tilde{\alpha}}_{00}(t,t') - \Phi_2(t,t')S^{\tilde{\alpha}}_{10}(t,t') - \Phi_3(t,t')S^{\tilde{\alpha}}_{01}(t,t'),\\
    \nonumber
    \Phi^{\tilde{\alpha}}_1(t,t') &= S^{\tilde{\alpha}}_{20}(t,t') S^{\tilde{\alpha}}_{02}(t,t') - \left[ S^{\tilde{\alpha}}_{11}(t,t') \right]^2,\\
    \nonumber
    \Phi^{\tilde{\alpha}}_2(t,t') &= S^{\tilde{\alpha}}_{10}(t,t') S^{\tilde{\alpha}}_{02}(t,t') - S^{\tilde{\alpha}}_{01}(t,t')S^{\tilde{\alpha}}_{11}(t,t'),\\
    \nonumber
    \Phi^{\tilde{\alpha}}_3(t,t') &= S^{\tilde{\alpha}}_{01}(t,t') S^{\tilde{\alpha}}_{20}(t,t') - S^{\tilde{\alpha}}_{10}(t,t')S^{\tilde{\alpha}}_{11}(t,t'),\\
\nonumber
S_{pq}^{\tilde{\alpha}}(t,t') &=
\frac{1}{\sum_{n=1}^N M_n(M_n-1)}
\sum_{n=1}^N \sum_{ \substack{m,m'=1 \\ m\neq m'} }^{M_n}
\left( \frac{t-t_{nm}}{h_{a}} \right)^p
\left( \frac{t'-t_{nm'}}{h_{a}} \right)^q \cdot\\
\label{eq:proof_smoother_A_lemma_asymptotics_formula_Spq}
&\qquad\cdot
\frac{1}{h_a^2}
K\left( \frac{t-t_{nm}}{h_a} \right)
K\left( \frac{t'-t_{nm'}}{h_a} \right)
\left[\beta(s_{nm},s_{nm'}) \right]^2,
\qquad 0\leq p+q\leq 2,\\
\nonumber
Q_{pq}^{\tilde{\alpha}}(t,t') &=
\frac{1}{\sum_{n=1}^N M_n(M_n-1)}
\sum_{n=1}^N \sum_{ \substack{m,m'=1 \\ m\neq m'} }^{M_n}
\left( \frac{t-t_{nm}}{h_{a}} \right)^p
\left( \frac{t'-t_{nm'}}{h_{a}} \right)^q\cdot\\
\label{eq:proof_smoother_A_lemma_asymptotics_formula_Qpq}
&\qquad\cdot
\frac{1}{h_a^2}
K\left( \frac{t-t_{nm}}{h_a} \right)
K\left( \frac{t'-t_{nm'}}{h_a} \right)
\beta(s_{nm},s_{nm'})
\tilde{G}_{nmm'},
\qquad 0\leq p+q\leq 1.
\end{align}

The asymptotic behaviour of $Q_{pq}^{\tilde{\alpha}}$ and $S_{pq}^{\tilde{\alpha}}(t,t')$ is assessed similarly as the surface smoother in Lemmas~\ref{lemma:smoother_mu_Qpq_asymptotics} and \ref{lemma:smoother_mu_Spq_asymptotics}. We proceed again with defining the $[M^{max}]^2$-dependent sequences $\{(U_k,V_k,Z_k)\}_{k=1}^\infty$ by putting the pairs $(t_{nm},t_{nm'}),\, m,m'=1,\dots, M_n,\,m\neq m',\, n=1,\dots,N$ into the sequence $\{ (U_k, V_k) \}_{k=1}^\infty$ such that we set $U_k = t_{nm}$ and $V_k = t_{nm'}$ while starting from the data from the first surface ($n=1$), then proceeding with $n=2$ etc.

For the asymptotics of $Q_{pq}^{\tilde{\alpha}}$, define
\begin{equation}\label{eq:proof_smoother_A_lemma_asymptotics_Q_def_Z}
Z^Q_k = \beta(s_{nm},s_{nm'}) \tilde{G}_{nmm'}    
\end{equation}
where $s_{nm},s_{nm'},\tilde{G}_{nmm'}$ correspond to that sparse observation which was assigned to $(U_k,V_k)$. We use Theorem~\ref{thm:hansen_special_case} and Corollary~\ref{corollary:hansen} thanks to the assumptions \ref{assumption:B7} -- \ref{assumption:B9}. Moreover, we verify
\begin{align*}
\Ez{ Z^Q_k \middle| U_k = t, V_k = t' } &= \Ez{ \beta(s_{nm},s_{nm'}) \left( X_n(t,s_{nm}) + \varepsilon_{nm} \right)\left( X_n(t,s_{nm'}) + \varepsilon_{nm'} ) \right) \middle| U_k = t, V_k = t' } \\
&= \Ez{ \beta(s_{nm},s_{nm'}) a(t,t')b(s_{nm},s_{nm'}) }\\
&= a(t,t') \iint \beta(s,s') b(s,s') f_s(s)f_s(s') dsds',
\end{align*}
where $f_s(s) = \int f_{(t,s)}(t,s)dt$ is the marginal density of the random position $s_{11}$.
Likewise, denote $f_t(t) = \int f_{(t,s)}(t,s)ds$ is the marginal density of the random position $t_{11}$.
Then
$$
    Q_{00}^{\tilde{\alpha}}(t,t') = a(t,t') f_t(t)f_t(t') \iint \beta(s,s') b(s,s')  f_s(s)f_s(s') dsds' + \bigop{ \sqrt{\frac{\log N}{N h_a^2}} + h_a^2 },
$$
uniformly in $(t,t')\in[0,1]^2$ as $N\to\infty$ and where $c_r = \int x^r K(x) dx, \,r\in\mathbb{N}$. 

Similarly to the analysis above we assess the asymptotics of $S_{pq}^{\tilde{\alpha}}$. Instead of the definition in \eqref{eq:proof_smoother_A_lemma_asymptotics_Q_def_Z} we set here $Z^S_k = \left[ \beta(s_{nm},s_{nm'}) \right]^2$ and calculate
$$ \Ez{ Z^S_k \middle| U_k = t, V_k = t' } = \iint  [\beta(s,s')]^2 f_s(s)f_s(s') dsds'. $$
Hence
$$ S_{pq}^{\tilde{\alpha}}(t,t') = c_p c_q f_t(t)f_t(t') \iint  [\beta(s,s')]^2 f_s(s)f_s(s') dsds' + \bigop{ \sqrt{\frac{\log N}{N h_a^2}} + h_a^2 }$$
uniformly in $(t,t')\in[0,1]^2$ as $N\to\infty$.
By the assumptions on the kernel $\beta(\cdot,\cdot)$, the uniform Slutsky theorem, the formula \eqref{eq:proof_smoother_A_lemma_asymptotics_formula_smoother}, and the fact that $h_a\asymp h$ (assumption \ref{assumption:B9}):
$$ \tilde{\alpha}(t,t') =  a(t,t')
\frac{ \iint \beta(s,s') b(s,s')  f_s(s)f_s(s') dsds' }{ \iint \left[\beta(s,s')\right]^2  f_s(s)f_s(s') dsds' } +  \bigop{ \sqrt{\frac{\log N}{N h^2}} + h^2 }$$
uniformly in $(t,t')\in[0,1]^2$ as $N\to\infty$.

It remains to comment on the difference $\tilde{\alpha}(t,t')$ and $\widehat{\alpha}(t,t')$, i.e. when the empirical mean $\widehat{\mu}(\cdot,\cdot)$ is supplied into the raw covariances $G_{nmm'}$. Since
\begin{align*}
G_{nmm'} &= \tilde{G}_{nmm'}\\
&+\left( \mu(t_{nm},s_{nm}) - \widehat\mu(t_{nm},s_{nm}) \right)\left( Y_{nm'} - \widehat{\mu}(t_{nm'},s_{nm'}) \right)\\
&+\left( \mu(t_{nm'},s_{nm'}) - \widehat\mu(t_{nm'},s_{nm'}) \right)\left( Y_{nm} - \widehat{\mu}(t_{nm},s_{nm}) \right)\\
&+\left( \mu(t_{nm},s_{nm}) - \widehat\mu(t_{nm},s_{nm}) \right)\left( \mu(t_{nm'},s_{nm'}) - \widehat\mu(t_{nm'},s_{nm'}) \right)
\end{align*}
we conclude by Proposition~\ref{prop:estim_mu} that
\begin{equation}\label{eq:proofs_equivalence_G_G_tilde}
G_{nmm'} = \tilde{G}_{nmm'} + \bigop{ \sqrt{\frac{\log N}{N h^2}} + h^2 }    
\end{equation}
uniformly across all $n,m,m'$ as $N\to\infty$.
Therefore the claim \eqref{eq:proof_smoother_A_lemma_asymptotics} follows.
\end{proof}

\begin{corollary}\label{corollary:proof_smoother_A}
Assume the conditions \ref{assumption:A1}, \ref{assumption:B1} -- \ref{assumption:B9} and
consider a random kernel $\widehat{\beta}(\cdot,\cdot)$ such that
\begin{equation}\label{eq:corollary_smoother_A_beta_hat}
    \widehat{\beta}(s,s') = \beta(s,s') + \bigop{ \sqrt{\frac{\log N}{N h^2}} + h^2 }
\end{equation}
uniformly in $(s,s')\in[0,1]^2$ as $N\to\infty$, where $\beta(s,s'),\,s,s'\in[0,1]$ is a deterministic twice continuously differentiable kernel such that $\iint [\beta(s,s')]^2 dsds'>0$.
Then the smoother $\widehat{\alpha}(t,t'),\,t,t'\in[0,1]$, obtained by smoothing the set
\[
\left\{ \left(t_{nm},t_{nm'},\frac{G_{nmm'}}{\widehat{\beta}(s_{nm},s_{nm'})}\right) \; \Bigg| \; m,m'=1,\ldots,M_n, \, m \neq m',\; n=1\ldots,N \right\}
\]
using weights $\{ \widehat{\beta}^2(s_{nm},s_{nm'}) \}$
admits the same asymptotics as in the previous lemma:
\[
\widehat{\alpha}(t,t') =  a(t,t')
\frac{ \iint \beta(s,s') b(s,s')  f_s(s)f_s(s') dsds' }{ \iint \left[\beta(s,s')\right]^2  f_s(s)f_s(s') dsds' } +  \bigop{ \sqrt{\frac{\log N}{N h^2}} + h^2 }
\]
uniformly in $(t,t')\in[0,1]^2$ as $N\to\infty$.
\end{corollary}
\begin{proof}
The proof of this corollary follows the same lines as the proof of Lemma~\ref{lemma:proof_smoother_A_lemma_asymptotics}. We define 
\begin{align*}
S_{pq}^{\widehat{\alpha}}(t,t') &=
\frac{1}{\sum_{n=1}^N M_n(M_n-1)}
\sum_{n=1}^N \sum_{ \substack{m,m'=1 \\ m\neq m'} }^{M_n}
\left( \frac{t-t_{nm}}{h_{a}} \right)^p
\left( \frac{t'-t_{nm'}}{h_{a}} \right)^q\cdot\\
&\qquad\cdot
\frac{1}{h_a^2}
K\left( \frac{t-t_{nm}}{h_a} \right)
K\left( \frac{t'-t_{nm'}}{h_a} \right)
\left[\widehat{\beta}(s_{nm},s_{nm'}) \right]^2,
\qquad 0\leq p+q\leq 2,\\
Q_{pq}^{\widehat{\alpha}}(t,t') &=
\frac{1}{\sum_{n=1}^N M_n(M_n-1)}
\sum_{n=1}^N \sum_{ \substack{m,m'=1 \\ m\neq m'} }^{M_n}
\left( \frac{t-t_{nm}}{h_{a}} \right)^p
\left( \frac{t'-t_{nm'}}{h_{a}} \right)^q\cdot\\
&\qquad\cdot
\frac{1}{h_a^2}
K\left( \frac{t-t_{nm}}{h_a} \right)
K\left( \frac{t'-t_{nm'}}{h_a} \right)
\widehat{\beta}(s_{nm},s_{nm'})
\tilde{G}_{nmm'},
\qquad 0\leq p+q\leq 1,
\end{align*}
as analogues of \eqref{eq:proof_smoother_A_lemma_asymptotics_formula_Spq} and \eqref{eq:proof_smoother_A_lemma_asymptotics_formula_Qpq}.
Thanks to the assumption \eqref{eq:corollary_smoother_A_beta_hat}, the difference in asymptotically negligible
\begin{align*}
    S_{pq}^{\widehat{\alpha}}(t,t') &= S_{pq}^{\tilde{\alpha}}(t,t') + \bigop{ \sqrt{\frac{\log N}{N h^2}} + h^2 },\\
    Q_{pq}^{\widehat{\alpha}}(t,t') &= Q_{pq}^{\tilde{\alpha}}(t,t') + \bigop{ \sqrt{\frac{\log N}{N h^2}} + h^2 },
\end{align*}
uniformly in $(t,t')\in[0,1]^2$ as $N\to\infty$.
The rest of the proof follows from the proof of Lemma~\ref{lemma:proof_smoother_A_lemma_asymptotics}.
\end{proof}

We are now ready to prove our main result.
\begin{proof}[Proof of Theorem~\ref{thm:estim_a_b}]
The proof is now quite a simple application of Lemma~\ref{lemma:proof_smoother_A_lemma_asymptotics} and Corollary~\ref{corollary:proof_smoother_A}. First note that even though these results are formulated for the estimation of the covariance kernel $a(\cdot,\cdot)$, they can be likewise applied for the estimation of $b(\cdot,\cdot)$ due to their symmetry in the separable model \eqref{eq:separability}.

The estimator $\widehat{a}_0(\cdot,\cdot)$ is realised by smoothing the raw covariances $G_{nmm'}$ without any weights, thus corresponding to the initial guess $\beta(s,s')\equiv 1,\,s,s'\in[0,1]$. Therefore its asymptotic behaviour is by Lemma~\ref{lemma:proof_smoother_A_lemma_asymptotics}:
$$ \widehat{a}_0(t,t') = \Theta a(t,t')  + \bigop{  \sqrt{\frac{\log N}{N h^2}} + h^2 }$$
uniformly in $(t,t')\in[0,1]^2$ as $N\to\infty$ where $\Theta$ is defined in \eqref{eq:assumption:B10_equation}.

Now, applying  Corollary~\ref{corollary:proof_smoother_A} three times and by the assumption~\ref{assumption:B10} we obtain
\begin{align*}
    \widehat{b}_0(s,s') &=\frac{1}{\Theta}  b(s,s') + \bigop{  \sqrt{\frac{\log N}{N h^2}} + h^2 },\\
    \widehat{a}(t,t')   &= \Theta  a(t,t')+ \bigop{  \sqrt{\frac{\log N}{N h^2}} + h^2 },\\
    \widehat{a}(s,s')   &=\frac{1}{\Theta} b(s,s')  + \bigop{  \sqrt{\frac{\log N}{N h^2}} + h^2 },
\end{align*}
uniformly in $(t,t')\in[0,1]^2$ or $(s,s')\in[0,1]^2$, as $N\to\infty$.
\end{proof}

\begin{proof}[Proof of Corollary~\ref{corollary:estimation_complete_covariance_structure}]
This corollary follows directly by applying Theorem~\ref{thm:estim_a_b} onto the right hand side of:
$$ \left| \widehat{a}(t,t')\widehat{b}(s,s') - a(t,t')b(s,s') \right|
\leq
\left| \widehat{a}(t,t') - \Theta \widehat{a}(t,t') \right|
\left| \widehat{b}(s,s') \right|
+ \left| a(t,t') \right| \Theta
\left| \widehat{b}(s,s') - \frac{1}{\Theta} b(s,s') \right|.
$$
\end{proof}

%%%%%%%%%%%%%%%%%%%%%%%%%%%%%%%%%%%%%%%%%%%%%%%%%%%%%%%%%%%%%%%%%%%%%%%%%%%%%%%%%%%%%%%%%%%%%%%%%%%%%%%%%%%%%%%%%%%%%%%%%%%%%%%%%%%%%%%%%%%%%%%%%%%%%%%%%%%%%%%%%%%%%%%%%%%%%%%
%%%%%%%%%%%%%%%%%%%%%%%%%%%%%%%%%%%%%%%%%%%%%%%%%%%%%%%%%%%%%%%%%%%%%%%%%%%%%%%%%%%%%%%%%%%%%%%%%%%%%%%%%%%%%%%%%%%%%%%%%%%%%%%%%%%%%%%%%%%%%%%%%%%%%%%%%%%%%%%%%%%%%%%%%%%%%%%

\subsection{Proof of Proposition~\ref{prop:estim_sigma}}

The noise level estimator asymptotic behaviour is treated analogously to previous estimators of the mean surface $\mu(\cdot,\cdot)$ and the covariance kernels $a(\cdot,\cdot)$ and $b(\cdot,\cdot)$.

The estimator $\widehat{V}(t,s)$ is formed by smoothing the raw covariances $G_{nmm}$ agains  $(t_{nm},s_{nm})$. for $m=1,\dots,M_n,\,n=1,\dots,N$. Therefore we form the sequence of vectors $\{(U_k,V_k,Z_k)\}_{k=1}^\infty$ by putting $\{t_{nm}\}$, $\{s_{nm}\}$ and $\{\tilde{G}_{nmm}\}$ (defined in \eqref{eq:proofs_raw_covariance_tilde}) in order such that
\begin{align*}
    \{U_1,U_2,\dots \} &= \{ t_{11}, t_{12}, \dots, t_{1m_1}, t_{21}, \dots, t_{2m_2}, t_{31}, \dots \}, \\
    \{V_1,V_2,\dots \} &= \{ s_{11}, s_{12}, \dots, s_{1m_1}, s_{21}, \dots, s_{2m_2}, s_{31}, \dots \}, \\
    \{Z_1,Z_2,\dots \} &= \{ \tilde{G}_{111}, \tilde{G}_{122}, \dots, \tilde{G}_{1m_1m_1}, \tilde{G}_{211}, \dots, \tilde{G}_{2m_2m_2}, \tilde{G}_{311}, \dots \}.
\end{align*}
By following the steps of the proof of Lemma~\ref{lemma:smoother_mu_Qpq_asymptotics} or Lemma~\ref{lemma:proof_smoother_A_lemma_asymptotics}. 
Verifying $ \Ez{ Z_1 \middle| U_1=t,V_1=s } = a(t,t)b(s,s) + \sigma^2 $ for $t,s\in[0,1]$, the asymptotic equivalence \eqref{eq:proofs_equivalence_G_G_tilde}, and the assumption \ref{assumption:B11} implies
$$
\widehat{V}(t,s) = a(t,t)b(s,s) + \sigma^2 + \bigop{ \sqrt{\frac{\log N}{N h^2}} + h^2  }
$$
uniformly in $(t,s)\in[0,1]^2$ as $N\to\infty$.

This fact, together with Corollary~\ref{corollary:estimation_complete_covariance_structure} reduced to $(t,s,t',s')=(t,s,t,s)$ implies the statement of Proposition~\ref{prop:estim_sigma}.

%%%%%%%%%%%%%%%%%%%%%%%%%%%%%%%%%%%%%%%%%%%%%%%%%%%%%%%%%%%%%%%%%%%%%%%%%%%%%%%%%%%%%%%%%%%%%%%%%%%%%%%%%%%%%%%%%%%%%%%%%%%%%%%%%%%%%%%%%%%%%%%%%%%%%%%%%%%%%%%%%%%%%%%%%%%%%%%
%%%%%%%%%%%%%%%%%%%%%%%%%%%%%%%%%%%%%%%%%%%%%%%%%%%%%%%%%%%%%%%%%%%%%%%%%%%%%%%%%%%%%%%%%%%%%%%%%%%%%%%%%%%%%%%%%%%%%%%%%%%%%%%%%%%%%%%%%%%%%%%%%%%%%%%%%%%%%%%%%%%%%%%%%%%%%%%

\subsection{Proof of Theorem \ref{thm:prediction_bands}}

By Proposition~\ref{prop:estim_mu} and Theorem~\ref{thm:estim_a_b}, the model components $\mu,a,b,\sigma^2$ are estimated consistently.
Moreover, consider all the following statements conditionally on $ \mathbb{Y}^{new}$.
Consequently,
$$ \widehat{\var}( \mathbb{Y}^{new}) \stackrel{def}{=} \left( \widehat{a}(t^{new}_m,t^{new}_{m'}) \widehat{b}(s^{new}_m,s^{new}_{m'}) + \widehat{\sigma}^2 \mathds{1}_{[m=m']}  \right)_{m,m'=1}^{M^{new}} \stackrel{\Prob}{\to} \var( \mathbb{Y}^{new}),\qquad\text{as}\quad N\to\infty, $$
in the matrix space $\mathbb{R}^{M^{new}\times M^{new}}$. Due to continuity of the matrix inversion and the fact that $\var( \mathbb{Y}^{new})$ is positive definite,
$$ \left[\widehat{\var}( \mathbb{Y}^{new})\right]^{-1} \stackrel{\Prob}{\to}
\left[ \var( \mathbb{Y}^{new})\right]^{-1}
\qquad\text{as}\quad N\to\infty.$$
Moreover
$$
\widehat{\cov}( X^{new}(t,s), \mathbb{Y}^{new} ) \stackrel{def}{=} \left( \widehat{a}(t,t^{new}_m) \widehat{b}(s,s^{new}_m) \right)_{m=1}^{M^{new}}
= \cov( X^{new}(t,s), \mathbb{Y}^{new} ) + \smallop{1},\qquad\text{as}\quad N\to\infty,
$$
in the supremum norm over $(t,s)\in[0,1]^2$.
Therefore, together with the consistency of $\widehat{\mu}$ in the supremum norm, we conclude the statement \eqref{eq:thm:prediction_bands_statement_1}.

Assuming \ref{assumption:A2}, we conclude by the similar steps as above that
\begin{equation}\label{eq:proof_prediction_convergence_of_cov}
\sup_{(t,s,t',s')\in[0,1]^4} \left|
\widehat{\cov}\left( X^{new}(t,s),X^{new}(t',s') \vert \mathbb{Y}^{new} \right)
-
\cov \left( X^{new}(t,s),X^{new}(t',s') \vert \mathbb{Y}^{new} \right)
\right| = \smallop{1},    
\end{equation}
as $N\to\infty$.

Fixing $(t,s)\in[0,1]^2$ we have the conditional distribution given $\mathbb{Y}^{new}$
$$ \frac{ \Pi( X^{new}(t,s) \vert \mathbb{Y}^{new} ) - X^{new}(t,s)}{ \var \left( X^{new}(t,s)\vert \mathbb{Y}^{new} \right)} \sim
N\left( 0, 1 \right)
$$
where the denominator is positive for all $t,s\in[0,1]$.
Therefore
$$ \Prob\left( 
\left| \Pi(X^{new}(t,s) \vert \mathbb{Y}^{new} ) - X^{new}(t,s) \right|  
\leq u_{1-\alpha}
\sqrt{ \var\left( X^{new}(t,s) \vert \mathbb{Y}^{new} \right) }
\,\middle|\, \mathbb{Y}^{new} \right)
 = 1-\alpha.$$
Now, since
$$
\frac{ \widehat{\Pi}( X^{new}(t,s) \vert \mathbb{Y}^{new} ) - X^{new}(t,s) }{
\sqrt{ \widehat{\var} \left( X^{new}(t,s)\vert \mathbb{Y}^{new} \right) }
} \stackrel{d}{\to} N(0,1), \qquad\text{as}\quad N\to\infty.
$$
where $d$ denotes the convergence in distribution and therefore
$$ \lim_{N\to\infty} \Prob\left( 
\left| \widehat\Pi(X^{new}(t,s) \vert \mathbb{Y}^{new} ) - X^{new}(t,s) \right|  
\leq u_{1-\alpha}
\sqrt{ \widehat\var\left( X^{new}(t,s) \vert \mathbb{Y}^{new} \right) }
\,\middle|\, \mathbb{Y}^{new} \right)
 = 1-\alpha.$$

It remains to justify the asymptotic coverage of the simultaneous confidence band.
By the constriction of the simultaneous confidence bands \`a la \cite{degras2011simultaneous}, reviewed in Section~\ref{sec:prediction}, we have
$$ \Prob\left(
\sup_{(t,s)\in[0,1]^2}
\left| \Pi(X^{new}(t,s) \vert \mathbb{Y}^{new} ) - X^{new}(t,s) \right|  
\leq z_{1-\alpha}
\sqrt{ \var\left( X^{new}(t,s)\vert \mathbb{Y}^{new} \right) }
\,\middle|\, \mathbb{Y}^{new} \right)
= 1-\alpha$$
where the quantile $z_{1-\alpha}$ is calculated from the law of $W = \sup_{(t,s)\in[0,1]^2} |Z(t,s)|$ where the true (non-estimated) correlations are used: $\cov( Z(t,s), Z(t',s') ) = \corr\left( X^{new}(t,s),X^{new}(t',s') \vert \mathbb{Y}^{new} \right)$ with $t,t',s,s'\in[0,1]$. Recall that we denote the empirical analogue of this law as $\widehat{W}$ already defined in \eqref{eq:prediction_law_hat_W}.

In other words
$$
\sup_{(t,s)\in[0,1]^2} \left|
\frac{ \Pi( X^{new}(t,s) \vert \mathbb{Y}^{new} ) - X^{new}(t,s) }{
\sqrt{ \var \left( X^{new}(t,s)\vert \mathbb{Y}^{new} \right) }
} \right|
\sim W, \qquad\text{conditionally on}\,\mathbb{Y}^{new},
$$
and therefore
$$
\sup_{(t,s)\in[0,1]^2} \left|
\frac{ \widehat{\Pi}( X^{new}(t,s) \vert \mathbb{Y}^{new} ) - X^{new}(t,s) }{
\sqrt{ \widehat{\var} \left( X^{new}(t,s)\vert \mathbb{Y}^{new} \right) }
} \right|
\stackrel{d}{\to} W, \qquad\text{as}\quad N\to\infty,\quad\text{conditionally on}\,\mathbb{Y}^{new}.    
$$
Now, if $c_n(\cdot,\cdot,\cdot,\cdot) \to c(\cdot,\cdot,\cdot,\cdot)$ uniformly (cf. \eqref{eq:proof_prediction_convergence_of_cov}), then $N(0,c_n)\stackrel{d}{\to} N(0,c) $. Therefore $\widehat{W}\stackrel{d}{\to} W$ and thus $\widehat{z}_{1-\alpha}\to z_{1-\alpha}$ where $\widehat{z}_{1-\alpha}$ and $z_{1-\alpha}$ are the quantiles calculated from the law of $\widehat{W}$ and $W$ respectively. 
We conclude the proof by observing
\begin{multline*}
\Prob\left(
\sup_{(t,s)\in[0,1]^2}
\frac{\left| \widehat{\Pi}(X^{new}(t,s) \vert \mathbb{Y}^{new} ) - X^{new}(t,s) \right| }{\sqrt{ \widehat{\var}\left( X^{new}(t,s)\vert \mathbb{Y}^{new} \right) }}
\leq \widehat{z}_{1-\alpha}
\,\middle|\, \mathbb{Y}^{new} \right)
\\=
\Prob\left(
\sup_{(t,s)\in[0,1]^2}
\frac{\left| \widehat{\Pi}(X^{new}(t,s) \vert \mathbb{Y}^{new} ) - X^{new}(t,s) \right| }{\sqrt{ \widehat{\var}\left( X^{new}(t,s)\vert \mathbb{Y}^{new} \right) }}
\frac{z_{1-\alpha}}{\widehat{z}_{1-\alpha}}
\leq z_{1-\alpha}
\,\middle|\, \mathbb{Y}^{new} \right)
\to 1-\alpha,\qquad\text{as}\quad N\to\infty.
\end{multline*}

\bibliographystyle{imsart-nameyear}
\bibliography{biblio}

\end{document}